\documentclass[11pt]{article}
\usepackage[english]{babel}
\usepackage[colorlinks,linkcolor=blue,citecolor=blue,urlcolor=blue]{hyperref}
\usepackage{amsmath,amssymb,graphicx,stmaryrd,cleveref,latexsym,fullpage}
\usepackage{amsopn,amsthm,upref}
\usepackage{microtype}

\usepackage{cite}
\usepackage{url}
\usepackage[caption=false]{subfig}

\usepackage{tikz}
\usetikzlibrary{cd}
\usetikzlibrary{decorations.pathreplacing}
\usepackage{ifthen}

\tikzset{vertex/.style={ draw , circle , fill , inner sep=0em , minimum size=0.3em}}
\tikzset{empty/.style={inner sep=0em, outer sep=0em, minimum size=0em}}

\newcommand{\Iff}{\textbf{if\textcompwordmark f} }

\newcommand{\nin}{\not\in}

\newcommand{\nocomma}{}
\newcommand{\nospace}{}
\newcommand{\tmem}[1]{{\em #1\/}}
\newcommand{\tmmathbf}[1]{\ensuremath{\boldsymbol{#1}}}
\newcommand{\ort}{\tmmathbf{o}}
\newcommand{\tmop}[1]{\ensuremath{\operatorname{#1}}}
\newcommand{\tmscript}[1]{\text{\scriptsize{$#1$}}}

\newcommand{\absbr}[1]{\left|#1\right|}
\newcommand{\rbr}[1]{\left(#1\right)}

\newcommand{\floorbr}[1]{\left\lfloor #1\right\rfloor}
\newcommand{\ceilbr}[1]{\left\lceil #1\right\rceil}

\newcommand{\Nv}{N_{\uparrow}}

\newcommand{\dLM}{d_{\mathrm{LM}}}

\let\le\leqslant
\let\ge\geqslant
\let\phi\varphi

\newcommand{\eps}{\varepsilon}
\newcommand{\der}{\mathbf{D}}

\newcommand{\tanner}{\mathbf{T}}
\newcommand{\fT}{\mathfrak{T}}
\newcommand{\card}{\operatorname{card}}

\newcommand{\inc}{\mathrm{inc}}

\newcommand{\adj}{\leftrightarrow}
\newcommand{\Mod}[1]{\ (\mathrm{mod}\ #1)}

\newcommand{\Ev}{E_{\uparrow}}
\newcommand{\Eh}{E_{\rightarrow}}

\theoremstyle{plain}
\newtheorem{lemma}{Lemma}
\newtheorem{proposition}{Proposition}
\newtheorem{theorem}{Theorem}
\newtheorem{corollary}{Corollary}

\theoremstyle{definition}

\newtheorem*{definition}{Definition}

\theoremstyle{remark}
\newtheorem{remark}{Remark}
\newtheorem{example}{Example}
\newtheorem*{remark*}{Remark}

\newcommand{\RR}{\mathbb{R}}
\newcommand{\FF}{\mathbb{F}}
\newcommand{\ZZ}{\mathbb{Z}}
\newcommand{\NN}{\mathbb{N}}

\newcommand{\abs}[1]{\left\lvert #1 \right\rvert}

\newcommand{\HX}{H_X}
\newcommand{\HZ}{H_Z}
\newcommand{\CX}{\mathcal{C}_X}
\newcommand{\CZ}{\mathcal{C}_Z}
\newcommand{\dX}{d_X}
\newcommand{\dZ}{d_Z}
\newcommand{\wt}{\mathbf{wt}}

\newcommand{\T}{*}
\newcommand{\id}{\mathrm{id}}

\newcommand{\zm}{0}

\newcommand{\cA}{\mathcal{A}}
\newcommand{\cB}{\mathcal{B}}
\newcommand{\cC}{\mathcal{C}}
\newcommand{\cN}{\mathcal{N}}

\newcommand{\hG}{\hat{\Gamma}}
\newcommand{\cQ}{\mathcal{Q}}

\newcommand{\cT}{\ensuremath{\mathcal{T}}}

\newcommand{\cM}{\ensuremath{\mathcal{M}}}

\newcommand{\cF}{\mathcal{F}}

\newcommand{\HP}{\mathrm{HP}}
\newcommand{\LP}{\mathrm{LP}}

\newcommand{\bC}{\mathbf{C}}

\newcommand{\cI}{\mathcal{I}}

\newcommand{\cay}{\mathrm{Cay}}

\newcommand{\Z}{\mathbb{Z}}
\newcommand{\F}{\mathbb{F}}

\DeclareMathOperator{\rk}{rk}
\DeclareMathOperator{\im}{im}
\DeclareMathOperator{\supp}{supp}

\DeclareMathOperator*{\argmin}{arg\,min}

\begin{document}

\title{Asymptotically Good Quantum and Locally Testable Classical LDPC Codes}
\author{Pavel~Panteleev and Gleb~Kalachev\thanks{Pavel~Panteleev and Gleb~Kalachev are with the Faculty of Mechanics and Mathematics, Moscow State University, Moscow, Russia.% 
}}

\setcounter{page}{1}
\maketitle

\begin{abstract}
  We study classical and quantum LDPC codes of constant rate obtained by the lifted product construction over non-abelian groups. We show that the obtained families of quantum LDPC codes are asymptotically good, which proves the~qLDPC conjecture. Moreover, we show that the produced classical LDPC codes are also asymptotically good and locally testable with constant query and soundness parameters, which proves a~well-known conjecture in the field of locally testable codes.
  \end{abstract}

\section*{Introduction}

Classical low-density parity-check (LDPC) codes \cite{gallager1963}, as well as their quantum counterparts~\cite{qldpc}, have many important applications in theory and practice. These codes are represented by sparse parity-check matrices,  where the term \emph{sparse} usually means that the corresponding Tanner graphs are of  bounded degree. Besides numerous applications in data storage and transmission systems, such codes are often used to construct classical and quantum locally testable codes~\cite{Kaufman:2007,Aharonov:2015,Eldar:2017a}, where the sparseness of a~code ensures the~constant-query property, also known as the constant locality. Informally speaking, a~classical locally testable code (LTC) is a~code that comes with an~efficient non-deterministic procedure that allows to test with high probability whether a~given sequence is close to some codeword by looking at a~very small, usually constant, number of randomly chosen bits from this sequence. There are several ways how one can formally define LTCs~\cite{Goldreich:2010}. In this paper, we adopt a~very simple combinatorial definition (see~\cite[Definition~11]{Leverrier:2021a}) that implies a~rather strong form of local testability. According to this definition, a~linear code~$\mathcal{C} \subseteq \mathbb{F}_q^n$ is called {\tmem{$(\omega,s)$-locally testable}} if it has a~parity-check matrix $H\in \F_q^{m\times n}$ with rows of weight at most $\omega$ such that for any vector $x \in \mathbb{F}_q^n$ we have
\[ \frac{1}{m} | H x | \geqslant \frac{s}{n} d(x, \mathcal{C}), \]
where $d(x, \mathcal{C}) := \min_{c\in\cC} d(x,c)$, and we denote by $d(\cdot,\cdot)$ and $|\cdot|$  the Hamming distance and the Hamming weigh.
The parameters $\omega$ and $s$ are positive real numbers called the~{\tmem{locality}} and
{\tmem{soundness}}, respectively. As we already mentioned above, this definition implies
a~strong form of local testability. Indeed, if our test procedure picks uniformly at random a~row from $H$ and finds the corresponding syndrome component, then the~probability of rejection
$\tmop{rej}_{H} (x) = \frac{1}{m} | H x |$ grows
{\tmem{at least}} linearly with the~\emph{normalized minimum distance} $\delta (x,
\mathcal{C}) := \frac{1}{n} d (x, \mathcal{C})$ from the~tested vector $x
\in \mathbb{F}_q^n$ to the~code~$\mathcal{C}$. In~fact, for any~family of LDPC codes with $m = \Theta (n)$ where the the weights of rows and columns in~$H$ are bounded from above by~$\omega$ (such codes are called \emph{$\omega$-limited}), it follows that $\frac{1}{m} | H x |$ can not grow {\tmem{more}} than linearly with
$\delta (x, \mathcal{C})$ since for every parity-check matrix~$H$ we get $| Hx | \leqslant \omega \cdot d (x, \mathcal{C})$.

In the case of quantum locally testable codes (qLTCs) introduced in~\cite{Aharonov:2015}, one can give a~similar to the above definition if a~sparse parity-check matrix~$H$ is replaced by a~local Hamiltonian~$\mathcal{H}$ defining the quantum code. However, for a~quantum CSS code $\cQ$ (see~\cite{CSS:1996, CSS2:1996}), obtained from a~pair of classical codes $\cC_X$ and~$\cC_Z$, it is possible~\cite{Aharonov:2015, Leverrier:2021a} to infer the local testability of $\cQ$ from the local testability of $\cC_X$ and~$\cC_Z$. Let us recall that a~quantum CSS code $\cQ$ of~\emph{dimension~$k$} is defined by a~pair of classical linear
codes $\CX, \CZ \subseteq \mathbb{F}_q^n$ such that $\CZ^{\perp} \subseteq
\CX$, and $k = \dim \CX / \CZ^{\perp}$. Its~\emph{minimum distance $d$} is defined as
$\min (\dX, \dZ)$, where $\dX$ and $\dZ$ are the minimal Hamming weights of
the~vectors from $\CX \setminus \CZ^{\perp}$ and $\CZ \setminus \CX^{\perp}$,
respectively. In this case, we often say that $\cQ$ is an~$\llbracket n, k, d
\rrbracket_q$~code. The codes $\CX$, $\CZ$ are usually represented respectively by parity-check
matrices~$\HX$, $\HZ$, and the condition
$\CZ^{\perp} \subseteq \CX$ is equivalent to $\HX \HZ^{\T} = \zm$, where $\HZ^*$ is the transpose of $\HZ$. It was shown in~\cite[Lemma~13]{Leverrier:2021a} that if a~CSS code $\cQ$ is defined by two classical $(\omega,s)$-locally testable codes with parity-check matrices $H_X$, $H_Z$, then the quantum code $\cQ$ is $(\omega, s')$-locally testable, where 
$s' := s\min\left(\frac{m_X}{m_X + m_Z}, \frac{m_Z}{m_X + m_Z}\right)$, and $m_X$ (resp. $m_Z$) is the number of rows in the matrix $H_X$ (resp. $H_Z$). 

Classical and quantum LTCs have many interesting applications in theoretical computer science since they are intimately related to a~number of important problems in complexity theory~\cite{Dinur:2007, Aharonov:2015}. 
A~major open problem is whether there are such codes of \emph{constant} locality $\omega$, \emph{constant} rate, and \emph{constant} normalized minimum distance, sometimes also known as the \emph{$c^3$-conjecture} (in the context of classical codes~\cite{Dikstein:2020}) and \emph{qLTC conjecture} (in the quantum case~\cite{Aharonov:2015}). In this respect, the situation for classical LTCs is much better than for their quantum counterparts since classical LTCs of almost constant rate have been known for a~long time~\cite{Goldreich:2002}. However, in the quantum case, even if the property of local testability is not required, it is still a~widely open problem, known as the \emph{qLDPC conjecture}~\cite{Breuckmann:2021}, to obtain an~\emph{asymptotically good} family of quantum LDPC (qLDPC) codes\footnote{Note that if one goes beyond the standard definition of a~quantum LDPC code then codes with very good parameters were already known to exist~\cite{Bacon:2017, Bohdanowicz:2019}.}, i.e., with the constant rate and normalized minimum distance. Up~until very recently~\cite{Hastings:2021:fiber, Panteleev&Kalachev:2021, Breuckmann:balanced:2021,Hastings:2021}, the best provable lower bounds on the distances of qLDPC codes were, up to polylogarithmic factors, at most of the order $\sqrt{n}$ as the number of qubits $n\to\infty$~\cite{Kitaev:2002, Freedman:2002:best-code, Tillich&Zemor:2009, Guth:2014, Evra:2020, Kaufman:2021}. 
At the same time, asymptotically good families of classical LDPC codes have been known since their introduction by Robert Gallager in the 1960s~\cite{gallager1963}. 

In the current work, we show the existence of classical LTCs of constant rate, constant locality, and constant normalized minimum distance. In particular, we prove the following theorem, which gives a~positive answer to the $c^3$-conjecture. Let us recall that a~classical linear code~$\cC\subseteq\F_q^n$ has the parameters $[n, k, d]_q$ if $k=\dim \cC$ and $d = \min_{c\in\cC\setminus\{0\}} |c|$.

\begin{theorem}\label{th:LTC}
  For every number $R \in (0, 1/2)$ and finite field $\mathbb{F}_q$ it is
  possible to find universal constants $s$ and~$\omega$ such that there exists
  an~explicit family of $(\omega, s)$-locally testable classical LDPC codes with
  the~parameters $[n, k \geqslant Rn, d = \Theta (n)]_q$ as $n \rightarrow
  \infty$.
\end{theorem}

In the quantum case, we obtained a~somewhat weaker analog of the above theorem, given below, which shows the existence of asymptotically good families of qLDPC codes, not necessarily the locally testable ones. This gives an~affirmative answer to the qLDPC conjecture.

\begin{theorem}\label{th:qLDPC}
For every number $R \in (0, 1)$ and finite field $\mathbb{F}_q$ there exists an~explicit
family of quantum LDPC codes over~$\mathbb{F}_q$
with the~parameters {$\llbracket n, k \geqslant Rn, d = \Theta (n)
\rrbracket_q$} as $n \rightarrow \infty$.
\end{theorem}

\begin{remark}
  In the case of classical codes from Theorem~\ref{th:LTC}, it is relatively easy to show that an~algorithm, similar to the bit-flipping algorithm, corrects in linear time any error of weight up to the constant fraction of the code length~$n$. As for the quantum codes from Theorem~\ref{th:qLDPC}, we conjecture that it is also possible with a~variant of the small-set-flip decoding algorithm from~\cite{Leverrier:2015} (see also~\cite{Evra:2020}). 
\end{remark}

The codes from the above two theorems are obtained using the recently introduced  lifted product construction~\cite{Panteleev&Kalachev:2021}, which can be seen as a~generalization of the (tensor) product construction for classical codes~\cite{MacWilliams:1977, Ben-Sasson:2006} and the hypergraph product construction for quantum codes~\cite{Tillich&Zemor:2014}. This product operation is a~special case of the balanced product  from~\cite{Breuckmann:balanced:2021} and best understood in terms of homological algebra\footnote{In this text, we assume that the reader is familiar with the standard notions of homological algebra such as a~(co)chain complex and the corresponding  (co)homology groups. See Appendix~\ref{sc:chain} for the relevant definitions and~\cite{Brown:1982} for a~short introduction into this subject.}. Let us briefly recall that a~chain complex is a~sequence 
\[\cdots\xrightarrow{\partial_{i+1}} \cC_{i} \xrightarrow{\partial_{i}} \cC_{i-1} \xrightarrow{\partial_{i-1}}\cdots\]
of abelian groups and morphisms called \emph{boundary maps} such that $\partial_{i}\circ \partial_{i+1} = 0$ for all $i\in\Z$. The term $\cC_i$ in a~complex~$\cC$ is called the group of \emph{$i$-chains} and the assertion $\partial_{i}\circ \partial_{i+1} = 0$ is equivalent to $\im \partial_{i+1} \subseteq \ker \partial_i$, which allows us to consider for every $i\in\Z$ the quotient group $H_i(\cC) = \ker \partial_i / \im \partial_{i+1}$ called the~\emph{$i$-th homology group} of the complex~$\cC$. The abelian groups in a~complex often come with some additional algebraic structure that makes them vector spaces over a~field~$\F$ or modules over a~ring~$R$, in which case it is further assumed that all boundary maps are linear maps. In the context of error correcting codes, we are interested in the complexes with $\tau$ non-zero terms (\emph{$\tau$-term} complexes) where each term~$\cC_i$ can be naturally identified with~$\F_q^{n_i}$ and interpreted as a~space of $n_i$ symbols over~$\F_q$ (code symbols or parity-checks of the code). In such cases, it is natural to represent an~$\tau$-term complex $\cC$ by the corresponding $\tau$-partite graph called its \emph{Tanner graph}, where the edges connect only the parts corresponding to adjacent terms $\cC_i$, $\cC_{i-1}$ and the connection is governed by $\partial_i\in \F_q^{n_{i-1}\times n_{i}}$ considered as a~biadjacency matrix if we replace each non-zero entry by 1.

Given two classical linear codes invariant under a~free action of a~group~$G$ on their index sets\footnote{A~classical linear code $\cC\subseteq \F_q^n$ is \emph{invariant} under an~action of a~group $G$ on the index set $[n]$ if for every $g\in G$ and $(c_i)_{i\in[n]}\in\cC$ it follows that $(c_{\pi_g(i)})_{i\in[n]}\in\cC$, where $\pi_g$ is the permutation corresponding to the action of $g$ on~$[n]$. If the action of $G$ is free then each orbit has $|G|$ elements, and $\cC$ can be considered as a~subspace of $R^s$, where $R = \F_qG$ is the group algebra over $\F_q$ for $G$, and $s := n/|G|$. If $G$ is a~cyclic group then such codes correspond to the class of \emph{quasi-cyclic} codes, which contains classical \emph{cyclic} codes as a~special case when $s=1$.}, we can represent them by $2$-term chain complexes $\cA\colon R^{n_a} \xrightarrow{A} R^{m_a}$ and $\cB\colon R^{n_b} \xrightarrow{B} R^{m_b}$ over the~group algebra $R=\F_qG$, where $A\in R^{m_a\times n_a}$, $B\in R^{m_b\times n_b}$ are the corresponding parity-check matrices\footnote{If $G$ is non-abelian then when we multiply a~vector over $R=\F_q G$ by the matrix $A$ (resp. $B$), we assume that we multiply by the elements from $R$ from the right (resp. from the left). See Appendix~\ref{sc:lp-mat} for more details on the definition of the lifted product in terms of the parity-check matrices.}. The \emph{lifted product over $R$} is defined as the tensor product complex\footnote{The general definition of the \emph{tensor product complex $\cA \otimes_R \cB$} over an~arbitrary ring~$R$ can be found in~\cite[p.~7]{Brown:1982}.} $\cC = \cA \otimes_R \cB$ over the ring $R$, i.e., the $3$-term complex 
\[
R^{n_a n_b} \xrightarrow{\partial_2} R^{n_a m_b} \oplus R^{m_a n_b} \xrightarrow{\partial_1} R^{m_a m_b}
\]
with the boundary map $\partial\colon \cC \to \cC$ given by the following diagram
\[\begin{tikzcd}
	{R^{n_a m_b}} & {R^{m_a m_b}} \\
	{R^{n_a n_b}} & {R^{m_a n_b}}
	\arrow["{A\otimes_R \id}"', from=2-1, to=2-2]
	\arrow["{-\id \otimes_R B}", from=2-1, to=1-1]
	\arrow["{A\otimes_R \id}", from=1-1, to=1-2]
	\arrow["{\id \otimes_R B}"', from=2-2, to=1-2]
\end{tikzcd},
\]
which means that $\partial_2 := \left[\begin{smallmatrix} A\otimes_R \id \\ -\id \otimes_R B  \end{smallmatrix}\right]$, $\partial_1 := [A\otimes_R \id,  \id \otimes_R B]$.
One can easily check that $\partial_1\circ\partial_2 = A \otimes_R B - A \otimes_R B = 0$, and $\cC$ is indeed a chain complex.
Now we can consider the classical code $\ker \partial_2$ with the parity-check matrix $\partial_2$ and the quantum CSS code $\cQ(\partial_1,\partial_2^*)$ where  $\CX := \ker \partial_1$ and $\CZ := \ker \partial_2^*$. We can naturally identify these codes with the second homology group $H_2(\cC)$ and the first homology group $H_1(\cC)$ of the complex~$\cC$, and we use them to obtain the classical codes from~Theorem~\ref{th:LTC} and the quantum ones from~Theorem~\ref{th:qLDPC}, respectively. Note that when~$G$ is a trivial group, i.e., $|G| = 1$, then $R \cong \F_q$, and one can see that $\ker \partial_2$ and $\cQ(\partial_1,\partial_2^*)$ are respectively the tensor product and the hypergraph product of the two classical codes $\ker A$ and $\ker B$. Hence the lifted product complex $\cA \otimes_R \cB$, which we also sometimes denote by $\LP(A,B)$, can be seen as a~generalization of these two constructions, where instead of individual symbols from $\F_q$ we have blocks of $|G|$ symbols represented by elements from~$\F_q G \cong \F_q^{|G|}$. In fact, lifted product can also be used with arbitrary finite-dimensional associative algebra $R$ over $\F_q$, not necessary equal to $\F_q G$.
In~the current paper, if $R = \F_q G$ we call this operation \emph{lifted product over~$G$} or \emph{$G$-lifted product} and denote the corresponding lifted product complex by $\cA \otimes_G \cB$.

The idea of the lifted product was used recently in~\cite{Panteleev&Kalachev:2021} to obtain the~first family of qLDPC codes with almost linear distance. In the follow-up paper~\cite{Breuckmann:balanced:2021}, where some of the ideas from~\cite{Panteleev&Kalachev:2021} were developed independently, a~very similar construction called \emph{balanced product} was used to get qLDPC codes of very large distances\footnote{Note that the codes from~\cite{Panteleev&Kalachev:2021} are CSS codes, while the codes from~\cite{Breuckmann:balanced:2021} are in general from a~wider class of quantum codes called \emph{subsystem codes}.}. As in the case of lifted product, the balanced product $\cA \otimes_G \cB$ of two chain complexes $\cA$ and~$\cB$ can also be considered as the tensor product complex $\cA \otimes_R \cB$ over the the group algebra $R = \F_q G$, but this time $\cA$ and $\cB$ are arbitrary (i.e., not necessary free) $R$-modules. 
As it was shown in~\cite{Breuckmann:balanced:2021}, the $G$-lifted product and the balanced product can both be viewed as instances of even more general topological idea called a~\emph{fiber bundle}, proposed as a~way to construct qLDPC codes in the breakthrough paper~\cite{Hastings:2021:fiber}, which first broke the  $n^{1/2} \mathrm{polylog}(n)$ barrier on the distance of qLDPC codes. It is also interesting to note that the codes that were actually used to get the main results in~\cite{Hastings:2021:fiber, Panteleev&Kalachev:2021, Breuckmann:balanced:2021} are equivalent to $\LP(A,b)$ where $A$ is a~sparse matrix over $R=\F_2\mathbf{C}_\ell$, and $b\in R$, where $\mathbf{C}_\ell$ is the cyclic group of order~$\ell$. This more restricted class of lifted product codes were previously studied in~\cite{Panteleev&Kalachev:2019} under the name GHP codes and shown to have surprisingly good error-correcting performance under the BP-OSD decoder.

A~very important ingredient of the constructions from~\cite{Panteleev&Kalachev:2021, Breuckmann:balanced:2021} is expander codes~\cite{Sipser:1996}, which are the Tanner codes~\cite{Tanner:1981} obtained from spectral expander graphs. The individual symbols of the expander code $\cT(\Gamma; h)$ are assigned to the edges of the corresponding  graph~$\Gamma$, and we get a~codeword precisely when for each vertex $v$ from~$\Gamma$ the symbols assigned to the edges connected to~$v$ form a~codeword of the \emph{local code} $\ker h$. In~\cite{Panteleev&Kalachev:2021} expander graphs $\Gamma$ are obtained as \emph{$G$-lifts} (i.e., regular $\abs{G}$-fold covers) of some small base graphs, where $G$ is a~very large group\footnote{In~\cite{Panteleev&Kalachev:2021} this general idea was applied to cyclic groups to obtain the main result.}. It is not hard to see that the obtained in this way expander codes are invariant under the free action of $G$, and thus they are free modules over the group algebra~$\F_q G$. Therefore such codes can be used with the $G$-lifted product to obtain a~$3$-term chain complex~$\cC$, which can also be considered as a~quantum CSS code. 

It is shown in~\cite[Example~3]{Panteleev&Kalachev:2021} that using a~$G$-lifted product of two classical codes  it is possible to obtain qLDPC codes of constant rate\footnote{A~similar observation about balanced products is also made (without a~proof) in~\cite{Breuckmann:balanced:2021}.}. In particular, if $\rho := 1 - m/n$ is the design rate of a~classical code $\ker A$ represented by the~complex $\cA := R^n \xrightarrow{A} R^m$, then the rate of the corresponding quantum code represented by $\cA \otimes_G \cA^*$ is at least $\frac{(n-m)^2}{n^2 + m^2} = \frac{\rho^2}{1 + (1-\rho)^2}$. Here $\cA^* := R^m \xrightarrow{A^*} R^n$ is the \emph{dual chain complex} for $\cA$, i.e., $A^*$ is the transpose of the parity-check matrix $A$, considered as a~matrix over $\F_q$.   Hence the rate of the quantum codes obtained from $\cA\otimes_G\cA^*$ can be arbitrary close to $1$ as $\rho\to 1$.
Moreover, some particular examples of such codes~\cite[Example~4]{Panteleev&Kalachev:2021}, indicate that they may also have very large minimal distances, close to the distances of the classical codes $\ker A$ used in the lifted product. However, if the group $G$ is abelian, then the upper bound on the minimum distance of such classical codes~\cite[Eq.~24]{Panteleev&Kalachev:2021} provides strong evidence that to obtain an~asymptotically good family of qLDPC codes by a~$G$-lifted product one has to use non-abelian groups. 

One particular construction of balanced product codes, analogous to the aforementioned $G$-lifted product $\cA \otimes_G \cA^*$, for non-abelian group $G$, was conjectured in~\cite{Breuckmann:balanced:2021} to give an~asymptotically good family of qLDPC codes. Unfortunately, our proof strategy does not work for complexes $\cA \otimes_G \cA^*$, and we can not prove this conjecture with the methods developed here. Instead, we consider similar complexes $\cA \otimes_G \cB^*$, where $\cA$ and $\cB$ are respectively the expander codes $\cT(\Gamma; h)$ and $\cT(\Gamma; h')$ defined for the \emph{same} expander graph $\Gamma$ but for \emph{different} local codes $\ker h$ and $\ker h'$. It is very simple to show by counting the number of the code symbols and parity-checks in the LTC and the qLDPC code obtained from~$\cA \otimes_G \cB^*$ that these codes have constant rate.  However, for our proof of Theorems~\ref{th:LTC} and~\ref{th:qLDPC} to work, the pair of local codes used in~$\cA \otimes_G \cB^*$ can not be arbitrary and should satisfy some special property we call \emph{product-expansion}, which is similar to the \emph{robust testability} property often used in the context of LTCs~\cite{Ben-Sasson:2006,Dinur:2006}. We prove that a~pair of random linear codes has the product-expansion property with high probability. Informally speaking, the product-expansion of the pair of local codes corresponds to the \emph{local expansion} in the complex $\cA \otimes_G \cB^*$, but to get the main result we also need the~\emph{global expansion} property of the graph~$\Gamma$, which connects the~local codes attached to its vertices.
Our main technical result (Proposition~\ref{prop:main}) shows that the general construction $\cA \otimes_G \cB^*$ can be used with arbitrary regular graphs $\Gamma$ obtained as $G$-lifts if they are sufficiently good \emph{small set expanders}\footnote{Informally this means that every sufficiently small set of vertices has a lot of outgoing edges. See Subsection~\ref{sc:graph-exp} for the relevant definitions and results.}.  We prove that spectral expander graphs and their finite covers are good small set expanders. Hence we can let the graph $\Gamma$ to be the bipartite double-cover of a~Cayley graph for some finite group~$G$. This is important for the $G$-lifted product construction since such graphs $\Gamma$ can be also represented as $G$-lifts\footnote{Note that in most cases a~Cayley graph with $w$ generators can also be viewed as a~$G$-lift of the $w$-bouquet graph~$B_w$. However, this is not true if we have an~order~$2$ generator.}. In particular, we use the  Ramanujan Cayley graphs~\cite{Lubotzky:1988,Margulis:1988}, which were also used in the original construction of the expander codes~\cite{Sipser:1996} and in the mentioned earlier conjecture from~\cite{Breuckmann:balanced:2021}.

The~main technical tool in our proof of Theorems~\ref{th:LTC} and~\ref{th:qLDPC} is the notion of a~\emph{locally minimal} (co)chain, often used in the context of high-dimensional expanders to show expansion properties in simplicial complexes~\cite{Kaufman:2014}. It is known that such expansion properties can be used to show local testability of a~classical code~\cite{Kaufman:2014a} and to give a~lower bound on the minimum distance of a~quantum code~\cite{Evra:2020}. In the current work, we~extend these ideas to a~much more general context of (co)chain complexes with local system of coefficients, which can be considered as high-dimensional analogs of the Tanner codes, similar to the ones studied in~\cite{Meshulam:2018}. Instead of graphs such generalized Tanner codes are defined on high-dimensional complexes. Since the $G$-lifted product is defined for arbitrary complexes, it can naturally be applied to graphs, viewed as $1$-dimensional complexes. If we consider graphs $\Gamma$ and $\Gamma'$ as topological spaces, their $G$-lifted product (as a~topological space) can be viewed as the balanced product $\Gamma \times_G \Gamma'$ of these spaces~\cite{Breuckmann:balanced:2021}. In fact, it can be shown that the products $\Gamma \times_G \Gamma'$ are examples of a~well-known class of $2$-dimensional complexes called \emph{complete square complexes}~\cite{Wise:2007}. The defining property of a~complete square complex is that the links of all its vertices are isomorphic to a~\emph{complete} bipartite graph. Since complete bipartite graphs are perfect expanders, then, in some sense, this property is analogous to the property of other high-dimensional expanders to have links that are good expanders~\cite{Kaufman:2014}.

Using the discussed above $G$-lifted products of expander codes over non-abelian groups $G$ we show that it is possible to obtain qLDPC codes with the parameters as in Theorem~\ref{th:qLDPC}. This gives a~positive answer to the questions posed in~\cite[Conclusion]{Panteleev&Kalachev:2021} and in~\cite[Conjecture]{Breuckmann:balanced:2021} of whether respectively lifted and balanced products of classical codes can give an~asymptotically good family of qLDPC codes.
Moreover, we also show that, under some additional assumptions, if $H_X$ and $H_Z$ are the parity-check matrices of such qLDPC codes, then the classical code $\ker H_Z^\T$ is locally testable with the parameters as in Theorem~\ref{th:LTC}. 

\begin{remark*}
  After the first draft of this manuscript was published we became aware that a~result similar to our Theorem~\ref{th:LTC} for the case of binary field $\F_2$ was independently claimed in~\cite{Dinur:2021}. The $3$-term complex used in~\cite{Dinur:2021} to get the main result is equivalent to the balanced product over~$G$ of the expander codes~\cite{Breuckmann:balanced:2021,Breuckmann:2021}, defined on \emph{two different} Cayley graphs for the \emph{same} group~$G$. It is interesting to note that this construction is similar to the lifted product construction we consider in Remark~\ref{rm:LTC-constr}, where instead of the product $\cA \otimes_G \cB^*$ we propose to use the product $\cA \otimes_G \cB$ and conjecture that this way it is still possible to get asymptotically good LTCs. The diagrams for $\cA \otimes_G \cB$ and $\cA \otimes_G \cB^*$ are shown below 
  \[
  \cA \otimes_G \cB := \begin{tikzcd}
	{R^{n_a m_b}} & {R^{m_a m_b}} \\
	{R^{n_a n_b}} & {R^{m_a n_b}}
	\arrow["{A\otimes_R \id}"', from=2-1, to=2-2]
	\arrow["{-\id \otimes_R B}", from=2-1, to=1-1]
	\arrow["{A\otimes_R \id}", from=1-1, to=1-2]
	\arrow["{\id \otimes_R B}"', from=2-2, to=1-2]
\end{tikzcd},\quad
  \cA \otimes_G \cB^* :=
  \begin{tikzcd}
	{R^{n_a m_b}} & {R^{m_a m_b}} \\
	{R^{n_a n_b}} & {R^{m_a n_b}}
	\arrow["{A\otimes_R \id}"', from=2-1, to=2-2]
	\arrow["{-\id \otimes_R B^*}"', from=1-1, to=2-1]
	\arrow["{A\otimes_R \id}", from=1-1, to=1-2]
	\arrow["{\id \otimes_R B^*}", from=1-2, to=2-2]
\end{tikzcd}.
\]
  In fact, the Tanner graphs of the complexes $\cA \otimes_G \cB$ and $\cA \otimes_G \cB^*$ are isomorphic. What is different is the interpretation of the Tanner graph vertices as \emph{code symbols} and \emph{parity-checks} when we make a~code out of the complex.
  In some sense, the product $\cA \otimes_G \cB$ is better suited for LTCs since it gives classical codes of rate arbitrary close to~$1$ (please, see~Remark~\ref{rm:LTC-constr}). Hence it is an interesting open question whether the approach used in~\cite{Dinur:2021} can also succeed on our codes from Remark~\ref{rm:LTC-constr}.  At the same time, the construction $\cA \otimes_G \cB^*$, which we use to prove the main results, is much better suited for qLDPC codes since it symmetric. This symmetry allows us to prove the lower bound on the $Z$-distance of our qLDPC code in the same way as for the $X$-distance. Besides, we can get equal number of $X$-checks and $Z$-checks, which gives qLDPC codes of rates arbitrary close to~$1$.
\end{remark*}

\section{Preliminaries}

\subsection{Chain complexes}
In recent years, ideas from homological algebra found many interesting applications in the field of classical and quantum codes\cite{Kaufman:2014a, Bravyi:HMP:2014,Pryadko:2019}. A~common approach is to consider some based\footnote{The term \emph{based} means that the vector spaces of a (co)chain complex come with some distinguished bases.  If in a~vector space $V$ we fix a~basis $\tilde{V}\subseteq V$, we can identify $V$ and its dual space $V^*$ with the coordinate space $\F_q^{\dim V}$ in the standard way. This also allows us to identify linear maps between such spaces with the corresponding matrices.} (co)chain complex of finite-dimensional vector spaces over a~finite field~$\F_q$, and use it to define a~code with the desired parameters. 
For example, a~$2$-term chain complex
\[ 
\mathbb{F}_q^{n}
   \xrightarrow{\partial_1} \mathbb{F}_q^m
\]
can be identified with the classical linear code $\ker \partial_1$ defined by the \emph{parity-check matrix} $H:= \partial_1$. Here, the space $\F_q^n$ of 1-chains corresponds to the $n$~bits, while the space $\F_q^m$ of $0$-chains to the $m$~checks. At the same time, a~$3$-term chain complex
\[ 
   \cC := \rbr{\mathbb{F}_q^{m_Z}\xrightarrow{\partial_2} \mathbb{F}_q^n \xrightarrow{\partial_1} \mathbb{F}_q^{m_X}}
\]
can be identified with the quantum CSS $\llbracket n, k, d \rrbracket_q$ code $\cQ = \cQ(H_X,H_Z)$ defined by the parity-check matrices $H_X := \partial_1$ and $H_Z := \partial_2^*$, where $\partial_2^*\colon \F_q^n\to \F_q^{m_Z}$ is the transpose of the map $\partial_2\colon \F_q^{m_Z}\to\F_q^n$. 
In this case, the space $\F_q^n$ of 1-cells corresponds to the $n$ qubits, and the space $\F_q^{m_X}$ of 0-cells (resp. the space $\F_q^{m_Z}$ of 2-cells) to the \emph{$X$-checks} (resp. \emph{$Z$-checks}). The length of $\cQ$ is equal to $n=\dim \F_q^n$, while its dimension~$k$ is equal to the dimension of the first homology group $H_1(\cC) := \ker \partial_1 / \im  \partial_2 = \cC_X / \cC_Z^\perp$, where $\cC_X := \ker \partial_1$ and $\cC_Z := \ker\partial_2^*$.  The minimum distance $d = d(\cQ)$ can also be described in the language of homology groups if we consider the quotient vector space $H_1(\cC)$ as a~metric space, where the distance $d(A,B)$ between homology classes $A,B\in H_1(\cC)$ is defined as $d(A,B) := |A - B|$ using the corresponding quotient Hamming norm $|A| := \min_{a\in A} |a|$. It is easy to see that $d = \min(d(H_1(\cC), d(H_1(\cC^*))$, where
\[ 
   \cC^* := \rbr{\mathbb{F}_q^{m_X}\xrightarrow{\partial_1^*} \mathbb{F}_q^n \xrightarrow{\partial_2^*} \mathbb{F}_q^{m_Z}}
\]
is the \emph{dual chain complex} for $\cC$. The distances $d(H_1(\cC))$ and $d(H_1(\cC^*))$ are sometimes called the \emph{$1$-systolic} and \emph{$1$-cosystolic distances} of $\cC$.

\subsection{Lifted product}
In this work, we consider several new families of classical and quantum LDPC codes of constant rate based on the introduced recently lifted product construction~\cite{Panteleev&Kalachev:2021}, which generalizes many known constructions of quantum LDPC codes~\cite{qldpc,Hagiwara:2007,Tillich&Zemor:2014,Haah:2011,Kovalev:2013, Pryadko:2019}. This construction can be defined in terms of parity-check matrices (see Appendix~\ref{sc:lp-mat}) and in the abstract language of homological algebra, which we prefer in the current work. Before we proceed, let us briefly remind some standard definitions from algebra. Consider some ring~$R$. A~left~$R$-module~$M$ is called {\tmem{free}} if there exists a~set of elements $\{m_1, \ldots, m_r \}
\subseteq M$ called {\tmem{basis}} such that every $m \in M$ is uniquely
represented as:
\[ m = a_1 m_1 + \ldots + a_r m_r, \]
where $a_1, \ldots, a_r \in R$, and the parameter~$r$ is called the \emph{rank}\footnote{Note that there are some infinite non-commutative rings $R$ such that $R^m\cong R^n$ when $m\ne n$. However, all the rings we consider here are either finite or commutative, and hence have the \emph{invariant basis number} (IBN) property that implies that this never happens.} of $M$. Hence $M \cong R^r$, and if the ring $R$ is
a~field, then $M$ is just an~$r$-dimensional vector space over~$R$.
A~canonical example of a~free $R$-module of rank $r$ is the~module $RS$ of
formal $R$-linear combinations of the~elements of some set~$S$, where $\abs{S}
= r$. One can also define free right $R$-modules in a~similar way.

\begin{definition}
Suppose we have a~finite-dimensional associative algebra $R$ over $\F_q$ with some fixed basis $\tilde{R}\subseteq R$. Consider two chain complexes $\cA=\bigoplus_{i=0}^m \cA_i$ and $\cB=\bigoplus_{j=0}^n\cB_j$ over $\F_q$ such that the vector spaces $\cA_i$ and $\cB_j$ are also free $R$-modules with some distinguished bases (over~$R$) $\tilde{\cA}_R \subseteq\cA$ and $\tilde{\cB}_R \subseteq\cB$, and the boundary maps $\partial_{\cA}\colon \cA\to\cA$,  $\partial_{\cB}\colon \cB\to\cB$ are $R$-linear. If the algebra $R$ is not commutative, then we further assume that $R$ acts from the right on $\cA$ and from the left on $\cB$, i.e., $\cA$ is a~right free $R$-module, and $\cB$~is a~left free $R$-module.  The \emph{lifted product} of $\cA$ and $\cB$ \emph{over $R$} is their  \emph{tensor product complex} $\cA\otimes_R\cB$ (see~\cite[p.~7]{Brown:1982}), where for $k=0,1,\dots,m+n$ the space of $k$-chains $(\cA\otimes_R\cB)_k$ is equal to $\bigoplus_{i+j = k} \cA_i\otimes_R\cB_j$, while the  boundary map~$\partial \colon \cA\otimes_R\cB\to \cA\otimes_R\cB$ is defined for $a\in\cA_i$, $b\in\cB_j$ as\footnote{We should note that the sign $(-1)^i$ in this definition is only relevant in the case of finite fields of odd characteristic.} 
\begin{equation}\label{eq:lp-boundary}
\partial (a \otimes_R b) := \partial_{\cA} a \otimes_R b + (-1)^{i} a \otimes_R \partial_{\cB} b,
\end{equation}
and extended by linearity. 
Furthermore, we always assume that the lifted product $\cC = \cA \otimes_R \cB$ is a~\emph{based} chain complex of vector spaces over $\F_q$. By definition its distinguished basis (over~$\F_q$) is given by
\[
\tilde{\cC} := \{a \cdot r \cdot b \mid a\in \tilde{\cA}_R, b\in\tilde{\cB}_R, r\in\tilde{R}\},
\]
where we used a~short-hand notation:
\begin{equation}\label{eq:short-hand}
a \cdot r \cdot b := ar \otimes_R b = a \otimes_R rb.
\end{equation}
From the properties of the tensor product $\otimes_R$ it follows that the map $(a,r,b)\mapsto a \cdot r \cdot b$ is $\F_q$-multilinear, which means that for every $a,a'\in \cA$, $b,b'\in \cB$, and $r,r'\in R$ we have:
\begin{eqnarray*}
  (a + a') \cdot r \cdot b &=& a \cdot r \cdot b + a' \cdot r \cdot b,   \\
  a \cdot (r + r') \cdot b &=& a \cdot r \cdot b + a \cdot r' \cdot b,   \\
  a \cdot r \cdot (b + b') &=& a \cdot r \cdot b + a \cdot r \cdot b',   
\end{eqnarray*}
and for every $\lambda\in \F_q$  we get:
\begin{equation*}
    (\lambda a) \cdot r \cdot b = a \cdot (\lambda r) \cdot b = a\cdot r \cdot (\lambda b) = \lambda (a\cdot r\cdot b). 
\end{equation*}

\end{definition}

We should note that if $R=\F_q$, then the lifted product is equivalent to the product construction from~\cite{Pryadko:2019}, while if, in addition, we have $m=n=1$, then it is the same as the hypergraph product~\cite{Tillich&Zemor:2014}. Moreover, if $m=n=1$ and $R=\F_2[x]/(x^\ell - 1)$, it is essentially equivalent to the hyperbicycle codes construction~\cite{Kovalev:2013}. It is also important to note that when $m=n=1$ the complexes 
\[
\cA:=\left(\cA_1\xrightarrow{A} \cA_0\right)\text{ and }\cB:=\left(\cB_1\xrightarrow{B} \cB_0\right)
\]
are uniquely defined by the corresponding matrices $A$, $B$ over $R$. In this case, we denote the lifted product $\cA \otimes_R \cB$  as $\LP(A,B)$ and usually identify it with the corresponding CSS code. Note that this code also has a~concise description in terms of the parity-check matrices $H_X$ and $H_Z$ (see~\cite[Eq.~12]{Panteleev&Kalachev:2021} and Eq.~\ref{eq:LP} from Appendix~\ref{sc:lp-mat}).

Though the lifted product can be defined over an~arbitrary finite-dimensional associative algebra~$R$,  the most interesting case~\cite{Panteleev&Kalachev:2021,Breuckmann:balanced:2021} is when $R$ is the group algebra $\F_qG$ for some finite group~$G$. 
The~elements of $\F_q G$ are formal sums $\sum_{g \in G} \alpha_g
g$, where $\alpha_g \in \F_q$. Consider elements $a = \sum_{g \in G} \alpha_g
g$ and $b = \sum_{g \in G} \beta_g g$ from $\F_q G$. Their sum $a + b$ and
product $ab$ are defined as follows:
\[ a + b := \sum_{g \in G} (\alpha_g + \beta_g) g, \quad ab := \sum_{g \in G}
   \left( \sum_{hr = g} \alpha_h \beta_r \right) g. \]
In this case, the condition that the vector spaces $\cA$ and $\cB$ over $\F_q$ are free $\F_q G$-modules is equivalent to the condition that the group~$G$ has a~free action\footnote{A~\emph{left} (resp. \emph{right}) action of a~group~$G$ on a~set~$S$ is called \emph{free} if for every $g\in G$ when we have $gs = s$ (resp. $sg = s$) for some $s\in S$, then $g$ is the identity element of $G$. Note that the sizes of all orbits of a~free action are the same and equal to $\abs{G}$.} on the their bases over~$\F_q$ (from the right for $\cA$ and from the left for $\cB$), which is extended by linearity to $\cA$ and $\cB$. Moreover, the boundary map~$\partial$ is $\F_q G$-linear \Iff it is an~$\F_q$-linear map that commutes with the action of the group $G$. Therefore in~what follows, in tensor products over $R=\F_qG$ instead of $\otimes_R$ we write $\otimes_G$, and assume that $\tilde{R} := G$. Let $\tilde{\cA}_G = \bigsqcup_{i\in\Z} \tilde{\cA}_{G,i}$ and $\tilde{\cB}_G = \bigsqcup_{j\in\Z} \tilde{\cB}_{G,j}$ be respectively the distinguished bases (over $\F_q G$) of the the free $\F_qG$-modules $\cA = \bigoplus_{i\in\Z} \cA_i$ and $\cB = \bigoplus_{j\in\Z} \cB_j$. It is clear that the elements $a g$ (resp. $g b$), where $a\in\tilde{\cA}_{G,i}$, $g\in G$, $b\in\tilde{\cB}_{G, j}$, constitute the basis for $\cA_i$ (resp. $\cB_j$), considered as a~vector space over~$\F_q$.  Moreover, we see, using short-hand notation~(\ref{eq:short-hand}), that the distinguished basis of $\cA \otimes_G \cB$ over $\F_q$ consists of the elements $a \cdot g \cdot b$, where $a\in\tilde{\cA}_{G,i} $, $g\in G$, $b\in \tilde{\cB}_{G,j}$; $i,j\in\Z$. Furthermore, we can express the boundary operator given in equation~(\ref{eq:lp-boundary}) as follows:
\begin{equation}\label{eq:lp-boundary2}
\partial (a \cdot g \cdot b) := (\partial_{\cA} a) \cdot g \cdot b + (-1)^{i} a \cdot g \cdot (\partial_{\cB}b).
\end{equation}

We can also express the boundary operator $\partial$ as 
\[
\partial := \partial_{\cA} \otimes_G \id  + \id \otimes_G \partial_{\cB} 
\]   
if, by definition, assume that $(\partial_{\cA} \otimes_G \id) (a \cdot g \cdot b) := (\partial_{\cA} a) \cdot g \cdot b$ and $(\id \otimes_G \partial_{\cB}) (a \cdot g \cdot b) := (-1)^{i} a \cdot g \cdot (\partial_{\cB}b)$.

\begin{remark}\label{rm:left-right}
  For any chain complex $\cC$ we can consider its \emph{dual chain complex} $\cC^*$ obtained from $\cC$ if we replace the boundary map $\partial$ of $\cC$ by its transpose map $\partial^*$ (see Appendix~\ref{sc:chain}). It is not hard to see that if $\cC$ is a~left (resp. right) $G$-module, then $\cC^*$ is a~right (resp. left) $G$-module. Therefore if chain complexes $\cA$ and $\cB$ are right $G$-modules, we can consider the $G$-lifted product $\cA \otimes_G \cB^*$. In~fact, for any set $S$ with a~left action $(g, s)\mapsto g\cdot s$ (resp. a~right action $(s, g)\mapsto s\cdot g$) of a~group $G$ we can also consider the corresponding right (resp. left) action of $G$ defined as $(s, g)\mapsto g^{-1} \cdot s$ (resp. $(g, s)\mapsto s\cdot g^{-1}$). Therefore if a~group $G$ has a~right free action on a~chain complex $\cC$, then it also has the corresponding left free action on $\cC$, and vice versa. This allows us to apply $G$-lifted product $\cA\otimes_G \cB$ to two right $G$-modules $\cA$ and $\cB$, if we use the corresponding left action of $G$ on $\cB$. 
\end{remark}
\begin{remark}
Let us note that $G$-lifted product is a~special case of balanced product from~\cite{Breuckmann:balanced:2021}, where a~non-free action of the group~$G$ is also allowed. We should also emphasize that the first examples of the lifted products over $R=\F_2G$ for a~\emph{non-abelian} group~$G$ were also considered in~\cite{Breuckmann:balanced:2021}, while in~\cite{Panteleev&Kalachev:2021} all the examples were only for the abelian case. In the current work, we also give new examples of non-abelian lifted products  based on the~double-cover of a~Cayley graph, which are similar, though not equivalent, to the horizontal subsystem codes mentioned in the Conjecture from~\cite{Breuckmann:2021}.     
Generally speaking, the term \emph{$G$-lifted product}, used in the current work, may seem redundant since it is just a~special case of the balanced product. However, we think that this special case deserves its own name since the free action of $G$ implies that the obtained complex has a~much more regular structure than in the general case. In~some sense, the relation of the $G$-lifted product to the more general balanced product is similar to the relation of Cayley graphs to Schreier graphs. While the latter are more general, the former are usually much easier to describe and study.   
\end{remark}

\subsection{Expander graphs and lifts}

To produce linear maps $\phi\colon\F_q^n\to\F_q^m$ with good expansion and coexpansion properties it was proposed  in~\cite{Panteleev&Kalachev:2021, Breuckmann:balanced:2021} to use expander codes~\cite{Sipser:1996}, i.e., the Tanner codes~\cite{Tanner:1981} defined on some spectral expander graph. Before we move on, let us recall some standard definitions related to expander graphs and Tanner codes.

Let $\Gamma$ be a~graph\footnote{It may have loops and multiple edges.} with the set of vertices $V(\Gamma)$ and the set of edges $E(\Gamma)$. If vertices $v, v'\in V(\Gamma)$ are connected by an edge $e\in E(\Gamma)$, we call $v, v'$ \emph{adjacent} and denote this fact by $v \leftrightarrow v'$ or by $v\adj_e v'$ when we want to emphasize the edge~$e$. A graph~$\Gamma$ is called \emph{$d$-regular} if all its vertices have degree~$d$.  The \emph{adjacency matrix} of a~graph~$\Gamma$ with $V(\Gamma)=\{v_1,\dots,v_n\}$ is the matrix $A(\Gamma) = (a_{i j})_{n\times n}$, where $a_{i j}$ is the number of edges $e\in E(\Gamma)$ such that $v_i \adj_e v_j$.
Since $A(\Gamma)$ is a symmetric matrix, it has $n$ real-valued eigenvalues $\lambda_1 \ge \dots \ge \lambda_n$. Let~$\lambda_2(\Gamma) := \lambda_2$, and $\lambda(\Gamma) := \max(\abs{\lambda_2}, \abs{\lambda_n})$. It is obvious that $\lambda_2(\Gamma)\le \lambda(\Gamma)$. We call an~$n$-vertex~$d$-regular graph $\Gamma$ an~\emph{$(n, d, \lambda)$-expander} if $\lambda(\Gamma) \le \lambda$. The term expander here means that the graph~$\Gamma$ has a~very good connectivity, which can be quantified by its Cheeger constant.  Consider a~subset of vertices $S\subseteq V(\Gamma)$ in the graph~$\Gamma$. We call the set 
\begin{align*}
    \partial S &:= \{e \in E(\Gamma) \mid v \adj_{e}\! v'\!, v\in S, v'\notin S\}
\end{align*}
the~\emph{edge boundary}, which is the set of all edges that go outside of $S$. The \emph{Cheeger constant $h(\Gamma)$} of the graph~$\Gamma$ is defined as follows: 
\[
h(\Gamma):= \min_{\substack{
0 < \abs{S} \le \frac12\abs{V(\Gamma)} \\
S \subseteq V(\Gamma)}} 
\frac{\abs{\partial S}}{\abs{S}}.
\]
Since for $d$-regular graphs it is known~\cite[Theorem~4.11]{Hoory:2006} that $h(\Gamma)\ge \frac12(d - \lambda_2(\Gamma))$, then the smaller the value of $\lambda_2(\Gamma)$, the higher the Cheeger constant. However, the Alon-Boppana bound~\cite[Theorem~5.3]{Hoory:2006} implies that for $d$-regular graphs with $n$ vertices we have $\lambda_2(\Gamma) \ge 2\sqrt{d-1} - o_n(1)$ as $n\to\infty$. 
There are a~number of different constructions that almost attain this lower bound. In fact, it was shown in~\cite{Friedman:2003} that for any fixed $\eps > 0$, a~random $d$-regular graph with $n$ vertices has $\lambda_2(\Gamma) < 2\sqrt{d-1} + \eps$ with high probability as $n\to\infty$. A~$d$-regular graph $\Gamma$ that satisfy the condition $\lambda(\Gamma) \le 2\sqrt{d - 1}$ is called \emph{Ramanujan}\footnote{In this work we consider only non-bipartite Ramanujan graphs.}. There are a~number of explicit constructions of such graphs~\cite{Margulis:1988, Lubotzky:1988} that use Cayley graphs of some non-commutative groups (see~\cite{Davidoff:2003} for a~good survey). 

We will see later that Tanner codes with such Ramanujan graphs (or their double-covers) can be used with the lifted product construction. The obtained chain complexes, which we can also consider as CSS codes, have very interesting expansion properties, similar to the ones studied in the theory of high-dimensional expanders (HDXs)~\cite{Lubotzky:2018}. We will show later that some of the standard definitions from this theory (e.g., the local minimality of (co)chains) can be naturally extended to a~more broad context of based (co)chain complexes.

\begin{figure}
  \centering
  \begin{tikzpicture}[scale=0.7]
    \node[vertex] (v) at (-0.4,-0.4) [label=below:$v$]{};
    \node[vertex] (v') at (0.4,0.4) [label=$v'$]{} ;
    \draw (v) -- (v');
    \draw (0,0) ellipse[x radius=1.5 cm,y radius=1.7 cm, rotate=-45];    
    \node[empty]  (cap) at (-0.2,-2.2) {base graph $\Gamma$};
  \end{tikzpicture}
  \quad
  \begin{tikzpicture}[scale=0.9]
    \node[vertex] (v1)    at (-1.2,-0.4) [label=below:$\hat{v}_1$]{};
    \node[vertex] (v2)    at (-0.9,-0.4) {};
    \node[empty]  (vdots) at (-0.6,-0.4) {\footnotesize\ldots};
    \node[vertex] (vell)    at (-0.3,-0.4) [label=below:$\hat{v}_\ell$]{};

    \node[vertex] (v1')    at (-1.2+1,0.4) [label=above:$\hat{v}'_1$]{};
    \node[empty]  (vdots') at (-0.9+1,0.4) {\footnotesize\ldots};
    \node[vertex] (v2')    at (-0.6+1,0.4) {};
    \node[vertex] (vell')    at (-0.3+1,0.4) [label=above:$\hat{v}'_\ell$]{};
    \draw (v1) -- (v2');
    \draw (v2) -- (vell');
    \draw (vell) -- (v1');
    \draw (-0.1,0) ellipse[x radius=1.7 cm,y radius=2 cm, rotate=-45];
    \node[empty]  (pi) at (0.8,-0.1) {$\pi\in\mathbf{S}_\ell$};
    \node[empty]  (cap) at (-0.2,-2.2) {$\ell$-lift $\hat{\Gamma}$ of $\Gamma$};
  \end{tikzpicture}
  \caption{Lifting of the base graph $\Gamma$.}
  \label{fg:lifting}
\end{figure}
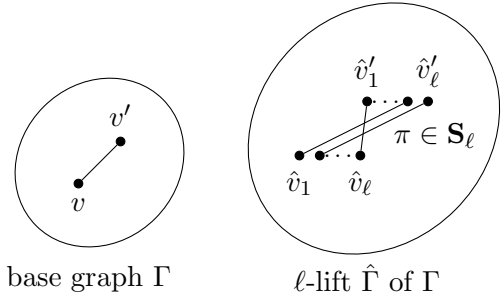

In~\cite{Panteleev&Kalachev:2021}, the graph $\hat{\Gamma}$ for the Tanner code was obtained as an~$\ell$-lift of a~small base graph $\Gamma$ using voltage assignments~\cite{Gross-Tucker:1987} with the cyclic group~$\mathbf{C}_\ell$ as the voltage group. Recall that an~\emph{$\ell$-lift} (also called an~\emph{$\ell$-fold cover}) of a~base graph\footnote{Multiple edges and loops are usually allowed in the base graph~$\Gamma$.} $\Gamma$ is a~graph~$\hat{\Gamma}$ obtained if we replace in the base graph each vertex $v\in V(\Gamma)$ with $\ell$~replicas $\hat v_{1},\dots,\hat v_{\ell}$, and replace each edge $e\in E(\Gamma)$ that connects vertices $v, v'\in V(\Gamma)$ with $\ell$~replicas~$\hat e_1,\dots,\hat e_\ell$ such that $\hat e_i$ connects in $\hat{\Gamma}$ the vertices $\hat v_i$ and ${\hat v}'_{\pi(i)}$, where ${\pi\in\mathbf{S}_\ell}$ is some permutation on the set $\{1,\dots,\ell\}$ (see Fig.~\ref{fg:lifting}). 
Note that the~permutations for different edges may be different and are usually defined~\cite{Gross-Tucker:1987} by a~voltage assignment using some group $G$, in which case we call the obtained graph a~\emph{$G$-lift} of $\Gamma$.

A~\emph{voltage assignment} for a~graph~$\Gamma$  with a~\emph{voltage group} $G$ is a~map $\gamma\colon E(\Gamma)\to G$ . Let us fix some \emph{orientation} of the edges, i.e., a~function $\ort\colon E(\Gamma)\to V(\Gamma)\times V(\Gamma)$, which tells us that the edge $e$ is oriented from $v$ to $v'$ if $\ort(e)=(v,v')$. For any voltage assignment $\gamma$, we can obtain the $G$-lift $\hat{\Gamma}$ of the base graph~$\Gamma$ called  the (\emph{left}) \emph{derived graph} for $\Gamma$ and $\gamma$, which we denote  by $\der(\Gamma; \gamma)$. To define $\hat{\Gamma} = \der(\Gamma;\gamma)$ we first let $V(\hat{\Gamma}):= V(\Gamma)\times G$, $E(\hat{\Gamma}):= E(\Gamma)\times G$, and introduce the following short-hand notations: $\hat{v}_g:= (v, g)$, $\hat{e}_g := (e, g)$, where $v\in V(\Gamma)$, $e\in E(\Gamma)$, $g\in G$. Now, if in the base graph $\Gamma$ an~edge~$e\in E(\Gamma)$ connects vertices $v,v'\in V(\Gamma)$, and  $\ort(e) = (v, v')$, then in the derived graph $\hat{\Gamma}$, for every $g\in G$, the edge $\hat{e}_g$ connects the vertices $\hat{v}_g$ and $\hat{v}'_{\gamma(e)g}$. One can also define the \emph{right derived graph} if the edge $\hat{e}_g$ connects the vertices $\hat{v}_g$ and $\hat{v}'_{g\gamma(e)}$. We call the $G$-lifts obtained from the left and right derived graphs \emph{left} and \emph{right} respectively. 

Note that a~$G$-lift $\hat{\Gamma}$ obtained by a~voltage assignment from a~base graph $\Gamma$ is usually called a~\emph{regular lift} or a~\emph{regular cover} of $\Gamma$. 
If a~group $G$ has a~right (resp. left) free action on the vertices and edges of a~graph, and the condition $v\adj_e v'$ implies $vg \adj_{eg} v'g$ (resp. $gv \adj_{ge} gv'$) for every vertices $v,v'$, edge $e$, and $g\in G$, then we say that $G$ has a~\emph{right} (resp. \emph{left}) \emph{free action} on this graph. One can easily check that for any left $G$-lift we can define a~right action of $G$ if for every $\hat{v}_g\in V(\hat{\Gamma})$, $\hat{e}_g\in E(\hat{\Gamma})$, and $h\in G$ we put $\hat{v}_g h:= \hat{v}_{gh}$, $\hat{e}_g h := e_{g h}$.  In what follows, we consider only left $G$-lifts and omit the word ``left''. Note that when the group $G$ is abelian, then there is no difference between left and right $G$-lifts.

When the voltage group is a~cyclic group~$\mathbf{C}_\ell$, then the corresponding derived graphs are also called \emph{shift $\ell$-lifts} and the assigned voltages are called \emph{shifts}.  In the special case when $\ell=2$, and we assign to each edge $e$ of the base graph $\Gamma$ the non-identity shift from $\bC_2$, we obtain the bipartite graph $\bar{\Gamma}$ called the (\emph{bipartite}) \emph{double-cover of $G$}. Since $\bar{\Gamma}$ is the tensor product of $\Gamma$ and the complete graph~$K_2$, then it is not hard to show that $\lambda_2(\bar{\Gamma}) = \lambda(\Gamma)$. Hence this particular $2$-lift almost preserves the spectral expansion properties. Note that if $\Gamma$ is a~bipartite graph then $\bar{\Gamma}$ is a~disconnected graph. Hence, it does not make a~lot of sense to apply this simple construction more than once since on the second iteration one inevitably obtains a~disconnected graph. However, the situation is not that bad if we apply a~large shift $\ell$-lift only once. 
As it was shown in~Theorem~1.2 from~\cite{Agarwal:2019}, if the base graph~$\Gamma$ has good spectral expansion properties, then by using random shifts the obtained graph $\hat{\Gamma}$ also has good expansion properties, even when the lift size $\ell$ is very large. In~\cite{Panteleev&Kalachev:2021}, such graphs $\hat{\Gamma}$ were used to construct quasi-cyclic expander codes of very large lift size $\ell$ such that the corresponding parity-check matrix $H$ and its transpose $H^*$ have good expansion properties.

In the current work, we also obtain graphs $\hat{\Gamma}$ using voltage assignments. We start from a~very small base graph $\Gamma$ such as the bouquet graph $B_w$ (one vertex, $w$ loops) or the graph $D_w$ (two vertices connected by $w$ multiple edges). 
Then we consider a~finite group $G$ with some fixed $w$-element set of generators $S\subseteq G$ and assign each generator from $S=\{s_1,\dots,s_w\}$ to exactly one of the $w$ edges (see Fig.~\ref{fg:base-graphs}). It is not hard to see that the derived graphs for $B_w$ correspond to Cayley graphs $\cay(G,S)$ if the generating set $S$ is \emph{symmetric}, i.e. $S = \{s^{-1} \mid s\in S\}$, and there are no generators $s\in S$ such that $s = s^{-1}$. Let us remind that, given a~finite group $G$ with some symmetric generating set $S$, the corresponding (\emph{left}) \emph{Cayley graph} is the simple graph $\cay(G, S)$ with the set of vertices $V(\Gamma):= G$ and the set of edges $E(\Gamma):= \{ \{g, sg\}\mid g\in G, s\in S  \}$. Now if we assign the elements of a~symmetric generating set $S$ of some finite group $G$ one-to-one to the $w$ edges of the graph $D_w$ (the orientation is shown in~Fig.~\ref{fg:base-graphs}), then we obtain the graph $\cay_2(G,S)$, which is the double-cover of $\cay(G, S)$. The graph $\cay_2(G,S)$ has the set of vertices $V(\Gamma):= G \times \{0, 1\}$ and the set of edges: 
\[E(\Gamma):= \{ \{(g, 0), (sg, 1)\}\mid g\in G, s\in S\}.\] 
Note that the free right action of the group~$G$ on this graph is defined as $(g, a) h:= (gh, a)$ and $\{(g, 0), (sg, 1)\}h:= \{(gh, 0), (sgh, 1)\}$, where $h,g\in G$, $s\in S$, and $a\in \{0, 1\}$. 

\begin{example}\label{ex:Xpq}
Let us now consider the infinite family of $(p+1)$-regular non-bipartite Ramanujan graphs $X^{p,q}$ from~\cite{Lubotzky:1988}, where $p$ and $q$ are two unequal primes such that $q > 2\sqrt{p}$, $p \equiv q \equiv 1 \Mod{4}$, and $p^{(q-1)/2} \equiv 1 \Mod{q}$. The graph $X^{p,q}$ is obtained in~\cite{Lubotzky:1988} as the~Cayley graph $\cay(G, S_{p,q})$, where\footnote{The group $\mathrm{PSL}(\F_q^2)$ is the \emph{projective special linear} group for $\F_q^2$, i.e. the quotient of the group of matrices $A\in \F_q^{2\times 2}$ with $\det A = 1$ modulo its subgroup $\{\pm \left(\begin{smallmatrix} 1 & 0\\ 0 & 1 \end{smallmatrix}\right) \}$.} $G:= \mathrm{PSL}(\F_q^2)$ and $S_{p,q}$ is some specific symmetric set of $p+1$ generators. Denote by $\bar{X}^{p,q}$ the corresponding double-cover $\cay_2(G, S_{p,q})$. Hence $\bar{X}^{p,q}$ is a~$(p+1)$-regular bipartite graph with $n = 2\abs{G}$ vertices, where $\abs{G} = q(q^2 - 1)/2$. Since it is proved in~\cite{Lubotzky:1988} that $\lambda(X^{p,q}) \le 2\sqrt{p}$, then we also have $\lambda_2(\bar{X}^{p,q}) \le 2\sqrt{p}$. Moreover, the graph $\bar{X}^{p,q}$ is a~$G$-lift of the base graph $D_{p+1}$ from Fig.~\ref{fg:base-graphs}, and the group~$G$ has a~free right action on $\bar{X}^{p,q}$.
\end{example}

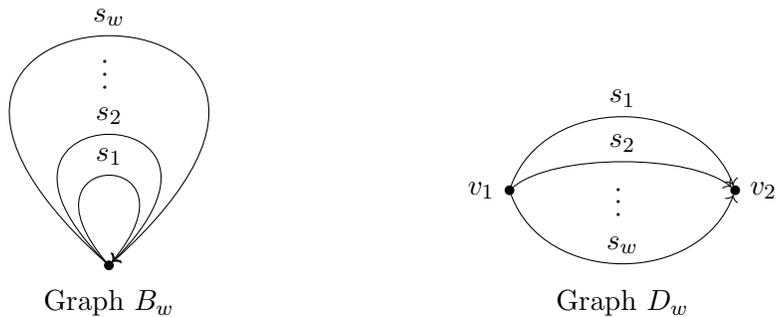
\begin{figure}
  \centering
  \begin{tikzpicture}[scale=1]
    \node[vertex] (v) at (0,0) {};
    \node [anchor=north west][inner sep=0.75pt]  [rotate=-90] at (0,2.8cm) {$\dots$};
    \path[->] (v) edge [in=50,out=130,loop,distance=2cm] node[auto] {$s_1$} (v);
    \path[->] (v) edge [in=45,out=135,loop,distance=3.2cm] node[auto] {$s_2$} (v);
    \path[->] (v) edge [in=42,out=138,loop,distance=6cm] node[auto] {$s_w$} (v);
    \node[empty]  (cap) at (0,-0.5cm) {Graph $B_w$};
  \end{tikzpicture}
  \begin{tikzpicture}[scale=1]
    \node[vertex,label=left:{$v_1$}] (v1) at (0,0) {};
    \node[vertex,label=right:{$v_2$}] (v2) at (3cm,0) {};
    \node [anchor=north west][inner sep=0.75pt]  [rotate=-90] at (1.5cm,0.1cm) {$\dots$};
    \path[->] (v1) edge [in=110,out=70,distance=1.3cm] node[auto] {$s_1$} (v2);
    \path[->] (v1) edge [in=140,out=40,distance=0.7cm] node[auto] {$s_2$} (v2);
    \path[->] (v1) edge [in=-110,out=-70,distance=1.3cm] node[auto] {$s_w$} (v2);
    \node[empty]  (cap) at (1.5cm,-1.5cm) {Graph $D_w$};
  \end{tikzpicture}
    \caption{Voltage assignments for the graphs $B_w$ and $D_w$. The derived graph for $B_w$ corresponds to $\mathrm{Cay}(G,S)$, while the derived graph for $D_w$ is the double-cover of $\mathrm{Cay}(G,S)$. The small arrows shows the orientation that we fix.}
  \label{fg:base-graphs}
\end{figure}

\subsection{Classical codes}
In this subsection we review some standard definitions and terminology related to classical linear codes.
A~{\tmem{linear $[n, k]_q$ code}} is a~$k$-dimensional subspace $\cC \subseteq
\mathbb{F}_q^n$, where the parameters~$n$ and $k$ are called
the~{\tmem{length}} and the~{\tmem{dimension}} of $\cC$, respectively. We
denote the~dimension $k$ of the~code~$\cC$ by $\dim \cC$. The~{\tmem{rate}} of
the~code~$\cC$ is equal to $k / n$. The~elements of $\cC$ are called
{\tmem{codewords}}. The~{\tmem{minimal distance}} $d (\mathcal{C})$ of
the~code~$\mathcal{C}$ is the minimal weight of a non-zero codeword from $\cC$, 
and $d (\cC) := \infty$ when $k = 0$. When a~linear
{$[n, k]_q$} code~$\cC$ has minimal distance~$d$, we say that $\cC$ is an~$[n,
k, d]_q$ code. 

A~linear $[n, k]_q$ code is usually defined either as the row space of
a~matrix $G$ called the {\tmem{generator matrix}}, or as the~kernel of
a~matrix $H$ called the {\tmem{parity-check matrix}}. It~is easy to see that
$GH^* = \zm$, $\rk G = k$, and $\rk H = n - k$. The~code defined by
a~parity-check matrix $H$ is denoted by~$\ker H$. The~vector
space~$\mathbb{F}_q^n$ usually comes with the~standard scalar product $\langle
x, y \rangle = x_1 y_1 + \cdots + x_n y_n$. The {\tmem{dual
code}}~$\cC^{\perp}$ for a~linear $[n, k]_q$ code $\cC$ is the $[n, n - k]_q$
code
\[ \cC^{\perp} = \{x \in \mathbb{F}_q^n \mid \langle x, y \rangle = 0 \text{
   for all } y \in \cC \} . \]
It is not hard to see that a~generator matrix for $\cC$ is a~parity-check matrix for
$\cC^{\perp}$ and vice versa.

\begin{remark}
Note that in the current work it is convenient to consider a~slightly more general case,
where instead of $\mathbb{F}_q^n$ we have an~arbitrary based $n$-dimensional vector
space~$\mathcal{M}$ over~$\mathbb{F}_q$ equipped with some distinguished
basis {$\tilde{\mathcal{M}} = \{ m_1, \ldots, m_n \} \subseteq
\mathcal{M}$}. In this case, $\mathcal{M} \cong \mathbb{F}_q^n$, and we can consider subspaces
$\mathcal{C} \subseteq \mathcal{M}$ as linear codes, and apply all
the~terminology we introduced above to this case as well. 
\end{remark}

In what follows, we will often use the following important definitions.

\begin{definition}
Consider two linear codes $\cC\subseteq \cM$ and $\cC'\subseteq \cM'$, where $\cM$ and $\cM'$ are two $n$-dimensional vector space over $\F_q$ with distinguished bases $\tilde{\cM}$ and $\tilde{\cM}'$ respectively. We say that $\cC$ and $\cC'$ are (\emph{permutation}) \emph{equivalent} and write $\cC\sim \cC'$ if there exists a~linear map $\pi\colon \cM\to\cM'$ such that $\pi(\tilde{\cM}) = \pi(\tilde{\cM}')$ and $\pi(\cC) = \pi(\cC')$. We also say that two $m\times n$ matrices $A$ and $B$ are (\emph{permutation}) \emph{equivalent} and write $A\sim B$ if we can obtain one from another by some row/column permutations. It is clear that if $A\sim B$ then $\ker A \sim \ker B$.

\end{definition}

\subsection{Expander codes}

In this subsection, we describe expander codes, which are Tanner codes obtained from expander graphs. We adopt a~very convenient way, used in~\cite{Meshulam:2018}, \cite{Breuckmann:balanced:2021} to represent these codes in the language of chain complexes and local systems. If $\cF$ is some abelian group and $X$ is some $n$-element set, then we denote by $\cF X$ the abelian group of all formal linear combinations $\sum_{x\in X} a_x x$ of the elements $x\in X$ with coefficients $a_x\in\cF$. When $n=1$, and $X=\{x\}$, we usually write $\cF x$ instead of $\cF\{x\}$.   If $\cF=\F_q$ then the group $\cF X$ is isomorphic to the vector space~$\F_q^n$. When $\cF = \F_q^m$, the group $\cF X$ can be identified with the vector space~$\F_q^{mn}$ of block vectors $(v_1,\dots,v_{n})$ with the blocks $v_i\in\F_q^m$, $i\in [n]$. If $S\subseteq X$ and $a = \sum_{x\in X} a_x x$, then $a|_S := \sum_{x\in S} a_x x$.
Now let us introduce the following important definition.

\begin{definition} Consider a graph $\Gamma = (V, E)$ and a~collection $(\partial^{(v)})_{v\in V}$ of linear maps $\partial^{(v)}\colon\F_q E_v\to\F_q^r v$ called \emph{local boundary maps}, where $E_v$ is the set of edges incident to the vertex $v\in V$. A~\emph{Tanner chain complex $\cT=\tanner_\bullet(\Gamma; (\partial^{(v)})_{v\in V})$} is a~chain complex $\mathbb{F}_q E\xrightarrow{\partial_1} \mathbb{F}_q^r V$ such that for every $e\in E$ that connects $v$ and $v'$ we have: 
\begin{equation}\label{eq:tanner-comp}
\partial e := \partial^{(v)} e + \partial^{(v')} e.    
\end{equation}

\end{definition}

Any Tanner complex $\cT$ defines the~\emph{global} linear code $\cC:= \ker \partial_1$, also known as the~\emph{Tanner code}, and a~number of \emph{local} linear codes $\cC_v:=\ker \partial^{(v)}$, $v\in V$, also known as \emph{subcodes}. We see that $c\in \cC$ \Iff $c|_{E_v}\in \cC_v$ for all $v\in V$. In what follows, we consider Tanner complexes where the matrices of all local boundary maps $ \partial^{(v)}$ are equivalent to one matrix $h\in\F_q^{r\times w}$. Hence all local codes $\cC_v$ are also equivalent to the same linear $[n, k, d]_q$ code $\ker h$. We denote the class of all such Tanner complexes on the graph $\Gamma$ as $\fT(\Gamma; h)$.

We can lift Tanner complexes in a~similar way as we lift graphs using voltage assignments. Consider a~Tanner complex $\cT=\tanner_\bullet(\Gamma; (\partial^{(v)})_{v\in V})$ for the graph $\Gamma$. For any $G$-lift $\hat{\Gamma}=(\hat{V},\hat{E})$, obtained from $\Gamma$ by a~voltage assignment $\gamma\colon E(\Gamma)\to G$, we can define the \emph{$G$-lifted Tanner complex} $\hat{\cT} = \der(\cT; \gamma)$. It is convenient to represent $\hat{\cT}$ as the complex 
\[
 \mathbb{F}_q E\otimes  \F_q G\xrightarrow{\hat{\partial}_1} \mathbb{F}_q^r V \otimes  \F_q G,
\]
where by the tensor product $\otimes$ we mean the tensor product over $\F_q$.
Since $\mathbb{F}_q^r V\otimes  \F_q G\cong \F_q^r \hat{V}$ and $ \mathbb{F}_q E \otimes  \F_q G\cong \F_q \hat{E}$, we can assume that $v\otimes g = (v,g)$ and $e\otimes g = (e, g)$ and still consider $\hat{\cT}$ as a~Tanner complex 
\[\F_q\hat{E}\xrightarrow{\hat{\partial}_1} \F_q\hat{V}\] 
for the graph $\hat{\Gamma}$. The boundary map $\hat{\partial}$ of this complex is defined for every $g\in G$ and $e\in E$ with $\ort(e) = (v, v')$ as \[
\hat{\partial}(e\otimes  g) := \partial^{(v)} e \otimes  g + \partial^{(v')} e\otimes \gamma(e) g,
\]
and extended by linearity (cf. Equation~(\ref{eq:tanner-comp})). Let $\hG = \der(\Gamma;\gamma)$ be a~$G$-lift of a~graph $\Gamma$.  We denote by $\fT_G(\hG; h)$ the class of all $G$-lifted Tanner complexes $\hat\cT=\der(\cT;\gamma)$ where $\cT\in \fT(\Gamma; h)$.

Since the $G$-lifted Tanner complex $\hat{\cT}$ is a~right $G$-module\footnote{We can multiply from the right on its basis as follows: $(e \otimes g) h := e \otimes gh$, $(v \otimes g) h := v \otimes gh$.}, we can use any such complex with the $G$-lifted product construction discussed earlier. Let us now consider a~local $[w, k, d]_q$ code $\ker h$ with the parity-check matrix $h\in \F_q^{r\times w}$, and the Tanner complex $\cT(h) := \left(\mathbb{F}_q E(D_w)\xrightarrow{\partial_1} \mathbb{F}_q^r V(D_w)\right)$ with the boundary map defined as
\[ \partial e_i := h_i v_1 + h_i v_2,\]
where $E(D_w) = \{e_1,\dots,e_w\}$, $V(D_w) = \{v_1, v_2\}$, and $h_i$ is the $i$-th column of the parity-check matrix~$h$. It~is easy to see that the two local codes $\cC_{v_1}$ and $\cC_{v_2}$ of $\cT(h)$ are both equivalent to $\ker h$.
As it was already mentioned, we can obtain the double-cover $\Gamma :=\cay_2(G,S)$ of any Cayley graph $\cay(G,S)$ as the $G$-lift of $D_w$, where $w := \abs{S}$, by a~one-to-one assignment of the $w$ generators from $S$ to the edges of $D_w$ (see~Fig.~\ref{fg:base-graphs}).  Thus we can consider the lifted Tanner complex $\cT(\Gamma;h):=\der(\cT(h);\gamma)$, where $\gamma$ is the corresponding voltage assignment map: $\gamma(e_i):= s_i$, $i\in[w]$. It is not hard to check that the boundary map $\hat{\partial}$ of this lifted complex acts on its bases as follows:
\[
\hat{\partial} (e_i \otimes g) = h_i v_1 \otimes g + h_i v_2 \otimes s_i g, \quad i\in [w].
\]
Let us remind that the chain complex $\cT(\Gamma; h)$ is a~$G$-module.

Let us fix a~graph $\Gamma = \cay_2(G,S)$ and two parity-check matrices $h\in\F_q^{r\times w}$, $h'\in\F_q^{r'\times w}$. We can define the following $3$-term chain complexes using the $G$-lifted product construction: 
\begin{align*}
 \cC_\bullet(\Gamma;h,h') &:= \cT(\Gamma; h) \otimes_G \cT^*(\Gamma; h'),  \\
 \cC'_\bullet(\Gamma;h,h') &:= \cT(\Gamma; h) \otimes_G \cT(\Gamma; h').  
\end{align*}

\begin{remark}\label{rm:LTC-constr}
Let $\bar{X}^{w-1,t} = \cay_2(G, S_{w-1,t})$ be the $w$-regular graph from Example~\ref{ex:Xpq}, where $G = \mathrm{PSL}(\F_t^2)$. Consider the chain complexes $\cC_\bullet(\bar X^{w-1,t};h,h')$ and $\cC'_\bullet(\bar X^{w-1,t};h,h')$ respectively. In the current work, we use the first complex to show the existence of two asymptotically good families of codes: quantum LDPC codes and classical LTCs. However, as we can see from Theorem~\ref{th:LTC}, the rate of the obtained LTCs is bounded above by $1/2$. We conjecture\footnote{Note that a~construction similar to this complex was used in~\cite{Dinur:2021} to produce asymptotically good classical LTCs.} that the complex $\cC'_\bullet(\bar X^{w-1,t};h,h')$ can be used to obtain asymptotically good LTCs of rate arbitrary close to $1$. For example, if $h,h'\in \F_q^{r\times w}$, then the rate of the classical codes $\ker \partial_2$ obtained from~$\cC'_\bullet(\bar X^{w-1,t};h,h')$ is at least $1 - 4r/w$ since we have $w^2|G|$ code symbols and $4wr|G|$ parity-checks. Hence if the rate of the local codes goes to $1$, the same happens with the rate of the obtained LTCs.
\end{remark}

\subsection{Posets and incidence chain complexes}
In this subsection, we consider based chain complexes $\cI$ with integer coefficients\footnote{Chain complexes with integer coefficient are often used in algebraic topology to study the integral homology groups of CW-complexes} and call the elements from the corresponding distinguished basis $\tilde\cI$ \emph{cells}. We say that $\cI$ is an~\emph{incidence chain complex} if the matrix of its boundary map $\partial$ contains only elements from $\{-1,0,1\}$, and for every such a~complex we also define its \emph{cell poset}, which can be viewed as a~combinatorial structure that represents the incidence relation between the cells.  In~some sense, one can view the cell poset with the corresponding incidence chain complex\footnote{In fact, if the reader is only interested in the codes over finite field of even characteristic, then the signs in the matrix~$\partial$ are not relevant, and we can represent every~abstract cell complex by the corresponding cell poset.} as an~\emph{abstract cell complex} (see, e.g.~\cite[Section~2.12]{Hilton:1960}), which generalizes the notion of an~abstract simplicial complex and an~abstract polytope~\cite{McMullen:2002}. 

Let $X$ be a~\emph{poset}, i.e., a~set with a~partial order $\le$. We say that an~element $a\in X$ \emph{covers} an~element $b\in X$ and write $a\prec b$ or $b \succ a$ if $a < b$, and there is no element $c\in X$ such that $a < c < b$. 
It is easy to see that any finite poset can be uniquely defined by its covering relation $\prec$ if we let $a \le b$ \Iff there exists a~sequence $c_0 \prec c_1 \prec \dots \prec c_n$ of elements from $X$ such that $c_0 = a$, $c_n = b$, and $n\ge 0$.  Let $\cC$ be a~based chain complex over some ring\footnote{In this section, we are interested in only two cases: $R=\Z$ and $R=\F_q$.} $R$. We can define the~partial order $\le$ on the distinguished basis~$\tilde{\cC}$ if for every two cells $c,c'\in \tilde{\cC}$ we put $c' \prec c$  \Iff $c'\in \supp \partial c$. We call the poset~$\tilde{\cC}$ with the relation $\le$ the \emph{cell poset of~$\cC$}.  
 
A~\emph{graded poset} is a~poset $X$ equipped with a~map $\rho\colon X\to \Z$ called a~\emph{rank function} such that for any $a,b\in X$ the following conditions hold:
\begin{enumerate}
    \item if $a \le b$ then $\rho(a) \le \rho(b)$;
    \item if $a\prec b$ then $\rho(b) = \rho(a) + 1$.
\end{enumerate}
If $X$ is a~finite graded poset, then it is not hard to see that it can be decomposed as 
\[
X = X(s) \sqcup X(s+1)\sqcup \dots \sqcup X(t),
\] 
where the subset $X(i) := \{a \in X \mid \rho(a) = i \}$ is called the \emph{$i$-th level of $X$}, $i\in [s,t]\cap\Z$. It~is clear that all the elements from $X(s)$ (resp. $X(t)$) are minimal (resp. maximal) elements in $X$. It~is also trivial to check that the cell poset $\tilde{\cC}$ of a~based (co)chain complex $\cC$ is a~graded poset, where the levels correspond to the cells of the same dimension. 

Another example of a~graded poset, often studied in the context of HDXs, is an~(\emph{abstract}) \emph{simplicial complex} on a~finite non-empty set $V$, which is defined as a~closed under taking subsets family $X \subseteq 2^V$. In this case, the partial order $\le$ is just the set inclusion relation $\subseteq$, and $\rho(x) := \abs{x} - 1$ for every $x\in X$. The elements $x\in X$ with $\rho(x) = i$ are called \emph{$i$-dimensional faces} or just \emph{$i$-faces}. The highest dimension of the faces from the simplicial complex $X$ is called its \emph{dimension}. Let us note that a~simple graph can be represented as a~1-dimensional simplicial complex, where the 0-faces and the 1-faces correspond respectively to the vertices and the edges of the graph. Hence we can also view an~undirected~graph $\Gamma$ as the graded poset with the levels $V(\Gamma)$ and $E(\Gamma)$, where for every $v\in V(\Gamma)$ and $e\in E(\Gamma)$ we have $v \prec e$ whenever $v$ is incident to $e$. In fact, 2-level posets are equivalent to the \emph{incidence systems}, and thus can be used to represent undirected multigraphs and hypergraphs as well.

In this work, it is convenient to define objects such as graphs, incidence systems, and simplicial complexes by the corresponding based chain complexes over $\Z$. In some way, we can view such complexes with integer coefficients as a~vast generalization of these objects. For example, for any 2-level poset $X$ with the levels $V$ and~$E$, we can define the based chain complex $\mathcal{C}_{\bullet} (X) := \left( \mathbb{Z}E\xrightarrow{\partial_1} \mathbb{Z}V \right)$ with the distinguished bases $\tilde{\cC}_0 := V$, $\tilde{\cC}_1 := E$, where 
\[
\partial e := \sum_{\substack{v\prec e\\v\in V}} v.
\]
The matrix of $\partial_1$ is a~zero–one matrix usually  called the \emph{incidence matrix of $X$}. For example, since we view an undirected graph $\Gamma$ as a~$2$-level poset, we can consider the corresponding chain complex~$\cC_\bullet(\Gamma)$. Now let $X$ be a~simplicial complex with some fixed linear order $<_V$ on its set of vertices $V = X(0)$. Then we can define the chain complex $\cC_{\bullet}(X)$ by the following diagram
\[
\mathbb{Z}X (n) \xrightarrow{\partial_n} \cdots \xrightarrow{\partial_1}
\mathbb{Z}X (0) \xrightarrow{\partial_0} \mathbb{Z}X (-1),
\]
where for every $k$-face $x=\{v_0, \dots, v_k\}\in X$ such that $v_0 <_V \dots <_V v_k$ the boundary map ${\partial\colon \Z X \to \Z X}$ is defined as $\partial x := \sum_{i=0}^{k} (-1)^i x\setminus\{v_i\}$, and then extended by linearity to all chains from $\Z X$. 
As we can see, the integer coefficients in the matrix of the boundary maps for $\cC_\bullet(\Gamma)$ and $\cC_\bullet(X)$ are from the set $\{-1,0,1\}$. Let us call any based chain complex $\cI$ with this property\footnote{In fact, sometimes it is also convenient to consider arbitrary integer coefficients. But this more general case is not covered here.} an~\emph{incidence complex}. Let $\cI$ be some incidence complex with a~distinguished basis $X$.  It is clear that its boundary map $\partial\colon \Z X \to \Z X$ acts on a cell $x \in X$ as
\begin{equation}\label{eq:incidence-comp}
\partial x = \sum_{\substack{x \succ x'\\ x'\in X}} [x : x'] x',
\end{equation}
where the coefficient $[x : x']\in \{-1, +1\}$ is called the \emph{incidence number} for $x,x'\in X$. It is also convenient to assume that $[x : x']=0$ whenever $x \not\succ x'$. Let us note that since $\partial^2 = 0$, then for every $x,x''\in X$ we obtain
\begin{equation}\label{eq:incidence}
    \sum_{\substack{x \succ x' \succ x''\\x'\in X}} [x : x'][x' : x''] = 0.
\end{equation}

\subsection{Products of graphs and posets}\label{sc:prod-graph}

By interpreting objects like graphs, hypergraphs, or more generally abstract cell complexes as the corresponding incidence complexes allows us to define the lifted product of such objects. 
We say that a~group $G$ \emph{acts} on a~poset~$P$ if it acts on $P$ as on a~set, and for every $g\in G$ if $x \le y$ then $gx \le gy$ (resp. $xg \le yg$ in the case of a~right action). It is readily seen that an~action of a~group on a~graph $\Gamma$ is also an~action on~$\Gamma$ as a~$2$-level poset.
Therefore if $\cI_X$ and $\cI_Y$ are incidence complexes with cell posets $X$ and~$Y$, respectively, where a~group $G$ acts freely (from the right on~$X$ and from the left on~$Y$), then we can define the \emph{lifted product} $X \times_G Y$  of $X$ and~$Y$ \emph{over~$G$} as the cell poset of the complex $\cI_X \otimes_G \cI_Y$. In fact, the lifted product $X \times_G Y$ can be defined for arbitrary finite posets $X$ and $Y$ with a~free action of a~group $G$. Recall that if we have a~free action of a~group~$G$ on a~set~$S$, then the size of each orbit is equal to $|G|$, and we can identify $S$ with $(S/G)\times G$, where $S/G$ is the set of all orbits under the action of $G$. We define the poset $X \times_G Y$ as the set $(X/G)\times G\times (Y/G)$ in terms of the covering relations as follows: we have $(x,g,y) \succ (x',g',y')$ \Iff either $x = x'$ and $(y,g) \succ_Y (y',g')$ or $(x,g) \succ_X (x',g')$ and $y=y'$.  If the posets $X$ and $Y$ are graded, then we can also define the rank function $\rho(\cdot)$ for $X\times_G Y$ in terms of the rank functions of $X$ and $Y$ as follows: $\rho(x,g,y) := \rho_X(x,g) + \rho_Y(y,g)$. If~$|G|=1$ we denote the poset $X \times_G Y$ simply by $X\times Y$.

\begin{remark}
If $X$ and $Y$ are two graphs (considered as $2$-level posets), then from the geometrical point of view the poset $X \times Y$ corresponds to the direct product of $X$ and $Y$ (as topological graphs). At the same time, the geometrical interpretation of the poset $X \times_G Y$ can be given in terms of the balanced product\footnote{A~geometric realization of a~graph can be considered as a~topological space. The \emph{balanced product} of two topological spaces $X$ and $Y$ with a~group $G$ acting on the right on $X$ and on the left on $Y$ is the quotient space $X \times_G Y := X \times Y/\sim$, where the equivalence relation $\sim$ is induced by $(xg,y)\sim (x, gy)$ for $x\in X$, $y\in Y$, $g\in G$.} of graphs~\cite{Breuckmann:balanced:2021}. 
Note that the \emph{$1$-skeleton} of $X \times Y$, i.e., its restriction to the first two levels, is the $2$-level poset representing the graph~$X \mathbin{\Box} Y$, which is usually called the \emph{Cartesian product} of the graphs $X$ and $Y$. Recall that for every $G$-lifted graph $\Gamma$ the group $G$ acts freely on $\Gamma$.
Hence, we can also define the \emph{$G$-lifted Cartesian product} $\hat X \mathbin{\Box_G} \hat Y$ for $G$-lifts $\hat X$, $\hat Y$ of base graphs $X$, $Y$ as the $1$-skeleton of $\hat X \times_G \hat Y$. It~is not hard to check that the graph $\hat X \mathbin{\Box_G} \hat Y$ is a~$|G|$-fold cover for the standard Cartesian product $X \mathbin{\Box} Y$. Furthermore, if $G$ is abelian, then this cover is regular, i.e., $\hat X \mathbin{\Box_G} \hat Y$ is a~$G$-lift of $X \mathbin{\Box} Y$. 
\end{remark}

Suppose that $\hG$ is a~$G$-lift of some base graph $\Gamma$. Consider the cell poset $\tilde X := \hG \times_G \hG$, and let us represent its elements by triples $x\cdot g\cdot y$, where  $x,y\in V(\Gamma)\cup E(\Gamma), g\in G$.   From the  definition of the poset $\tilde X$ it follows that 
$x' \cdot g' \cdot y' \succ x \cdot g \cdot y$ \Iff one of the following conditions hold:  
\begin{enumerate}
    \item $\hat x'_{g'} \succ_{\hG} \hat x_g$ and $y = y'$;
    \item $x = x'$ and $\hat y'_{g'} \succ_{\hG} \hat y_g$;
\end{enumerate}
where $\succ_{\hG}$ is the covering relation in the graph~$\hG$ considered as a $2$-level poset, i.e., its incidence relation. It is convenient to interpret the poset $\tilde X$ as a~$2$-dimensional geometric object.
An~element $x \cdot g \cdot y\in \tilde X$ is called:
\begin{itemize}
    \item a~\emph{vertex} if $x\in V(\Gamma),y\in V(\Gamma)$;
    \item a~\emph{horizontal edge} if $x\in E(\Gamma), y\in V(\Gamma)$;
    \item a~\emph{vertical edge} if $x\in V(\Gamma), y\in E(\Gamma)$;
    \item a~\emph{face} if $x\in E(\Gamma), y\in E(\Gamma)$,
\end{itemize}
and the corresponding subsets of elements are denoted as $V = V(\tilde X)$, $\Eh = \Eh(\tilde X)$, $\Ev = \Ev(\tilde X)$, and $F = F(\tilde X)$. We also define the set $E(\tilde X) = \Eh(\tilde X) \cup \Ev(\tilde X)$.

If $P$ is a~poset we denote by $P^*$ the \emph{dual poset}, i.e., $x \le_{P^*} y$ whenever $y \le_{P} x$. In what follows, we will also need a~poset $X := \hG \times_G \hG^*$, which is defined on the \emph{same} set as $\tilde X = \hG \times_G \hG$ but has \emph{different} partial order and rank function. This means that the grading of $X$ is different from~$\tilde X$. It~is easy to check that the cell poset $\tilde X$ has $3$ levels: $\tilde X(0) := V$, $\tilde X(1) := \Ev\cup \Eh$, and $\tilde X(2) := \Eh$, while the levels for $X$ are as follows: $X(0) := \Ev$, $X(1) := F\cup V$, and $X(2) := \Eh$. 
\begin{remark}
  As we will see in Section~\ref{sc:proof-out}, the poset $X = \hG \times_G \hG^*$ corresponds to the lifted product complex  $\cT(\Gamma; h) \otimes_G \cT^*(\Gamma; h')$, which we use to show the main result. However the levels in the poset $X$ do not correspond to the natural geometrical dimension of the cells, and in the proof of our main result it is more convenient to work with the poset $\tilde X = \hG \times_G \hG$ defined on the same set as $X$, but giving it a~natural geometrical interpretation as a~$2$-dimensional complex. To this end we define the incidence relation $\inc(\cdot,\cdot)$ on the set $V\cup \Eh\cup \Ev\cup F$ in a~standard geometrical sense, i.e., we assume that $\inc(x,y)$ \Iff $x \le y$ or $y \le x$, where $\le$ is the partial order of the poset~$\tilde X = \hG \times_G \hG$. For example, every face can be represented geometrically as a~square incident to two horizontal edges, two vertical edges, and to four vertices. 
  If $x\in X$ and $S,T\subseteq X$, then we also use the following notations:
\begin{align*}
    S_x &:=\{y\in S\mid \inc(x,y)\},\\
    S_T &:=\{y\in S\mid \inc(x,y) \mbox{ for some }x\in T\}.
\end{align*}
Hence $S_x$ is the subset of the elements from~$X$ incident to $x$, and $S_T$ is the the subset of the elements from~$S$ incident to some element from $T$. For example, $X_v = \{v\} \cup E_v\cup F_v$ is the set of all cells incident to $v$ called the \emph{star} of $v$, where $E_v$ (resp. $F_v$) is the set of edges (resp. faces) incident to $v$.   
\end{remark}

For the proof of our main result we will also need the~$1$-skeleton  $\Lambda := \hG \mathbin{\Box_G} \hG$ of $\hG \times_G \hG$ with the set of vertices~$V(\Lambda) := V(\tilde X)$ and the set of edges~$E(\Lambda) := E(\tilde X)$. 
\begin{remark}
In the proof of the main result in~Section~\ref{sc:proof}, when we mention sets $V$, $E$, $\Eh$, $\Ev$, $F$ or a~graph~$\Lambda$, we refer to the corresponding sets and the graph defined for the poset $\hG \times_G \hG$ in this section unless otherwise stated. 
\end{remark}

\subsection{Local systems}

In this subsection, we consider a~generalization of based chain complexes with coefficients from some field or ring to the complexes with \emph{local system of coefficients}, where the chains are formal linear sums of cells with coefficients in arbitrary abelian groups. 
In~fact, in this work, we are interested in the case when all these abelian groups are vector spaces over the same finite field $\F_q$, and thus the corresponding chain complexes can be still considered as complexes of vector spaces over $\F_q$.  In~some sense, a~complex with local coefficients gives us a~high-level view of the corresponding complex over~$\F_q$.

Let $X$ be some finite set, which we are going to use as an~index set. If a~vector space~$\cC$ is the~direct sum $\bigoplus_{x\in X} \cF_x$ of a~collection of vector spaces $\cF = (\cF_x)_{x\in X}$, then we can consider the elements of~$\cC$ as formal sums $\sum_{x\in X} a_x x$ of elements from $X$, where for every $x\in X$ the coefficient $a_x$ is from the vector space~$\cF_x$ called the \emph{local coefficient space} of~$x$. In such cases, we also denote the vector space~$\cC$ by~$\cF X$ or by~$A X$ when all the local coefficient spaces are equal to the same space~$A$. If each local coefficient space $\cF_x$ comes with a~distinguished basis $\tilde{\cF_x}$, then we assume that the distinguished basis for $\cF X$ is the set $\{ax \mid a\in \tilde{\cF_x},\ x\in X \}$, in which case we say that $\cF X$ is \emph{based}.

\begin{definition}
Given a~poset $X$ we say that $\cF$ is a~\emph{local system of coefficients} for $X$ if to each $x\in X$ we assign a~vector space $\cF_x$, and to each $x,x'\in X$ where $x \ge x'$ we assign an~$\F_q$-linear map $\cF_{x\to x'}\colon \cF_x \to \cF_{x'}$ such that whenever $x \ge x' \ge x''$ we have:
\[
\cF_{x'\to x''}\circ\cF_{x\to x'} = \cF_{x\to x''}.
\]
\end{definition}
\begin{remark}
Note that in the language of category theory we can view $\cF$ as a~\emph{functor} from a~poset $X$ to the category of vector spaces over $\F_q$. Here we consider the poset $X$ as a~small category, where the objects are the elements of $X$, and we have an~arrow $x\to x'$ whenever $x \ge x'$. 
\end{remark}

Given an~incidence chain complex $\cI$ with some local system $\cF$ on its cell poset $X:=\tilde{\cI}$, we can consider the chain complex $\cC_\bullet(\cI; \cF)$ as the vector space $\cF X$ over $\F_q$ with the boundary map $\partial\colon \cF X \to \cF X$ defined on the elements $ax\in \cF X$, where $a\in \cF_x, x\in X$, as follows:
\[
\partial (a x) := \sum_{\substack{x \succ x'\\x'\in X}} [x : x']\cF_{x\to x'}(a)x',
\]
and extended to all formal sums $\sum_{x\in X} a_x x$ by linearity. It is easy to prove that $\partial^2 = 0$. Indeed, it is enough to check that 
\[
\partial^2 (a x) := \partial\sum_{\substack{x \succ x'\\x'\in X
}} [x : x']\cF_{x\to x'}(a)x' = \sum_{\substack{x \succ x' \succ x''\\x',x''\in X}}[x : x'][x' : x''] \cF_{x\to x''}(a)x'' = 0,
\]
where the last step follows from~(\ref{eq:incidence}). Note that if $\cF X$ is based, then the chain complex $\cC_\bullet(\cI; \cF)$ is also based.
\begin{remark}
  With some small abuse of notation, we usually denote the complex $\cC_\bullet(\cI;\cF)$ by $\cC_\bullet(X;\cF)$, in which case we always assume that the cell poset $X$ comes with the corresponding incidence complex $\cI$, i.e., for every two elements $x\ge x'$ from $X$ their incidence number $[x:x']\in \{-1,+1\}$ is defined (cf. \emph{abstract cell complex} from~\cite[Section~2.12]{Hilton:1960}).   In fact, in the case of complexes over the fields of characteristic $2$, we can always assume that $[x:x']=1$ if $x \ge x'$, and $[x:x']=0$ otherwise. Hence, in such cases, the poset $X$ completely defines the corresponding incidence complex $\cI$ by (\ref{eq:incidence-comp}).   
\end{remark}

Consider a~based chain complex $\cC=\cC_\bullet(X;\cF)$ over $\F_q$. Let $a = \sum_{x \in X}a_x x \in \cC$, where each coefficient $a_x$ is from the based vector space $\cF_x$ over $\F_q$. We denote by $\wt(a)$ the standard Hamming weight of $a$, considered as a~vector over $\F_q$. We also consider the \emph{block weight}  $\wt_{X}(a)$ defined as the number non-zero blocks in $a$, viewed as a~block vector $(a_{x})_{x\in X}$, i.e. we have 
\[
\wt_X (a) := \card \{x \in X \mid a_x \ne 0 \}.
\]
Sometimes we need to take into account only the blocks that correspond to some subset $S\subseteq X$. In~this case, we can define the \emph{block weight} $\wt_S(a):= \card \{x \in S \mid a_x \ne 0 \}$ \emph{relative} to the subset~$S\subseteq X$. We also define $\supp a := \{x\in X \mid a_x \ne 0\}$ and $x|_S := \sum_{x\in S} a_x$, where $a = \sum_{x\in X} a_x x$.

Let $\partial\colon \cF X\to\cF X$ be the boundary map of $\cC$. In some cases, we want to restrict the domain and codomain of $\partial$. For every $S,T\subseteq X$ we consider the map $\partial_{S\to T}\colon \cF S \to \cF T$ defined as $a\mapsto (\partial a)|_T$.  From the definition it is clear that for every $a\in \cF X$ we have:
\begin{equation}\label{eq:boundary-res}
    (\partial (a|_S))|_T = \partial_{S\to T} (a|_S).
\end{equation}

As we already mentioned, local systems can be used to obtain a~high-level view of a~chain complex over $\F_q$. For example, we can represent a~Tanner complex 
\[
\tanner_\bullet(\Gamma; (\partial^{(v)})_{v\in V}) = \rbr{\F_q E\xrightarrow{\partial_1}\F_q^r V}
\] 
for a~graph $\Gamma$ (considered as a~2-level poset) as the~complex $\cC_\bullet(\Gamma; \cF)$, where for every $v\in V$ we have $\cF_v := \F_q^r$, for every $e\in E$ we have $\cF_e := \F_q$, and if $e$ is incident to $v$ then $\cF_{e\to v} := \partial^{(v)}|_{\F_q e}$. In~the next subsection, we show that the $G$-lifted product of two $G$-lifted Tanner complexes can also be represented as a~complex with a~local system on the poset $\hG \times_G \hG^*$ from Subsection~\ref{sc:prod-graph}.

With some small abuse of terminology in what follows we call Tanner complexes \emph{Tanner codes} and sometimes identify such a~complex with the global code it defines.

\section{Proof of the main results}\label{sc:proof}

\subsection{Local minimality}

One of the key ideas used in the proof of our main result is the idea of local minimality. It was used previously in the context of cohomology of simplicial complexes with $\F_2$-coefficients~\cite{Kaufman:2014,Evra:2020}. In the current work, we extend this idea to a~much more general context of (co)homology of abstract cell complexes with local systems of coefficients. 
As we mentioned before, by an~\emph{abstract cell complex} we mean a~poset $X$ with a~map $\partial\colon \Z X\to\Z X$ such that $\Z X$ is an~incidence complex with the boundary map~$\partial$, and $X$ is its cell poset. 

Consider an~abstract cell complex $X$ and a~based chain complex $\cC=\cC_\bullet(X;\cF)$ of vector spaces  
\[
\cdots\xrightarrow{\partial_{i+1}}\cC_{i}\xrightarrow{\partial_i}\cC_{i-1} \xrightarrow{\partial_{i-1}}\cdots
\]
over $\F_q$, where $\cF$ is a~local system on $X$. 
Denote by $|\cdot|$ the block weight $\wt_X(\cdot)$, which makes each term $\cC_i$ in this complex a~normed abelian group (see Appendix~\ref{sc:normed-group}) with the norm $|\cdot|$ and allows us to define the distance in the standard way: $d(a,b) := |a-b|$,  $d(a, \cB) := \min_{b\in \cB} |a-b|$. 
We also use $|\cdot|$ to define for every $i\in\Z$ the corresponding quotient norm on the $i$-th homology group $H_i(\cC) = Z_i(\cC)/B_i(\cC)$ called the \emph{systolic norm}  by the formula $|\cA| := \min_{a\in \cA} |a|$, where $\cA\in H_i(\cC)$, $B_i(\cC) = \im \partial_{i+1}$, $Z_i(\cC) = \ker \partial_i$. This in turn allows us to define the distance on $H_i(\cC)$ as $d(\cA,\cB) := |\cA - \cB|$ and consider the minimal distance of $H_i(\cC)$ given by the standard formulas:
\[
d(H_i(\cC)) := \min_{\substack{\cA\ne \cB\\ \cA, \cB\in H_i(\cC)}} d(\cA,\cB) = \min_{\cA\in H_i(\cC)\setminus\{B_i(\cC)\}} |\cA| = \min_{a\in Z_i(\cC)\setminus B_i(\cC)} |a|. 
\]
Note that the minimal distance of $H_i(\cC)$ is also called the \emph{$i$-systolic distance} of $\cC$, while the distance $d(H_i(\cC^*))$ of the dual chain complex $\cC^*$ is called its \emph{$i$-cosystolic distance}.  These distances are related to the minimal distance $d(\cQ)$ of the quantum CSS code $\cQ = \cQ(\partial_i,\partial_{i+1}^*)$ over $\F_q$ defined by three consecutive terms of the complex
\[
\cC_{i+1}\xrightarrow{\partial_{i+1}}\cC_{i}\xrightarrow{\partial_i}\cC_{i-1}.
\]
It is easy to see that $d(\cQ)\ge \min(d(H_i(\cC)), d(H_i(\cC^*)))$, where we have the equality if $\cF_x=\F_q$ for all $x\in X$ since the block Hamming weight $\wt_X(\cdot)$ is less than or equal to the corresponding Hamming weight $\wt(\cdot)$.

\begin{definition}
We say that an~$i$-chain $c\in \cC_i$, $i\in\Z$, is \emph{locally minimal} (\emph{with respect to $X$}) if $|c + \partial ax| \ge |c|$  for all $x\in X(i+1)$ and $a\in\cF_x$. 
We also define the value 
\[
\dLM^{(i)}(\cC):=\min \{|c| \mid c\in Z_i(\cC)\setminus\{0\},c\mbox{ is locally minimal}\}, 
\]
which we call the \emph{$i$-th locally minimal distance} of $\cC$. If we do not have non-zero locally minimal $i$-cycles, then we assume that $\dLM^{(i)}(\cC) = \infty$.
\end{definition}
Note that in general the locally minimal distance $\dLM^{(i)}(\cC)$ is not equal to the minimal distance of $Z_i(\cC)$ since the codewords of minimal weight from $Z_i(\cC)$ are not necessarily locally minimal. For example, in the context of $w$-limited qLDPC codes where $|\partial x| \le w$ for every $x\in X(i+1)$, and thus we have $\partial x \in Z_i(\cC)$ and $d(Z_i(\cC)) \le w$, one can see that the codeword $c = \partial x$ is not locally minimal since $|c - \partial x| = 0 < |c|$.

The next lemma connects the locally minimal distance of the complex to the properties of the corresponding quantum and classical codes obtained from it. The first assertion can be used to obtain the lower bound on the minimal distance $d(\cQ)$ of the corresponding quantum CSS code $\cQ$, while the second one can be used to show that the space $Z_{i+1}(\cC)$ is a~locally testable code.

\begin{lemma}\label{lemma:dist-from-lm}
    Let $\cC=\cC_\bullet(X;\cF)$ be a~chain complex, where $\cF$ is a~local system on $X$. Then for every $i\in\ZZ$ we have
    $$d(H_i(\cC))\ge \dLM^{(i)}(\cC),$$
    and for every chain $c\in\cC_{i+1}$ such that $|\partial c| < \dLM^{(i)}(\cC)$  we have
    \begin{equation}\label{eqn:ltc-from-lm}
        |\partial c|\ge d(c, Z_{i+1}(\cC)).
    \end{equation}
\end{lemma}
\begin{proof}
    By definition we have
    \[
    d(H_i(\cC))=|c_0|,\quad\mbox{ where }\quad c_0:=\argmin_{c\in Z_i(\cC)\setminus B_i(\cC)}|c|.
    \]
    Since the element $c_0$ has the minimal norm in the coset $c_0+B_i(\cC)$, it is also locally minimal. Hence we have $d(H_i(\cC))=|c_0| \ge \dLM^{(i)}(\cC)$.
    
    We prove the second claim by induction on $|\partial c|$. If $|\partial c|=0$ then $d(c,Z_{i+1}(\cC))=0$, and \eqref{eqn:ltc-from-lm} is true. Consider $c\in\cC_{i+1}$ such that $0<|\partial c|<\dLM^{(i)}(\cC)$. Since $\partial c\in Z_i(\cC)$ and $|\partial c| < \dLM^{(i)}(\cC)$, we see that $\partial c$ cannot be locally minimal, and hence there exists $a\in \cF_x$ where $x\in X(i+1)$ such that $|\partial(c+ax)|\le|\partial c|-1$. Therefore by the induction hypothesis we have $|\partial (c + ax)|\ge d(c + ax,Z_{i+1}(\cC))$. Thus we obtain
    \[
    d(c,Z_{i+1}(\cC))\le d(c+ax,Z_{i+1}(\cC))+|ax|\le |\partial(c+ax)| + \underbrace{|ax|}_{=1}\le |\partial c|,
    \]
    which completes the proof of the second claim.
\end{proof}

\subsection{Graph expansion}\label{sc:graph-exp}

For any graph $\Gamma$ we denote by $\Gamma^2$ the graph with $V(\Gamma^2) = V(\Gamma)$ and $A(\Gamma^2) = (A(\Gamma))^2$, i.e., the number of edges connecting two vertices in $\Gamma^2$ is equal to the number of length~2 paths connecting them in $\Gamma$. In this section, we prove several technical lemmas to establish expanding properties of the graphs $\Lambda$ and $\Lambda^2$, where $\Lambda$ is the graph defined in Subsection~\ref{sc:prod-graph}. If $\Gamma=(V,E)$ is a~graph (possibly with multiple edges), and $S,T\subseteq E$, then by $E_\Gamma(S,T)$, we denote the set of oriented edges from $S$ to $T$, i.e. $E_\Gamma(S,T):= \{(s,e,t)\mid e\in E; s\in S,t\in T; s\adj_e t\}$ (every edge connecting $s,t\in S\cap T$ is counted twice). We also usually write $E(S,T)$ and $E(S)$ instead of $E_\Gamma(S,T)$ and $E_\Gamma(S)$ if the graph $\Gamma$ is clear from the context. 

\begin{definition}
We say that a~graph $\Gamma$ is an~\emph{$(n, w, \lambda)$-expander} if it is a~simple $w$-regular graph on $n$~vertices such that $\lambda = \lambda(G)$. 
\end{definition}
Let us now state without proof a~well-known variant of the expander mixing lemma for $(n,w,\lambda)$-regular graphs~\cite[Lemma~2.5]{Hoory:2006}.
\begin{lemma}[Expanding mixing lemma]\label{lm:mixing}
    If $\Gamma = (V,E)$ is an~$(n,w,\lambda)$-expander graph, then for every $S,T\subseteq V$ we have:
    \[
    \abs{|E(S,T)| - w\frac{\abs{S}\abs{T}}{n}} \le \lambda\sqrt{\abs{S}\abs{T}}. 
    \]
\end{lemma}

In what follows, it will be convenient to define a~property called $(a,\lambda)$-edge-expansion, which captures the edge expansion on small sets of vertices in a~graph.

\begin{definition}
We say that a graph $\Gamma$ is \emph{$(a,\lambda)$-edge-expanding} if for any $S,T\subseteq V(\Gamma)$ such that $|S|,|T|\le a$ the following condition holds:
$$|E(S,T)|\le \lambda\sqrt{|S||T|}.$$
\end{definition}

\begin{lemma}\label{lemma:spec2edge}
    If $\Gamma$ is an~$(n,w,\lambda)$-expander graph, then it is $(\lambda n/w, 2\lambda)$-edge-expanding.
\end{lemma}
\begin{proof}
    If $\Gamma$ is a $w$-regular, then from Lemma~\ref{lm:mixing} it follows that for any $S,T\subseteq V(\Gamma)$ such that $|S|,|T|\le \lambda n/w$ we have $\absbr{|E(S,T)|- w\frac{|S||T|}{n}}\le\lambda\sqrt{|S||T|}$. Hence we have:
    $$|E(S,T)|\le w\frac{|S||T|}{n}+\lambda\sqrt{|S||T|}\le\rbr{\frac{\lambda n}{w}\cdot\frac{w}{n}+\lambda}\sqrt{|S||T|}=2\lambda\sqrt{|S||T|},$$
    and the Lemma is proved.
\end{proof}

\begin{lemma}\label{lemma:cover-exp}
    If $\hat{\Gamma}$ is a~$G$-lift of an~$(a,\lambda)$-edge-expanding base graph $\Gamma$, then $\hat{\Gamma}$ is $(a,|G|\cdot\lambda)$-edge-expanding. 
\end{lemma}
\begin{proof}
    Consider subsets $\hat{S},\hat{T}\subseteq V(\hat{\Gamma})$ such that $|\hat{S}|,|\hat{T}|\le a$, and let  $S,T\subseteq V(\Gamma)$ be their projections\footnote{The \emph{projection} of a~vertex $(v, g)\in V(\tilde{\Gamma})$ is the vertex $v\in V(\Gamma)$.} to the base graph $\Gamma$. Since each edge of $\Gamma$ is the projection of $m=\abs{G}$ edges from $\hat{\Gamma}$, then using the edge-expansion of the base graph~$\Gamma$ we have:
    $$|E(\hat{S},\hat{T})|\le m|E(S,T)|\le m\lambda\sqrt{|S||T|}\le m\lambda\sqrt{\rule{0ex}{1.9ex}\smash{|\hat{S}||\hat{T}|}}.$$
\end{proof}

\begin{lemma}\label{lemma:Xexp}
    Every graph $\bar{X}^{w-1,t}$ from Example~\ref{ex:Xpq} is $(n/\sqrt{w},8\sqrt{w})$-edge-expanding, where $n=t(t^2-1)$ is the number of its vertices.
\end{lemma}
\begin{proof}
    Since the Ramanujan graph $X^{w-1,t}$ is an~$(n/2,w,2\sqrt{w})$-graph, then by Lemma~\ref{lemma:spec2edge} it is $(n/\sqrt{w}, 4\sqrt{w})$-edge-expanding. Moreover, since the graph $\bar{X}^{w-1,t}$ is a~2-lift of $X^{w-1,t}$, then by Lemma~\ref{lemma:cover-exp} it is $(n/\sqrt{w},8\sqrt{w})$-edge-expanding.
\end{proof}

\begin{remark}\label{rm:edge-exp}
In what follows, we are going to use the following properties of $(a,\lambda)$-edge-expansion, which are easy to prove.
\begin{enumerate}
    \item \label{exp-prop:monotonic} If $a'\le a$, $\lambda'\ge\lambda$, and the graph $\Gamma$ is $(a,\lambda)$-edge-expanding, then $\Gamma$ is $(a',\lambda')$-edge-expanding.
    \item \label{exp-prop:subgraph}If a~graph $\Gamma=(V,E)$ is $(a,\lambda)$-edge-expanding, and $\Gamma'=(V,E')$ is a~subgraph of $\Gamma$ (i.e. $E'\subseteq E$), then $\Gamma'$ is also $(a,\lambda)$-edge-expanding.
    \item\label{exp-prop:disjoint-union} If graphs $\Gamma_1,\dots,\Gamma_m$ are $(a,\lambda)$-edge-expanding, then their disjoint union $\Gamma=\Gamma_1\sqcup\dots\sqcup \Gamma_m$ is also $(a,\lambda)$-edge-expanding.
    \item\label{exp-prop:same-vert-union} If graphs $\Gamma_1,\dots,\Gamma_m$ have the same set of vertices, and $\Gamma_i$ is $(a_i,\lambda_i)$-edge-expanding, then their union $\Gamma=\Gamma_1\cup\dots\cup\Gamma_m$ is $(\min_{i\in[m]}a_i,\sum_{i=1}^m\lambda_i)$-edge-expanding.
\end{enumerate}
\end{remark}
\begin{lemma}\label{lemma:minxy}
    Let $x_1,...,x_n\in\RR_+$, $y_1,...,y_n\in\RR_+$ be sequences of non-negative real numbers. Then
    $$\min_{i\in[n]} x_iy_i\le \bar{x}\bar{y},$$
    where $\bar{x}=\frac{1}{n}\sum_{i=1}^n x_i$, $\bar{y}= \frac{1}{n}\sum_{i=1}^n y_i$.
\end{lemma}
\begin{proof}
    By the Cauchy–Schwarz inequality for the vectors $(\sqrt{x}_i)_{i=1}^n$ and $(\sqrt{y}_i)_{i=1}^n$ we have
    $$\sum_{i=1}^n\sqrt{x_iy_i}\le \biggl(\sum_{i=1}^nx_i\biggr)^{1/2}\cdot\biggl(\sum_{i=1}^n y_i\biggr)^{1/2}=n\sqrt{\bar x\bar y}.$$
    Therefore 
    $\min_{i\in [n]}\sqrt{x_iy_i}\le \sqrt{\bar x\bar y}$, and finally we get 
    $$\min_{i\in [n]}x_iy_i=\bigl(\min_{i\in [n]}\sqrt{x_iy_i}\bigr)^2\le \bar x\bar y.$$
\end{proof}

\begin{lemma}\label{lemma:exp2}
    If a~$w$-regular graph $\Gamma$ is $(a,\lambda)$-edge-expanding, then the graph $\Gamma^2$ is $(a/w,2\lambda^2(1+\ln w))$-edge-expanding.
\end{lemma}
\begin{proof}
    Let $S,T\subseteq V(\Gamma)$ and $|S|,|T|\le a/w$. There are at most $w|S|\le a$ vertices adjacent to the vertices from $S$. Let $v_1,v_2,\dots$ be the sequence of the vertices incident to $S$ in the decreasing order of the number of length $2$ paths from $S$ to $T$ that goes through each of these vertices. Consider the set $U_j=\{v_1,\dots,v_j\}$ of size $j\le a$. By the edge-expansion property of the graph $\Gamma$ we have 
    $$|E_\Gamma(U_j,S)|\le\lambda\sqrt{j|S|},\qquad |E_\Gamma(U_j,T)|\le\lambda\sqrt{j|T|}.$$
    Hence, using Lemma \ref{lemma:minxy} with $n=j$, $x_i=|E_\Gamma(\{v_i\},S)|$, $y_i=|E_\Gamma(\{v_i\}, T)|$ the number of length 2 paths through the vertex $v_j$ is % at most 
    $$x_jy_j=\min_{i\in[j]}x_i y_i\le \frac{|E_\Gamma(U_j,S)|}{j}\cdot \frac{|E_\Gamma(U_j,T)|}{j}\le\frac{\lambda^2\sqrt{|S||T|}}{j}.$$
    On the other hand, the degree of each vertex in~$\Gamma$ is $w$, and hence the total number of pairs of edges incident to each vertex is $w^2$. Hence, if we let $\mu:=\lambda^2\sqrt{|S||T|}$, then the total number of length 2 paths from $S$ to $T$ can be estimated as 
    \begin{align*}
        |E_{\Gamma^2}(S,T)|&\le \sum_{j=1}^{\floorbr{\mu}}\min\left(w^2, \mu/j\right)=w^2\cdot \frac{\mu}{w^2}+\mu\sum_{j=\ceilbr{\mu/w^2}}^{\floorbr{\mu}}\frac{1}{j}\\
        &< \mu(2+\ln \mu-\ln(\mu/w^2))=2\mu(1+\ln w) = 2\lambda^2(1+\ln w)\sqrt{|S||T|}.
    \end{align*}
    Above we truncate the summation at $j = \floorbr{\mu}$ since for $j > \mu$ the number of length 2 paths going through the vertex $v_j$ is less or equal to $\min(w^2,\mu/j) < 1$, and therefore is equal to $0$. 
    Thus there exist at most $2\lambda^2(1+\ln w)\sqrt{|S||T|}$ edges from $S$ to $T$ in $\Gamma^2$, and $\Gamma^2$ is $(a/w,2\lambda^2(1+\ln w))$-edge-expanding.
\end{proof}

\subsection{Proof outline}\label{sc:proof-out}

In this subsection, we give some definitions and an~informal idea of the proof of our main results. Let $\hG = (\hat E, \hat V)$ be a~$G$-lift of some base graph $\Gamma$. We assume that $\hat\Gamma$ is a~$w$-regular $(a,\lambda)$-edge-expanding simple graph with $n$ vertices. For example, we can use an~infinite family of graphs $\bar X^{w-1,t}$ from Example~\ref{ex:Xpq}, where by Lemma~\ref{lemma:Xexp} we have $a = n/\sqrt{w}$, $\lambda = 8\sqrt{w}$.  

Let $h\in\FF_q^{r\times w}$, $h'\in\FF_q^{r'\times w}$ be some full rank matrices, and $\cA\in\mathfrak{T}_G(\hat\Gamma,h)$, $\cB\in\mathfrak{T}_G(\hat\Gamma,h')$ be the corresponding $G$-lifted Tanner codes: 
\[
\cA = \rbr{\FF_q \hat E\xrightarrow{\partial_A}\FF_q^{r} \hat  V},\qquad \cB = \rbr{\FF_q \hat E\xrightarrow{\partial_B}\FF_q^{r'} \hat V}.
\]
Now since $G$ acts freely on $\cA$ and $\cB$, we can consider their $G$-lifted product complex $\cC := \cA \otimes_G \cB^*$ over~$\F_q$ shown below:
\[
\underbrace{\FF_q \hat E\otimes_G \FF_q^{r'} \hat V}_{\cC_2}
\xrightarrow{\partial_2} \underbrace{\overbrace{\FF_q \hat E\otimes_G \FF_q \hat E}^{\cC_{F}}\oplus \overbrace{\FF_q^{r} \hat V\otimes_G\FF_q^{r'} \hat V}^{\cC_{V}}}_{\cC_1}
\xrightarrow{\partial_1} \underbrace{\FF_q^{r} \hat V\otimes_G \FF_2 \hat E}_{\cC_0}.
\]

It is convenient to represent $\cC$ as the the chain complex $\cC_\bullet(X, \cF)$, where $\cF$ is the~local system on $X := \hG \times_G \hG^*$. Since the poset $X$ has three levels: $X(2) = \Eh$, $X(1) = F \cup V$, and $X(0) = \Ev$,  it is not hard to see that $\cC_\bullet(X, \cF)$ has the following form:
\[
\underbrace{\FF_q^{r'} \Eh}_{\cC_2}
\xrightarrow{\partial_2} \underbrace{\overbrace{\FF_q  F}^{\cC_{F}}\oplus \overbrace{\FF_q^{r\times r'} V}^{\cC_{V}}}_{\cC_1}
\xrightarrow{\partial_1} \underbrace{\FF_q^{r} \Ev}_{\cC_0},
\]
where we identify $\F_q^r\otimes \F_q^{r'}$ with $\F_q^{r\times r'}$.

\begin{figure}
    \centering
    \begin{tikzpicture}
         \begin{scope}[xscale=-1]
            \filldraw[black!10!white] (1,1) -- (1,2) -- (2,2) -- (2,1) -- (1,1);
            \filldraw[black!10!white] (1,4) -- (1,6) -- (2,6) -- (2,4) -- (1,4);
            \filldraw[black!10!white] (4,1) -- (4,2) -- (6,2) -- (6,1) -- (4,1);
            \filldraw[black!10!white] (4,4) -- (4,6) -- (6,6) -- (6,4) -- (4,4);
            \draw[step=1] (0,0) grid (2,2);
            \draw[step=1] (3.99,0) grid (6,2);
            \draw[step=1,shift={(0.25,0)}] (3.99,0) grid (5.75,2);
            \draw[step=1,shift={(0.5,0)}] (3.99,0) grid (5.5,2);
            \draw[step=1,shift={(0.75,0)}] (3.99,0) grid (5.25,2);
            \draw[step=1] (0,3.99) grid (2,6);
            \draw[step=1,shift={(0,0.25)}] (0,3.99) grid (2,5.75);
            \draw[step=1,shift={(0,0.5)}] (0,3.99) grid (2,5.5);
            \draw[step=1,shift={(0,0.75)}] (0,3.99) grid (2,5.25);
            \draw[step=0.25] (3.99,3.99) grid (6,6);
            \draw[->] (1,2.3) --node[left]{$\id\otimes \partial_{\cB^*}$} (1,3.7);
            \draw[->] (5,2.3) --node[left]{$\id\otimes \partial_{\cB^*}$} (5,3.7);
            \draw[->] (3.7,1) --node[above]{$\partial_{\cA}\otimes\id $} (2.3,1);
            \draw[->] (3.7,5) --node[above]{$\partial_{\cA}\otimes\id $} (2.3,5);
            \node at (1.5,1.5) {$v$};
            \node[right] at (0,1) {$V$};
            \node[right] at (0,5) {$E_\uparrow$};
            \node[left] at (6,1) {$E_\rightarrow$};
            \node[left] at (6,5) {$F$};
        \end{scope}
    %\end{tikzpicture}\hspace{1cm}
    %    \begin{tikzpicture}
        \begin{scope}[shift={(8,0)},xscale=-1]
        \draw[step=1] (1,1) grid (2,2);
        \draw[step=1] (3.99,1) grid (6,2);
        \draw[step=1,shift={(0.25,0)}] (3.99,1) grid (5.75,2);
        \draw[step=1,shift={(0.5,0)}] (3.99,1) grid (5.5,2);
        \draw[step=1,shift={(0.75,0)}] (3.99,1) grid (5.25,2);
        \draw[step=1] (1,3.99) grid (2,6);
        \draw[step=1,shift={(0,0.25)}] (1,3.99) grid (2,5.75);
        \draw[step=1,shift={(0,0.5)}] (1,3.99) grid (2,5.5);
        \draw[step=1,shift={(0,0.75)}] (1,3.99) grid (2,5.25);
        \draw[step=0.25] (3.99,3.99) grid (6,6);
        \draw[->] (1.5,2.3) --node[left]{$\id\otimes h'^*$} (1.5,3.7);
        \draw[->] (5,2.3) --node[left]{$\id\otimes h'^*$} (5,3.7);
        \draw[->] (3.7,1.5) --node[above]{$h\otimes \id$} (2.3,1.5);
        \draw[->] (3.7,5) --node[above]{$h\otimes \id$} (2.3,5);
        \node at (1.5,1.5) {$v$};
        \node[right] at (1,5) {$E_{\uparrow v}$};
        \node[below] at (5,1) {$E_{\rightarrow v}$};
        \node[left] at (6,5) {$F_v$};
        \end{scope}
    \end{tikzpicture}
    \caption{High-level view of the tensor product complex $\cA\otimes \cB^*$, where $\cA\in \fT(D_w; h)$, $\cB\in \fT(D_w; h')$, and $w=8$ (on the right); and its part that corresponds to the elements from $X$ incident to the vertex $v$ (on the left).}
    \label{fig:local-syst}
\end{figure}
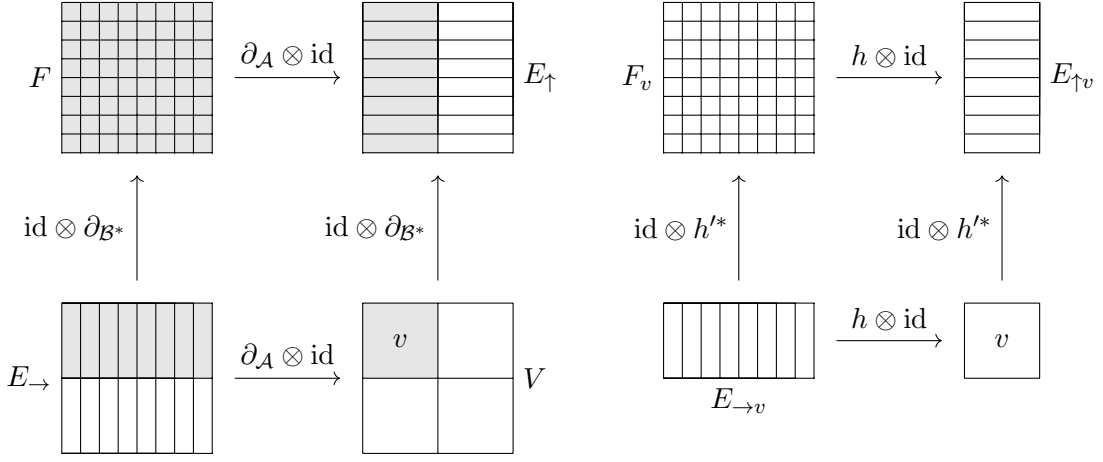

\begin{remark}
As we can see, $\cC_\bullet(X;\cF)$ gives us a~high-level representation of the complex~$\cC$. For example, on the left of~Fig.~\ref{fig:local-syst} you can find a~graphical representation of the tensor product complex $\cA\otimes \cB^*$, where $\cA\in \fT(D_w; h)$, $\cB\in \fT(D_w; h')$, and $w=8$.  For simplicity we consider in this example the tensor product instead of the $G$-lifted product. On the right of Fig.~\ref{fig:local-syst} you can see the ``part'' of this complex corresponding to the faces and edges incident to one particular vertex $v\in V$. 
\end{remark}

Now we consider the classical code $Z_2(\cC) = \ker \partial_2$ and the quantum code $\cQ(\cC) := \cQ(\partial_1,\partial_2^*)$, and show that for some sufficiently large number $w$ we can choose the matrices $h,h'$ such that $Z_2(\cC)$ and $\cQ(\cC)$ satisfy the requirements of Theorems~\ref{th:LTC} and \ref{th:qLDPC}, respectively.
The most difficult part of the proof is to show that $Z_2(\cC)$ is locally testable, and $\cQ(\cC)$ has linear minimum distance. However, from Lemma~\ref{lemma:dist-from-lm} it easily follows that if $\cC$ has the~locally minimal distance $\dLM^{(1)}(\cC) = \Theta(n)$ as $n\to\infty$, then both $Z_2(\cC)$ and $\cQ(\cC)$ have the desired properties.
Therefore we need to show that for every locally minimal 1-cycle $c\in Z_1(\cC)$ such that $|c| = o(n)$ as $n\to\infty$ we have $c=0$, where $|c| = \wt_X(c)$ is the block weight of $c$.  

Let us fix some non-zero locally minimal 1-cycle $c = \sum_{x\in X(1)} c_x x\in Z_1(\cC)$. Hence we have $c \ne 0$ and $\partial c=0$. We have $c = c_F + c_V$, where $c_F := c|_F$ and $c_V := c|_V$. 
Below we give a~number of important definitions used in the rest of the paper. Note that some of them depend on the fixed $1$-cycle $c$. However, for brevity, we usually do not mention $c$.

\begin{definition}
An~element $x\in X(1)$ (a~vertex or a~face) is called \emph{active} if $c_x\ne 0$.
A~vertical edge $e\in \Ev$ is called \emph{active} if it is incident to an~active vertex or an active~face. Furthermore, $e$ is called \emph{face-active} if it is not incident to any active vertex (only to an~active face).
\end{definition}

We also need another type of vertices we call \emph{labeled} that include active vertices as a~special case. However, the number of the labeled vertices is $O(|c|)$, and we can still use the expansion properties of the graphs involved in the proof.
We define the set of labeled vertices as the minimal set of vertices such that:
\begin{enumerate}
    \item every active vertex is labeled;
    \item every vertex of a~face-active edge adjacent to at least $m$ labeled vertices is labeled.
\end{enumerate}

We also consider $2$ types of labeled vertices:
\begin{enumerate}
    \item a~vertex is called \emph{$m$-edge-expanding} if there are at least $m$~edges connecting it to the labeled vertices in $\Lambda$;
    \item a~vertex is called \emph{$s$-face-expanding} if there are at least $s$~edges connecting it to the labeled vertices in $\Lambda^2$.
\end{enumerate}

%Sometimes, when the parameters $m$ and $s$ are clear from context we omit them and just say that a vertex is edge-expanding or face-expanding.

In the proof outlined below, we consider classical codes that are duals of the product codes. In~Subsection~\ref{sc:local-exp} we define a~special property of such codes called \emph{$(s,m,\beta)$-product-expansion}. Informally speaking, this property corresponds to the local expansion in the complex~$\cC$. In some sense, it plays a~role similar to the role of the minimal distance of the local codes in the classical proof of Sipser and Spielman from~\cite{Sipser:1996}, where it is shown that expander codes have linear minimum distances. 

Fix $\eps := 1/6$, and put $m := w^{1/2+\eps}$, $s := w^{1+\eps}$. From Lemma \ref{lemma:rand-gh} it follows that we can find a~sufficiently large number $w$ and choose matrices $h$ and $h'$ with $w$~columns such that both pairs $(\im h^*, \ker h')$ and $(\ker h, \im h'^*)$ are $(s,2m,\beta)$-product-expanding.

In the proof, we often use expansion properties of the graphs $\Lambda$ and $\Lambda^2$, where $\Lambda := \hG \mathbin{\Box_G} \hG$ is the graph defined in Subsection~\ref{sc:prod-graph}. 
Using the edge expansion of $\hG$ we show in~Lemma~\ref{lemma:expG'} that $\Lambda$ is $(\Theta(n), \lambda')$-edge-expanding where $\lambda'=\Theta(w^{1/2})$. We also show in Lemma~\ref{lemma:expG''} that
$\Lambda^2$ is $(\Theta(n), \lambda'')$-edge-expanding, where $\lambda''=\Theta(w\ln w)$.

Suppose that $|c| = o(n)$, i.e., the number of the active vertices and faces is relatively small.
Then the proof by contradiction contains the following steps.

\begin{enumerate}
    \item Since each labeled vertex is either active itself or incident to an~active face, then the number of the labeled vertices is $O(|c|)=o(n)$. Hence we can use the expansion properties of the graphs $\Lambda$ and $\Lambda^2$ for subsets of labeled vertices.
    \item Using the expansion properties of the graph~$\hG$ it is possible to show that each face-active edge is incident to a~labeled vertex (Lemma \ref{lemma:lbl-edge}).
    \item Note that, by definition, each labeled non-active vertex is $m$-edge-expanding. 
    \item Using local minimality of $c$ and $(s,2m,\beta)$-product-expansion of $(\im h,\ker h')$ we can show that each active vertex is either $m$-edge-expanding or $s$-face-expanding (Lemma \ref{lemma:AB*-vert-cases}---the key lemma).
    \item From the previous 2 items we have that each labeled vertex is either $m$-edge-expanding or $s$-face-expanding (Corollary \ref{col:edge-face-exp}).
    \item Thus using the expansion properties of $\Lambda$ and $\Lambda^2$ we obtain a~contradiction (Lemma \ref{lemma:AB*vzero}):
    \begin{enumerate}
        \item from the $(\Theta(n),\lambda')$-edge expansion of $\Lambda$ we obtain that the ratio of the $m$-edge-expanding labeled vertices is $\Theta(\lambda'/m)<1/2$ for a~sufficiently large $w$ since $\lambda'=\Theta(w^{1/2})$ and $m=\Theta(w^{1/2+\eps})$;
        \item from the $(\Theta(n),\lambda'')$-edge expansion of $\Lambda^2$ we obtain that the ratio of the $s$-face-expanding labeled vertices is $\Theta(\lambda''/s)<1/2$ for a~sufficiently large $w$ since $\lambda''=\Theta(w\ln w)$ and $s=\Theta(w^{1+\eps})$;
        \item the ratio of the labeled vertices that are either $m$-edge-expanding or $s$-face expanding is less than 1, which can be true only when the $1$-cycle $c$ is zero.
    \end{enumerate}
\end{enumerate}
Since we obtained a~contradiction, we have that $|c| = \Theta(n)$, i.e., the locally minimal distance $\dLM^{(1)}(\cC) = \Theta(n)$ as $n\to\infty$, which is in turn of the same order as the length of the classical or quantum codes obtained from the chain complex $\cC$. Hence by Lemma~\ref{lemma:dist-from-lm} we get what we need.

\subsection{Local expansion}\label{sc:local-exp}
In this section, we consider the dual code to the classical product code~\cite{Wolf:1965, Chien:1973} and study its expansion properties\footnote{The property we consider is similar to the \emph{robust testability} property of tensor product codes, often studied in the literature on~LTCs~\cite{Ben-Sasson:2006,Dinur:2006}.}. Such codes are related to the local expansion properties of the $G$-lifted product of two Tanner codes.  Let $\ker h\subseteq \FF_q^{w}$ and $\ker h' \subseteq \FF_q^{w}$ be linear codes with parity-check matrices $h$ and $h'$ respectively. Consider the code $\cC=\ker (h\otimes h')\subseteq \F_q^w\otimes \F_q^w$. We will identify the elements of $\F_q^w\otimes \F_q^w$ with the corresponding matrices $x=(x^i_j)_{i,j=1}^w\in \F_q^{w\times w}$, where $x^i$ is the $i$-th row, and $x_j$ is the $j$-th column. Note that the matrix $h \otimes h'$ is also a~generator matrix for the product of the codes $(\ker h)^\perp = \im h^*$ and $(\ker h')^\perp = \im {h'}^*$ with the generator matrices $h, h'$ respectively, which means that $\cC$ is the dual to this product code.
\begin{remark}
Using matrix representation, it is not hard to check that the codewords of $\cC$ are precisely the matrices $x\in \F_q^{w\times w}$ such that $h' x h^* = 0$. Therefore if $x\in \cC$ then every row of the matrix $s_{\uparrow} := h' x$ is a~codeword from $\ker h$ and every column of the matrix $s_{\rightarrow} := x h^*$ is a~codeword from $\ker h'$ (see Fig.~\ref{fig:hxh'code}).
\end{remark}

\begin{definition}
A~codeword $x\in \cC=\ker (h\otimes h')$ is called \emph{$\Delta$-minimal} if the following conditions hold:
\begin{enumerate}
    \item $\wt(x^i)\le d(x^i, \ker h)+\Delta$ for all $i\in [w]$,
    \item $\wt(x_j)\le d(x_j, \ker h')+\Delta$ for all $j\in [w]$,
\end{enumerate}
which means that we cannot decrease the weight of the matrix $x$ by more than $\Delta$ if we add any~codeword from $\ker h$ (resp. $\ker h')$ to some row (resp. column) of $x$.
A~pair of codes $(\ker h, \ker h')$ is called \emph{$(s,m,\beta)$-product-expanding} if for each non-zero $\beta w$-minimal codeword $x\in \cC$ and for each $A, B\subseteq [w]$ such that $|A|,|B|\ge w-m$ we have $\wt_{A\times B}(x)\ge s$, where $\wt_{A\times B}(x):= \wt(x|_{A\times B})$.
\end{definition}
In this section, we often use the following short-hand notations: $x^I := x|_{I \times [w]}$, $x_J := x|_{[w]\times J}$, and $x_J^I := x|_{I\times J}$, where $x\in \F_q^{w\times w}$, $I,J\subseteq [w]$.
\begin{lemma}\label{lemma:nonzero-block}
    Let $h\in \FF_q^{r\times w}$, $h'\in \FF_q^{r'\times w}$ be parity-check matrices such that $\min(d(\ker h), d(\ker h'))\ge d$, and $x=(x^i_j)_{i,j=1}^w\in \F_q^{w\times w}$ be a $d/3$-minimal %non-zero 
    codeword of $\ker (h\otimes h')$.
    If there exist $A,B\subseteq [w]$ such that $|A|> w-d/3,|B|\ge w-d+1$ and $x_A^B= 0$ or $x_B^A= 0$, then $x=0$.
\end{lemma}
\begin{proof}
    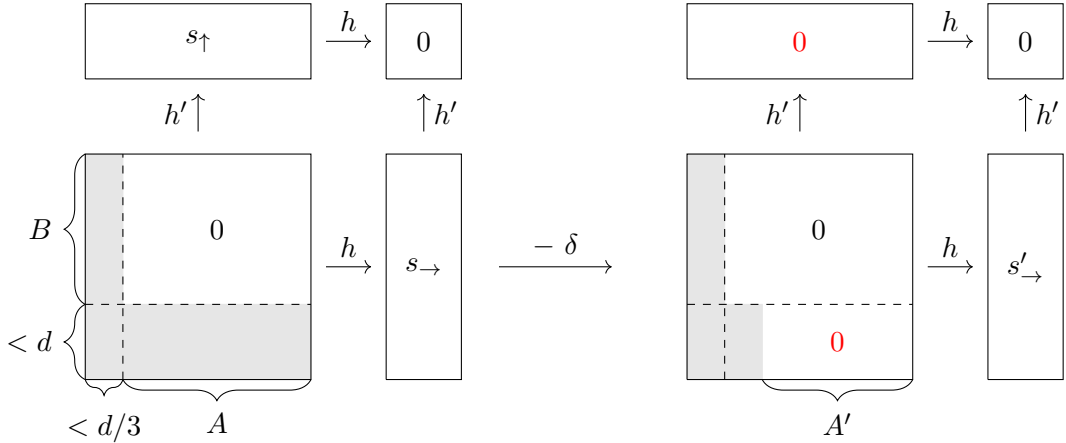
\begin{figure}[ht]
        \centering
        \begin{tikzpicture}
                \filldraw[black!10!white] (0,0) -- (3,0) -- (3,1) -- (0.5,1) -- (0.5,3)--(0,3)--(0,0);
                \draw (0,0)--(3,0)--(3,3)--(0,3)--(0,0);
                \draw (4,0)--(5,0)--(5,3)--(4,3)--(4,0);
                \draw (0,4)--(0,5)--(3,5)--(3,4)--(0,4);
                \draw (4,4)--(4,5)--(5,5)--(5,4)--(4,4);
                \draw[dashed] (0,1)--(3,1);
                \draw[dashed] (0.5,0)--(0.5,3);
                % \filldraw[black!10!white] (1,4) -- (1,6) -- (2,6) -- (2,4) -- (1,4);
                % \filldraw[black!10!white] (4,1) -- (4,2) -- (6,2) -- (6,1) -- (4,1);
                % \filldraw[black!10!white] (4,4) -- (4,6) -- (6,6) -- (6,4) -- (4,4);
                \draw[->] (1.5,3.3) --node[left]{$h'$} (1.5,3.8);
                \draw[->] (4.5,3.3) --node[right]{$h'$} (4.5,3.8);
                \draw[->] (3.2,1.5) --node[above]{$h$} (3.8,1.5);
                \draw[->] (3.2,4.5) --node[above]{$h$} (3.8,4.5);
                \node[] at (4.5,1.5) {$s_\rightarrow$};
                \node[] at (1.5,4.5) {$s_\uparrow$};
                \node[] at (4.5,4.5) {0};
                \node[] at (1.75,2) {0};
                \draw [decorate,decoration={brace,amplitude=8pt}]
                    (0,1) -- (0,3) node [midway,left,xshift=-0.8em] {$B$};
                \draw [decorate,decoration={brace,amplitude=8pt}]
                    (0,0) -- (0,1) node [midway,left,xshift=-0.8em] {$< d$};
                \draw [decorate,decoration={brace,amplitude=8pt}]
                    (3,0) -- (0.5,0) node [midway,below,yshift=-0.8em] {$A$};
                \draw [decorate,decoration={brace,amplitude=6pt}]
                    (0.5,0) -- (0,0) node [midway,below,yshift=-0.8em] {$< d/3$};
    
                \draw[->] (5.5,1.5)--node[above]{$-\ \delta$}(7,1.5);
                \begin{scope}[xshift=8cm]
                    \filldraw[black!10!white] (0,0) -- (0,3) -- (0.5,3) -- (0.5,1) -- (1,1)--(1,0)--(0,0);
                    \draw (0,0)--(3,0)--(3,3)--(0,3)--(0,0);
                    \draw (4,0)--(5,0)--(5,3)--(4,3)--(4,0);
                    \draw (0,4)--(0,5)--(3,5)--(3,4)--(0,4);
                    \draw (4,4)--(4,5)--(5,5)--(5,4)--(4,4);
                    \draw[dashed] (0,1)--(3,1);
                    \draw[dashed] (0.5,0)--(0.5,3);
                    % \filldraw[black!10!white] (1,4) -- (1,6) -- (2,6) -- (2,4) -- (1,4);
                    % \filldraw[black!10!white] (4,1) -- (4,2) -- (6,2) -- (6,1) -- (4,1);
                    % \filldraw[black!10!white] (4,4) -- (4,6) -- (6,6) -- (6,4) -- (4,4);
                    \draw[->] (1.5,3.3) --node[left]{$h'$} (1.5,3.8);
                    \draw[->] (4.5,3.3) --node[right]{$h'$} (4.5,3.8);
                    \draw[->] (3.2,1.5) --node[above]{$h$} (3.8,1.5);
                    \draw[->] (3.2,4.5) --node[above]{$h$} (3.8,4.5);
                    \node[] at (4.5,1.5) {$s'_\rightarrow$};
                    \node[red] at (1.5,4.5) {0};
                    \node[] at (4.5,4.5) {0};
                    \node[] at (1.75,2) {0};
                    \node[red] at (2,0.5) {0};
                    \draw [decorate,decoration={brace,amplitude=8pt}]
                        (3,0) -- (1,0) node [midway,below,yshift=-0.8em] {$A'$};
                \end{scope}
        \end{tikzpicture}
        \caption{Idea of the proof.}
        \label{fig:hxh'code}
    \end{figure}

    Suppose that $x_A^B=0$ for some $A,B\subseteq [w]$ such that $|A|> w-d/3,|B|\ge w-d+1$ (shown on the left of Fig. \ref{fig:hxh'code}). 
    Since $|A|> w-d(\ker h)$, $|B|> w-d(\ker h')$, there exist information\footnote{An~\emph{information set} for a~linear code $\cC\subseteq\F_q^n$ is a smallest by inclusion index set $I\subseteq [n]$ such that for every $c\in\cC$ if $c|_I = 0$ then $c = 0$. It is clear that for every $S\subseteq [n]$ such that $\abs{S} > n - d(\cC)$ if for some codeword $c\in \cC$ we have $c|_S = 0$ then $c=0$. Hence there should exist an~information set $I\subseteq S$.} sets $A'\subseteq A$, $B'\subseteq B$ of the codes $\ker h$ and $\ker h'$ respectively. 
    Let $g$ be the generator matrix in systematic form\footnote{A~generator matrix $g$ is in \emph{systematic form} for an~information set $I$ if the submatrix $g_{I}$ is the identity matrix.} for the information set $A'$. 
    Consider matrices $\delta := x_{A'}g$ and $x' := x-\delta$. 
    
    Let us show that $x'_A=0$. Since $\delta_{A'}=x_{A'}$, we have $x'_{A'}=0$. 
    On the other hand, $\delta h^*= x_{A'}g h^* = 0$, and hence 
    \[
    h'x' h^* = h' (x - \delta) h^* = h' x h^* - h' \delta h^* = 0.
    \] 
    Therefore $(h'x') h^* = 0$, and each row of $h'x'$ is a~codeword from $\ker h$. 
    The condition $x'_{A'} = 0$ implies that $(h'x')_{A'} = 0$, and since $A'$ is an~information set of $\ker h$, we get $h'x' = 0$, which means that every column of $x'$ is a codeword of $\ker h'$ (shown on the right of Fig. \ref{fig:hxh'code}). 
    %Note that we also have $x^{B'}_{A}=0$, and thus $x^{B'}_{A'}=0$. Hence 
    Since $x^{B'}_{A}=0$ and $x^{B'}_{A'}=0$, we have $\delta^{B'}=x^{B'}_{A'}g=0$, and therefore ${x'}^{B'}_{A}=x^{B'}_{A}-\delta^{B'}_A=0$. Now since $B'$ is an~information set of $\ker h'$, ${x'}_{A}^{B'}=0$, and $h'x'_{A}=0$, we have $x'_{A}=0$. 
    
    Suppose $\delta\ne 0$. In this case, there exists $i\in[w]$ such that $\delta^i\ne 0$. Taking into account that $\delta^i\in \ker h$, we obtain $\wt(\delta^i)\ge d$. But since $x'_{A} = 0$, and $|A|> w-d/3$, we have $\wt({x'}^i)\le w-|A|< d/3$, and thus $\wt(x^i)\ge \wt(\delta^i)-\wt({x'}^i)> 2d/3> \wt({x'}^i)+d/3$, which contradicts the $d/3$-minimality of $x$. 
    Thus $\delta=0$, which implies that $x'=x$ and $h'x=0$, hence $d(x_j,\ker h')=0$ for all $j\in[w]$.
    By $d/3$-minimality of $x$ we have $\wt(x_j)\le d/3 < d(\ker h')$, therefore $x_j=0$ for all $j\in [w]$, i.e. $x=0$. Hence we showed that $x_A^B= 0$ implies $x=0$. Thus to prove the lemma it remains to show that $x_B^A= 0$ also implies $x=0$, which can be shown in a similar way.
\end{proof}

\begin{lemma}\label{lemma:h1x-rank}
    Let $h\in\FF_q^{r\times w}$, $h'\in\FF_q^{r'\times w}$, 
    $d=\min(d(\ker h),d(\ker h'))$, $m\le d/6$. Suppose $x\in\ker (h\otimes h')$ is a $d/3$-minimal non-zero codeword such that $\wt_{A\times B}(x) < s$ for some $A,B\subseteq [w]$, $|A|=|B|=w-m$. 
    Then $\rk h' x\ge \frac{5}{36}\cdot\frac{d^2}{s}$.
\end{lemma}
\begin{proof}
    By Lemma \ref{lemma:nonzero-block} each submatrix of $w-d+1$ columns of $x$ must have at least $d/3$ nonzero rows, and each submatrix of $w-d+1$ rows of $x$ must have at least $d/3$ nonzero columns. In particular, $x$ has at least $d$ nonzero columns and at least $d$ nonzero rows. Indeed, otherwise we would have at least $w-d+1$ zero rows or columns, which contradicts what we said earlier.

    Let $k=\rk h' x$, and $\{h'\tilde{x}_1,\dots,h'\tilde{x}_k\}$ be a~generating set for the column space of $h' x$ with the minimal total weight $\wt_A(\tilde{x}):=\wt_A(\tilde{x}_1)+\dots+\wt_A(\tilde{x}_k)$, where $\tilde{x}$ is a~matrix with the columns $\tilde{x}_1,\dots,\tilde{x}_k$.
    Without loss of generality we assume that $\wt_A(\tilde{x}_1)\le \cdots\le \wt_A(\tilde{x}_k)$. 
    %By $\tilde{X}$ we denote the linear span $\abr{\tilde{x}_1,\dots,\tilde{x}_k}$.
    
    Let us show that $|\bigcup_{j=1}^k\supp\tilde{x}_j|\ge d/3$.  Denote $U=\bigcup_{j=1}^k\supp\tilde{x}_j$. Suppose $|U|<d/3$. Since $|\bigcup_{i=1}^w \supp x_i|\ge d$, there is a column $x_i$ such that $\supp x_i\not\subseteq U$, hence $x_i\not\in \im\tilde{x}$. However $h'x_i\in \im h' \tilde{x}$, and hence there exists some $y\in\ker h'\setminus\{0\}$ such that $x_i+y\in \im\tilde{x}$. Since $\supp (x_i+y)\subseteq U$, we have $\wt(x_i+y)<d/3$ and $\wt(x_i)\ge \wt(y)-\wt(x_i+y)>2d/3>\wt(x_i+y)+d/3$, which contradicts the $d/3$-minimality of $x$, and hence our assumption is wrong, and $|U|\ge d/3$.
    
    We have
    $$\sum_{i=1}^k\wt_A(\tilde{x}_i)
        \ge\Bigl|\bigcup_{i=1}^k(\supp \tilde{x}_i\cap A)\Bigr| 
        =|U\cap A|=|U\setminus ([w]\setminus A)|\ge |U|-\underbrace{(w-|A|)}_m\ge \frac{d}{3}-m\ge \frac{d}{6}.$$
    Let $k'$ be the minimal number such that $\sum_{j=1}^{k'}\wt_A(\tilde{x}_j) \ge d/6$, then $\wt_A(\tilde{x}_{k'})\ge \frac{d}{6k'}\ge \frac{d}{6k}$. Put $U_0=\bigcup_{j=1}^{k'-1}\supp \tilde{x}_j$. Each column $x_i$ is uniquely represented as $x_i=y_i+\tilde{x}a_i$ where $y_i\in\ker h'$, $a_i\in\FF_q^k$. If $\supp a_i\subseteq [k'-1]$, then 
    $$\wt(\tilde{x}a_i)\le m+\wt_A(\tilde{x}a_i)\le m+\sum_{j=1}^{k'-1}\wt_A(\tilde{x}_j)<d/3,$$
    and hence $y_i=0$, otherwise $\wt(x_i)\ge d-\wt(\tilde{x}a_i)> 2d/3\ge \wt(x_i+y_i)+d/3$ which contradicts the $d/3$-minimality of $x$. Therefore $\supp x_i\subseteq U_0$.

    Since every $w-d+1$ columns of~$x$ have at least $d/3$ nonzero rows, there are at most $w-d$ columns $x_i$ such that $\supp a_i\subseteq [k'-1]$. 
    Hence there exists a~set $C\subseteq [w]$ of size $d$ such that $\max(\supp a_i)\ge k'$ for all $i\in C$. 
    Note that if $j=\max(\supp a_i)$, then $\wt_A(x_i)\ge \wt_A(\tilde{x}_j)$. Indeed, otherwise we can replace $\tilde{x}_j$ by $x_i$ and reduce $\wt_A(\tilde{x})$, which contradicts the minimality of $\wt_A(\tilde{x})$.
    Hence $\wt(x_i)\ge \wt_A(\tilde{x}_{k_1})\ge \frac{d}{6k}$ for all $i\in C$, and
    therefore
    $$\wt_{A\times B}(x)\ge \wt_{A\times (B\cap C)}(x)=\sum_{i\in B\cap C}\wt_A(x_i)\ge \frac{d|B\cap C|}{6k}.$$
    Since $|B\cap C|=|C\setminus ([w]\setminus B)|\ge |C|-(w-|B|)=d-m$, we have
    \begin{equation}
        k\ge \frac{d|B\cap C|}{6\wt_{A\times B}(x)}\ge \frac{d(d-m)}{6s}\ge \frac{5}{36}\cdot\frac{d^2}{s},
    \end{equation}
    and the lemma is proved.
\end{proof}

\begin{lemma}\label{lemma:rand-gh}
Let $\eps\in(0,1/4)$, $\alpha>0$, $\gamma>0$, $R_1\in(0,1)$, $R_2\in(0,1)$. Then there exist $\beta>0$ and $\delta>0$ such that for random\footnote{We suppose that the entries of the both matrices are chosen uniformly and independently at random from~$\F_q$.} matrices $h\in \FF_q^{\floorbr{R_1w}\times w}$, $g'\in \F_q^{\floorbr{R_2 w}\times w}$ the following three conditions hold with high probability as $w\to\infty$\textsc{:}
\begin{enumerate}
    \item $\min(d(\ker h), d(\im g'^*))\ge \delta w$;
    \item the matrices $h$ and $g'$ have full rank;
    \item the pair of codes $(\ker h, \im g'^*)$ is $(\alpha w^{1+\eps}, \gamma w^{1/2+\eps}, \beta)$-product-expanding.
\end{enumerate}   
\end{lemma}
\begin{proof}
    Let us start the proof by saying that the first two conditions follows from the probabilistic proof of the asymptotic Gilbert-Varshamov bound\footnote{Note that the probabilistic proof of the Gilbert–Varshamov bound can be used with a~random code defined either by a~random parity-check matrix or a~random generator matrix. See~\cite{Barg&Forney:2002} for a~good review of this bound.}. Indeed, it is enough to choose $\delta\le (q-1)/q$ such that 
    $H_q(\delta)=\min(R_1/2, (1-R_2)/2)$, where 
    \[H_q(x):= x\log_q(q-1)-x\log_q x-(1-x)\log_q(1-x)\] is the \emph{$q$-ary entropy function}.
    
    Now put $r_1:=\floorbr{R_1w}$, $r_2:=\floorbr{R_2 w}$, $d:= \delta w$, $\beta:=\delta/3$, and let us fix a~full-rank matrix $g'\in \F_q^{r_2\times w}$ such that $d(\im g'^*)\ge d$. In the rest of the proof, we will consider all the probabilities conditioned on this choice of $g'$. 
    
    Let $h'\in\FF_q^{(w-r_2)\times w}$ be a~parity-check matrix of the code $\im g'^*$, and consider the  code   $\cC:=\ker(h\otimes h')$.
    The entries of the matrix $h$ are independent uniformly distributed elements of $\FF_q$. Now we estimate the probability that the code $\cC$ has a codeword of some particular form. Recall that we interpret elements of $\FF_q^w\otimes \FF_q^w$ as $w\times w$ matrices over $\FF_q$. In this interpretation every $x\in\cC$ satisfies the condition $h'xh^*=0$. Hence, for $x\in\cC$ we have
    \begin{equation}\label{eqn:h2}
        0=h'xh^*=
        s_\uparrow h^*
    \end{equation}
    where $s_\uparrow=h' x$. 
    Let us remind that for matrix $u\in\FF_q^{a\times b}$ by $u_i$ we denote the $i$-th column of $u$ and by $u^j$ we denote the $j$-th row of $u$.
    
    When $h'$ and $x$ are fixed, then \eqref{eqn:h2} defines a system of linear equations on the elements 
    of the matrix $h$. To estimate the number of solutions we need to estimate the rank of this system. For all $j\in [r_2]$ we have $h^j\in\ker s_\uparrow$. Hence, the probability that the equation \eqref{eqn:h2} satisfied is $q^{-r_2\rk s_\uparrow}$.

    Put $\beta=\frac{d}{3w}=\delta/3$, $m=\gamma w^{1/2+\eps}$ and suppose $w$ is sufficiently large such that $m\le d/6$. 
    By Lemma \ref{lemma:h1x-rank} for every $\beta w$-minimal non-zero codeword $x\in\cC$  
    such that $\wt(x)\le \alpha w^{1+\eps}$ we have $\rk h' x\ge \frac{5}{36}\cdot \frac{d^2}{\alpha w^{1+\eps}}=c_1w^{1-\eps}$ where $c_1=\frac{5\delta^2}{36\alpha}$.
    So, to summarize, we proved that if $(\ker h,\ker h')$ is not $(\alpha w^{1+\eps},m,\beta)$-product-expanding, and $m\le \delta w/6$, then one of the following three cases is true:
    \begin{enumerate}
        \item $d(\ker h)<\delta w$;
        \item $d(\ker h')<\delta w$;
        \item there exist subsets $A,B\subseteq [w]$, $|A|=|B|=w-m$ and a~matrix $x\in \FF_q^{w\times w}$ such that $\wt(x|_{A\times B})<\alpha w^{1+\eps}$, $\rk h' x\ge c_1 w^{1-\eps}$, and equation~\eqref{eqn:h2} is satisfied.
    \end{enumerate}
    For every $i\in \{1,2,3\}$ let $p_i$ be the probability that the $i$-th case above holds if we choose the matrices $h$ and $g'$ uniformly at random. Recall that we have already chosen $\delta$ such that $p_1\to 0$ and $p_2\to 0$ as $w\to\infty$. Hence to complete the proof we also need to show that $p_3\to 0$ as $w\to\infty$. To estimate the probability $p_3$ we need to estimate the number of ways one can choose the matrix $x$ such that the third case above holds. 
    It is clear that we have
    \begin{enumerate}
       \item $\binom{w}{m}^2<w^{2m}$ choices for the subsets $A$ and $B$;
       \item less than $q^{2mw}$ choices for the elements of $x$ at the positions from $[w]\times [w]\setminus A\times B$;
       \item less than $\binom{w^2}{\alpha w^{1+\eps}}q^{\alpha w^{1+\eps}}<(qw)^{2\alpha w^{1+\eps}}$ choices for the elements of $x$ at the positions from~$A\times B$.
    \end{enumerate}
    Totally, we have $N$ choices of vector $x$, where
    $$\log_q N\le \log_q \rbr{q^{2mw} w^{2m}(qw)^{2\alpha w^{1+\eps}}} = 2\gamma w^{3/2+\eps}+2\gamma w^{1/2+\eps}\log_q w+2\alpha w^{1+\eps}(1+\log_q w).$$
    For each choice of the vector $x$ the probability that \eqref{eqn:h2} is satisfied equals to $q^{-r_2\rk h x}<q^{-c_1 r_2 w^{1-\eps}}=q^{-c_2 w^{2-\eps}}$ where $c_2=c_1(1-R_2)$.
    Thus, by the union bound, the probability $p_3$ is bounded from above by $N q^{-c_2 w^{2-\eps}}$, and we get 
    \[
    \log_q p_3\le \log_q N - c_2 w^{2-\eps}\le \gamma w^{3/2+\eps} + 2\gamma w^{1/2+\eps}\log_q w+2\alpha w^{1+\eps}(1+\log_q w)\;\underbrace{-\;c_2 w^{2-\eps}}_{\mbox{main term}}.
    \]
    It is easy to see that $\log_q p_3\to -\infty$ as $w\to\infty$ for any constants $\eps<1/4$, $\alpha>0$, $\gamma>0$.
    If $w$ is large enough then $m=\gamma w^{1/2+\eps}<\delta w/6$. Hence
    the probability $p$ that $(\ker h,\ker h')$ is not $(\alpha w^{1+\eps},m,\beta)$-product-expanding is bounded from above by $p_1+p_2+p_3\to 0$ as $w\to\infty$, and the lemma is proved.
\end{proof}

\subsection{Global expansion}\label{sc:global-exp}
\renewcommand{\mark}[1]{\textcolor{red}{#1}}

In this subsection, the graph $\Lambda$ is the graph from Subsection~\ref{sc:prod-graph}.

\begin{lemma}\label{lemma:expG'}
    The graph $\Lambda$ is $(a,2\lambda)$-edge-expanding.
\end{lemma}
\begin{proof}
    Since $E=E_\rightarrow\cup E_\uparrow$, we can split the graph~$\Lambda$ as $\Lambda=\Lambda_\rightarrow\cup\Lambda_\uparrow$, where $$\Lambda_\rightarrow:=V\cup E_\rightarrow=\{x\cdot g\cdot y\mid x\in V(\Gamma)\cup E(\Gamma),y\in V(\Gamma), g\in G\},$$ 
    $$\Lambda_\uparrow:=V\cup E_\uparrow=\{x\cdot g\cdot y\mid x\in V(\Gamma),y\in V(\Gamma)\cup E(\Gamma), g\in G\}.$$
    In terms of graphs, $\Lambda_\rightarrow$ is the subgraph of $\Lambda$ containing only horizontal edges, and $\Lambda_\uparrow$  is the subgraph of $\Lambda$ containing only vertical edges. It is easy to see that 
    $$\Lambda_\rightarrow=\bigsqcup_{y\in V(\Gamma)}\Lambda_\rightarrow^{(y)},\qquad \Lambda_\uparrow=\bigsqcup_{x\in V(\Gamma)}\Lambda_\uparrow^{(x)},$$
    where
    $$\Lambda_\rightarrow^{(y)}=\{x\cdot g\cdot y\mid x\in V(\Gamma)\cup E(\Gamma), g\in G\},$$
    $$\Lambda_\uparrow^{(x)}=\{x\cdot g\cdot y\mid y\in V(\Gamma)\cup E(\Gamma), g\in G\}.$$
    Since the graphs $\Lambda_\uparrow^{(x)}$ and $\Lambda_\rightarrow^{(y)}$ are isomorphic to $\hat \Gamma$, they are $(a,\lambda)$-edge-expanding. Hence, by property~\ref{exp-prop:disjoint-union} of the edge expansion (see Remark~\ref{rm:edge-exp}), their disjoint unions $\Lambda_\rightarrow$ and $\Lambda_\uparrow$ have the same edge expansion. Therefore by property \ref{exp-prop:same-vert-union} of the edge expansion their union $\Lambda$ is $(a,2\lambda)$-edge-expanding.
\end{proof}

\begin{lemma}\label{lemma:expG''}
    The graph $\Lambda^2$ is $(a/2 w, 8\lambda^2(\ln w+2))$-edge-expanding.
\end{lemma}
\begin{proof}
    By Lemma \ref{lemma:expG'} graph $\Lambda$ is $(a,2\lambda)$-edge-expanding. From the definition of $\Lambda$ it is easy to see that $\Lambda$ is a~$2w$-regular graph. 
    Hence by Lemma \ref{lemma:exp2} the graph ${\Lambda}^2$ is $(a/2 w,8\lambda^2(1+\ln(2w)))$-edge-expanding. 
Since $\ln(2w)< \ln w+1$, we obtain the assertion of the lemma.
\end{proof}

\newcommand{\gve}{\Lambda_{\Box}}

In the rest of this subsection, we assume that $c$ is some fixed locally minimal $1$-cycle in the complex $\cC_\bullet(X;\cF)$ from  Subsection~\ref{sc:proof-out}.

\begin{lemma}\label{lemma:face-active-edge}
    If $\partial c=0$, then each face-active vertical edge is incident to at least $d(\ker h)$ active faces.
\end{lemma}
\begin{proof}
    Consider a~face-active vertical edge $e$. Then $F_e$ is the set of faces incident to $e$, and $V_e$ is the set of (two) vertices incident to $e$. Since $e$ is face-active, $c|_{V_e}=0$ but $c|_{F_e}\ne 0$.
    Since $(\partial c)|_e$ depends only on $c|_{F_e}$ and $c|_{V_e}$, then using~(\ref{eq:boundary-res}) we have
    \[
    0 = (\partial c)|_e = (\partial (c|_{F_e} + \underbrace{c|_{V_e}}_{=0}))|_e 
    = \partial_{F_e\to e} (c|_{F_e}) 
    \]
Since $\partial_{F_e\to e}\sim h$, $c|_{F_e}\ne 0$, and $\partial_{F_e\to e} (c|_{F_e}) = 0$, we have that the number of active faces incident to the edge~$e$ is 
\[
\wt(c|_{F_e})\ge d(\ker \partial_{F_e\to e}) = d(\ker h),
\]
and the lemma is proved.
\end{proof}

\begin{lemma}\label{lemma:lbl-edge}
    If $d(\ker h)\ge 2m+\lambda$, $\partial c=0$ and $\wt_X(c)\le a/w$, then every active edge is incident to a labeled vertex.
\end{lemma}
\begin{proof}
    The number of active vertical edges is at most $\wt_X(c)w\le a$.
    Let $S\subseteq E_\uparrow$ be the set of active edges that are not incident to a labeled vertex, $A\subset E_\uparrow$ be the set of all active edges. %, $\bar{S}=A\setminus S$ be the set of active edges incident to labeled vertices. 
    If an~active edge is not incident to labeled vertices, then it is not incident to active vertices (every active vertex is labeled), then by definition it is face-active, hence $S$ is a subset of face-active edges.
    
    Consider the subposet $\gve=E_\uparrow\cup F$ of the poset $X$. Since each face from $F$ is incident to exactly two vertical edges from $E_\uparrow$, $\gve$ can be interpreted as a~graph with $V(\gve)=E_\uparrow$ and $E(\gve)=F$. 
    We have $$\gve=\{x\cdot g\cdot y\mid x\in V(\Gamma)\cup E(\Gamma),y\in E(\Gamma), g\in G\}=\bigsqcup_{y\in E(\Gamma)}\gve^{(y)}$$
    where
    $$\gve^{(y)}=\{x\cdot g\cdot y\mid x\in V(\Gamma)\cup E(\Gamma), g\in G\}\simeq \hat{\Gamma}.$$
    By property \ref{exp-prop:disjoint-union} of edge expansion $\gve$ has the same edge expansion as $\hat\Gamma$, i.e. it is $(a, \lambda)$-edge-expanding.
    The sets $S$ and $A$ can be interpreted as sets of vertices of graph $\gve$.  
    From the edge expansion of $\gve$ we have $|E_{\gve}(S,S)|\le \lambda |S|$. 
    
    On the other hand, by Lemma \ref{lemma:face-active-edge} since each edge $e\in S$ is face-active, it is incident to at least $d=d(\ker h)$ active faces, hence in the graph $\gve$ it is adjacent to at least $d\ge 2m+\lambda$ active edges, therefore $|E_{\gve}(S,A)|\ge (\lambda+2m)|S|$. 
    Thus
    $$|E_{\gve}(S,A\setminus S)|=|E_{\gve}(S,A)|-|E_{\gve}(S,S)|\ge (\lambda+2m)|S|-\lambda|S|=2m|S|.$$
    Suppose, $|S|\ne\varnothing$. 
    Then there exists an~edge $e\in S$ adjacent to $2m$ edges $e_1,\dots,e_{2m}\in A\setminus S$ in $\gve$. 
    By the definition of $A$ and $S$ each of the edges $e_i$ is incident to some labeled vertex $x_i$, which is adjacent to one of the two vertices of $e$ in $\Lambda$. 
    Hence, there are $2m$ different labeled vertices adjacent to one of the vertices of the edge $e$, and therefore one of these vertices is adjacent to at least $m$ labeled vertices, therefore it is labeled by definition.
    This contradicts the fact that the edge $e$ is from $S$ and cannot be incident to labeled vertices. Hence $S=\varnothing$, and the lemma is proved.
\end{proof}

In the next lemma, we need the following definition.
\begin{definition}
For a~given vector $y\in \FF_q^r$ and a~parity-check matrix $h\in\FF_q^{r\times w}$ we say that a~vector $x\in\F_q^w$ is an~\emph{$(y, h)$-coset leader} if it has the minimal possible Hamming weight among the vectors from $\{x\in \F_q^w \mid hx = y \}$. 
\end{definition}

\begin{lemma}\label{lemma:AB*-vert-cases}
    Suppose the pair of codes $(\ker h,\im h'^*)$ is $(s, 2m, \beta)$-product-expanding, $h'$ has full rank, $\beta w\ge 4m+3$, $d=\min(d(\ker h), d(\im h'^*))\ge 4m$, and $m\ge \max(4s/d,\lambda)$.
    If $c$ is a~locally minimal $1$-cycle, and $\wt_X(c)\le a/w$, then for each active vertex $v$ one of the following conditions holds:
    \begin{enumerate}
        \item $v$ is $m$-edge-expanding (i.e., it is adjacent to at least $m$ labeled vertices in $\Lambda$);
        \item $v$ is $s$-face-expanding (i.e., it is adjacent to at least $s$ labeled vertices in $\Lambda^2$).
    \end{enumerate}
\end{lemma}
\begin{proof}

    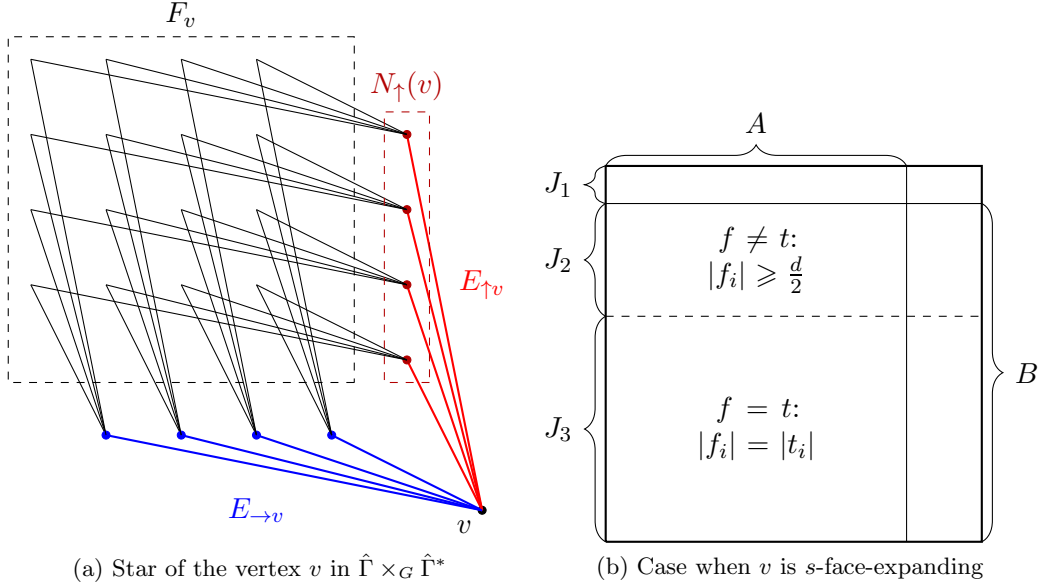
\begin{figure}[ht]
        \centering
        \subfloat[Star of the vertex $v$ in $\hat\Gamma\times_G\hat\Gamma^*$]{
            \begin{tikzpicture}
            \begin{scope}[xscale=-1]
            \filldraw [black] (0,0) circle (1.5pt);
            \foreach \i in {1,...,4}{
                \draw[-,red,thick] (0,0)--(1,1+\i);
                \draw[-,blue,thick] (0,0)--(1+\i,1);
                \filldraw [red!70!black] (1,1+\i) circle (1.5pt);
                \filldraw [blue] (1+\i,1) circle (1.5pt);
                \foreach \j in {1,...,4}{
                    \draw[-,black] (2+\j,2+\i)--(1,1+\i);
                    \draw[-,black] (2+\j,2+\i)--(1+\j,1);
                    %\filldraw [green!50!black] (2+\j,2+\i) circle (1.5pt);
                }
            }
            \draw[-,dashed,red!70!black] (0.7,1.7)--(1.3,1.7)--(1.3,5.3)--node[above]{$N_\uparrow(v)$}(0.7,5.3)--(0.7,1.7);
            \node[red] at (0,3) {$E_{\uparrow v}$};
            \node[blue] at (3,0) {$E_{\rightarrow v}$};
            \node[below left] at (0,0) {$v$};
            \draw[-,dashed,black] (1.7,1.7)--(6.3,1.7)--(6.3,6.3)--node[above]{$F_v$}(1.7,6.3)--(1.7,1.7);
            \end{scope}
        \end{tikzpicture}
        }
        %\hspace{1cm}
        \subfloat[Case when $v$ is $s$-face-expanding]{
        \begin{tikzpicture}
            %\filldraw [black] (0,0) circle (1.5pt);
            %\filldraw [black] (1,0) circle (1.5pt);
            \draw[-,thick] (0,0)--(5,0)--(5,5)--(0,5)--(0,0);
            \draw[-] (0,4.5)--(5,4.5);
            \draw[-,dashed] (0,3)--(5,3);
            \draw[-] (4,0)--(4,5);
            \draw [decorate,decoration={brace,amplitude=8pt}]
                (0,5) -- (4,5) node [black,midway,above,yshift=0.8em] {$A$};
            \draw [decorate,decoration={brace,amplitude=8pt}]
                (0,0) -- (0,3) node [black,midway,left,xshift=-0.8em] {$J_3$};
            \draw [decorate,decoration={brace,amplitude=8pt}]
                (0,3) -- (0,4.5) node [black,midway,left,xshift=-0.8em] {$J_2$};
            \draw [decorate,decoration={brace,amplitude=8pt}]
                (0,4.5) -- (0,5) node [black,midway,left,xshift=-0.8em] {$J_1$};
            \draw [decorate,decoration={brace,amplitude=8pt,mirror}]
                (5,0) -- (5,4.5) node [black,midway,right,xshift=0.8em] {$B$};
            \node[text width=1.5cm,align=center] at (2,3.75) {$f\ne t$: $|f_i|\ge \frac{d}{2}$};
            \node[text width=2cm,align=center] at (2,1.5) {$f = t$: $|f_i|=|t_i|$};
            %\draw[-,fill=white] (0,2)--(-0.4,2.3)--(0.6,2.3)--(1,2)--(0,2);
        \end{tikzpicture}
        }
        \caption{Local expansion for the vertex $v$}
        \label{fig:star}
    \end{figure}

    Before we start, let us fix some active vertex~$v=v'\cdot g\cdot v''$; $v',v''\in V(\Gamma)$, $g\in G$.
    Let $y=c|_v\in \FF_q^{r\times r'} v$, 
    $f=c|_{F_v}\in \FF_q F_v$. Then it is not hard to see that
    \begin{align*}
     E_{\rightarrow v} &= \{e'\cdot g'\cdot v''\in \Eh \mid \hat e'_{g'}\succ_{\hat\Gamma} \hat v'_g\},\\     
     E_{\uparrow v} &= \{v'\cdot g''\cdot e''\in \Ev \mid \hat e''_{g''}\succ_{\hat\Gamma} \hat v''_g\},\\
    F_{v} &= \{e'\cdot g'g^{-1}g''\cdot e''\in F\mid \hat e'_{g'}\succ_{\hat\Gamma} \hat v'_g,\  \hat e''_{g''}\succ_{\hat\Gamma} \hat v''_g \}. 
    \end{align*}
    Since $|E_{\rightarrow v}| = |E_{\uparrow v}| = w$, and each face from $F_v$ is incident to one edge from $\Ev$ and one edge from $\Eh$, the set $F_v$ is in natural one-to-one correspondence\footnote{An~equivalent way to express this property is to say that the $2$-dimensional complex $\tilde X = \hG \times_G \hG$ is a~\emph{complete square complex}~\cite{Wise:2007}, i.e., a~square complex where the link of each vertex is isomorphic to a~complete bipartite graph.} with the set $E_{\rightarrow v}\times E_{\uparrow v}$ (see Fig. \ref{fig:star}(a)). Therefore
    we can represent the restriction $f=c|_{F_v}$ as a~$w\times w$ matrix with the rows and columns indexed by the edges from $E_{\uparrow v}$ and $E_{\rightarrow v}$ respectively, i.e., $f\in\F_q F_v \cong \FF_q(E_{\rightarrow v}\times E_{\uparrow v})$. % where $E^{(\Gamma)}=E(\Gamma)$.
    Define the set 
    \[
    \Nv(v):=\{v'\in V\mid v \adj_e v', \ e\in E_\uparrow\},
    \]
    which consists of the vertices connected to $v$ by vertical edges. Note that the set of elements from $X(1)=V\cup F$ incident to the elements from $E_{\uparrow v}\subseteq X(0)$ is equal to $V_{E_{\uparrow v}}\cup F_{E_{\uparrow v}}$, where $V_{E_{\uparrow v}}=\Nv(v)\cup\{v\}$ and $F_{E_{\uparrow v}}=F_v$. %$E_{\uparrow v}\cap E_{\uparrow v'}=\{uu'\}$ for each $v'\in \Nv(v)$ (because $\hat \Gamma$ doesn't contain multiple edges), 
    Hence we obtain
    \begin{equation*}%\label{eqn:dve}
        (\partial c)|_{E_{\uparrow v}}
        =\bigl(\partial (c|_v + c|_{F_v} + c|_{\Nv(v)})\bigr)|_{E_{\uparrow v}}
        =\underbrace{\partial_{v\to E_{\uparrow v}}}_{\id \otimes {\partial_{\cB}^{(v'')}}^*}(y)  
             +\underbrace{\partial_{F_v\to E_{\uparrow v}}}_{\partial_{\cA}^{(v')}\otimes \id}(f)+\partial_{\Nv(v)\to E_{\uparrow v}}(c|_{\Nv(v)}).%=\sum_{v'\in \Nv(v)}\partial_{v'\to uu'}(c|_{v'}).  
    \end{equation*}
    Since $\cA\in\fT_G(\hat\Gamma; h)$, $\cB\in\fT_G(\hat\Gamma;h')$, we have $\partial_{\cA}^{(v')}\sim h$ and $\partial_{\cB}^{(v'')}\sim h'$, therefore with a~proper ordering of the edges in $E_v$ we can identify $\partial_{v\to E_{\uparrow v}}$ with $I_{r}\otimes {h'}^*$ and $\partial_{F_v\to E_{\uparrow v}}$ with $h\otimes I_w$.
    Consider $z_v:=(I_{r}\otimes {h'}^*)y$, $z_F:=(h\otimes I_w)f$, and $z_N:=\partial_{\Nv(v)\to E_{\uparrow v}}(c|_{\Nv(v)})$. Then we have
    $$0=(\partial c)|_{E_{\uparrow v}}=z_v+z_F+z_N.$$
    Since each vertex $v'\in \Nv(v)$ is connected to $v$ by a~single vertical edge\footnote{Here we use the assumption from Subsection~\ref{sc:proof-out} that $\hat\Gamma$ is simple. In fact, the lemma can also be proved in the case of multiple edges in~$\hat\Gamma$.}% and $\partial_{v'\to E_{\uparrow v}}$
    , we have that $|E_{\uparrow v}\cap E_{\uparrow v'}|=1$, $\supp \partial_{v'\to E_{\uparrow v}}(c|_{v'})\subseteq E_{\uparrow v'}\cap E_{\uparrow v}$, and hence $\wt_X(\partial_{v'\to E_{\uparrow v}}(c|_{v'}))\le\wt_X(c|_{v'}) \le 1$. Therefore we get
    \begin{multline*}
    \wt_X(z_N)=\wt_X\Bigl(\sum_{v'\in \Nv(v)}\partial_{v'\to E_{\uparrow v}}(c|_{v'})\Bigr)\le \sum_{v'\in \Nv(v)}\wt_X\rbr{\partial_{v'\to E_{\uparrow v}}(c|_{v'})}\le\\ \le \sum_{v'\in \Nv(v)}\wt_X(c|_{v'})= \wt_{\Nv(v)}(c).        
    \end{multline*}
    Note that $\wt_{\Nv(v)}(c)$ is the number of active vertices adjacent to $v$ by vertical edges. If $\wt_X(z_N)\ge m$, then $\wt_{\Nv(v)}(c)\ge m$, and hence $v$ is $m$-edge-expanding and the lemma is proved. 

    % \mark{Note that $\partial|_{v}=I_{w}\otimes h'^*$, $\partial|_{F_v}=h\otimes I_{w}$.}

    %We have $z_V+z_F=(\partial c)|_{E_v}=0$, %$\wt_{\cF}(z)$ is the number of syndrome edges in $E_v$, 
    % $z_V-z_v$ is the syndrome produced by adjacent to $v$ active vertices, hence the number of these vertices is at least $\wt_{\cF}(z_V-z_v)$. 
    % If $\wt_{\cF}(z_V-z_v)\ge m$, then the first case of the lemma assertion holds. 
    % In the rest of proof we will consider the most complex case when $\wt_{\cF}(z_V-z_v) < m$. 
    In the rest of the proof, we consider the most complex case when $\wt_X(z_N)<m$.
    Let $A\subseteq E_{\to v}$ (resp. $B\subseteq E_{\uparrow v}$) be the set of horizontal (resp. vertical) edges connecting $v$ with the unlabeled vertices. 
    Each pair of edges in $A\times B$ determines a~face incident to $v$ and not incident to the labeled vertices adjacent to $v$ in $\Lambda$. To prove the $s$-face expansion of $v$, first we need to show that $\wt_{A\times B}(f)\ge s$.
    If $|A|\le w-m$ or $|B|\le w-m$, then there are at least $m$ labeled vertices adjacent to $v$ in $\Lambda$, hence $v$ is $m$-edge-expanding. 
    In the rest of proof, we consider the case when $|A|,|B|>w-m$.
    
    % Informally, the plan is as follows:
    % \begin{enumerate}
    %     \item consider the matrix $t\in \FF_q F_v$ of minimal weight such that $th^*=z_v$; %rows of $t$ are coset leaders; 
    %     \item since $z_N$ is small we have $th^*=z_v\approx z_F=fh^*$;
    %     \item if $t\not\approx f$, then $f$ is far from minimal matrix with property $fh^*\approx z_v$ and we can obtain required bound;%most of these rows of $f$ are not coset leaders, and the weight of $f$ is $\Omega(m d)$;
    %     \item otherwise $t\approx f=c|_{F_v}$ %is the same as $f$ on $A\times J_3$ for sufficiently large $J_3\subseteq B$, 
    %     and using local minimality of $c$ we obtain $\beta w$-minimality of $t$;%\approx f=c|_{F_v}$;
    %     \item from product expansion of $(\ker h,\im {h'}^*)$ we obtain lower bound on  $\wt(t|_{A\times J_3})$ for some $J_3\subset B$ which implies required bound on $\wt(f|_{A\times B})$.
    %     \item since different active faces from $A\times B$ give different labeled vertices adjacent to $v$ in $\Lambda^2$, we obtain face expansion of vertex $v$.
    % \end{enumerate}
    
    It this case, we have 
    $\wt_X(z_F+z_v)=\wt_X(z_N) < m$.
    Let $z_v=(z_v^1,\dots,z_v^w)=(I_{r}\otimes h'^*)y$,
    $t=(t^1,\dots,t^w)\in\FF_q^w\otimes \FF_q^w$, where for each $i\in [w]$ the vector $t^i$ is some~$(z_v^i,h)$-coset leader. Then $(h\otimes I_{w})t=z_v$, and 
    $$(h\otimes g')t=(I_{r}\otimes g')z_v=(I_{r}\otimes g'{h'}^*)y=0,$$
    where $g'$ is a~parity-check matrix for the code $\im h'^*$. Hence $t\in\ker(h\otimes g')$. 
    
    Consider $f'=f+t=({f'}^1,\dots,{f'}^w)$. We call the component ${f'}^i$ the \emph{$i$-th row} of $f'$. We have $(h\otimes I_{w})f'=z_F+z_v$. 
    For each $i\in [w]$ we have one on the following cases:
    \begin{enumerate}
        \item $i\not\in B$: the corresponding vertical edge connects $v$ with an~active vertex; %, i.e. $h f'_i\ne 0$;
        \item $i\in B$ and ${f'}^i\ne 0$: in this case $h {f'}^i=0$, i.e. ${f'}^i\in \ker h\setminus \{0\}$, hence $|{f'}^i|\ge d$;
        \item $i\in B$ and ${f'}^i=0$: in this case $f^i = -t^i$, hence $\wt(f^i|_A)=\wt(t^i|_A)$.
    \end{enumerate}
    Denote by $J_1$, $J_2$, and $J_3$ the sets of indices corresponding to these cases (see Fig. \ref{fig:star}(b)).
    For these sets we have the following conditions:
    \begin{align*}
        &[m]=J_1\sqcup J_2\sqcup J_3,&&B=J_2\sqcup J_3,
        &&|J_1|<w.
    \end{align*}
    There are two cases we need to consider:
    \begin{enumerate}
        \item $|J_3|<w-2m$. Then 
        $$|J_2| = w-|J_1|-|J_3| > w-m-(w-2m)=m \ge 4s/d.$$
        Each of the rows ${f'}^i$ for $i\in J_2$ has weight at least $d$.
        On the other hand, for each $i\in J_2$ since $t^i$ is a~$(z^i_v, h)$-coset leader and ${f'}^i\in\ker h$, we have $\wt(t^i)\le\wt(t^i+{f'}^i)=\wt(f^i)$, and hence $\wt({f'}^i)\le \wt(t^i)+\wt(f^i)\le 2\wt(f^i)$. Therefore $\wt(f^i)\ge \wt({f'}^i)/2\ge d/2$. Thus we obtain 
        $$\wt(f|_{A\times B})\ge\wt(f|_{A\times J_2})\ge |J_2|\rbr{\frac{d}{2}-(w-|A|)}\ge\frac{4s}{d}\underbrace{\left(\frac{d}{2}-m\right)}_{\ge d/4}\ge s.$$
        \item $|J_3|\ge w-2m$. 
        Each row $t^i$ has the minimal weight in the coset $t^i+\ker h$ since it is a~$(z_v^i,h)$-coset leader. Suppose that $t$ is not a~$\beta w$-minimal codeword. Then there exist a~column $t_j$ and a~vector $\Delta t\in\im {h'}^*$ such that $\wt(t_j+\Delta t)\le\wt(t_j)-\beta w$. 
        Since $t_j|_{J_3}=f_j|_{J_3}$, we have 
        $$\wt(f_j+\Delta t)\le \underbrace{\wt(f_j+t_j)}_{\le w-|J_3|\le 2m}+\underbrace{\wt(t_j+\Delta t)}_{\le \wt(t_j)-\beta w}\le 2m+\underbrace{\wt(t_j)}_{\le \wt(f_j)+2m}-\beta w\le \wt(f_j)+4m-\beta w.$$
        Taking into account that $\beta w\ge 4m+3$, we have $\wt(f_j+\Delta t)\le \wt(f_j)-3$. 
        Since $\Delta t\in \im {h'}^*$, there exists $u\in \FF_q^{r'}$ such that $\Delta t={h'}^* u$. Consider $u e_j\in \cC_0$, where $e_j\in E_{\rightarrow v}$ is the $j$-th horizontal edge such that $v\adj_{e_j} v_j$, i.e., $e_j$ is incident to the faces corresponding to the $j$-th column of~$f$. Then we get $\partial_{e_j\to F_{e_j}}\sim {h'}^*$, and $|V_{e_j}|=2$. Therefore we obtain
        $$\wt_{F}(c+\partial (u e_j))-\wt_{F}(c)
        =\wt(\underbrace{c|_{F_{e_j}}}_{=f_j}+\underbrace{\partial_{e_j\to F_{e_j}}(u e_j)}_{={h'}^*u=\Delta t})-\wt(c|_{F_{e_j}})\le -3,$$
        $$\wt_V(c+\partial (u e_j))-\wt_V(c)\le |\supp \partial (u e_j)\cap V|\le |V_{e_j}|=2,$$
        and finally we see that
        $$\wt_{X}(c+\partial (u e_j))-\wt_{X}(c)\le -1$$
        which contradicts the local minimality of $c$.
        Hence our assumption is wrong, and $t$ is a~$\beta w$-minimal codeword. %$\wt(t^j+\Delta t)>\wt(t^j)-\beta w$ for each $j\in[w]$
        Therefore from the $(s,2m,\beta)$-product-expansion property of $(\ker h, \im h'^*)$ we obtain $\wt(t|_{A\times J_3})\ge s$, and it follows that 
        $$\wt(f|_{A\times B})\ge \wt(f|_{A\times J_3})=\wt(t|_{A\times J_3})\ge s.$$
    \end{enumerate}
    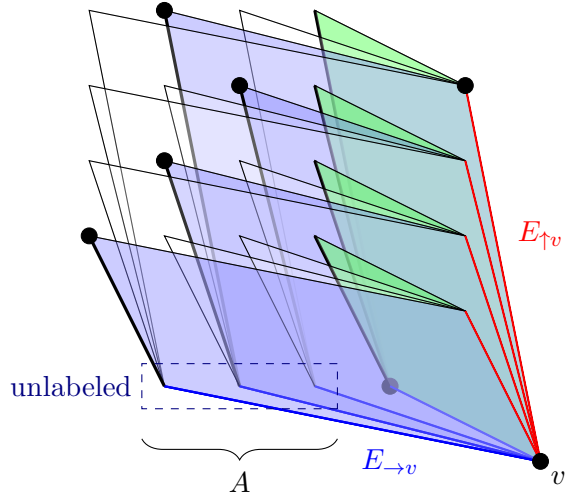
\begin{figure}
        \centering
        \begin{tikzpicture}
            \begin{scope}[xscale=-1]
            \foreach \i in {1,...,4}{
                %\filldraw [red!70!black] (1,1+\i) circle (1.5pt);
                %\filldraw [blue] (1+\i,1) circle (1.5pt);
                \filldraw[fill=green,opacity=0.4] (0,0)--(1,6-\i)--(3,7-\i)--(2,1)--(0,0);
                \draw[-,black,very thick] (3,7-\i)--(2,1);
                \filldraw [black] (2,1) circle (3pt);
                \foreach \j in {1,...,4}{
                    \ifthenelse{\i=1 \AND \j=3 \OR \i=\j \AND \j>1}{
                    \filldraw[fill=blue!50!white,opacity=0.4] (0,0)--(1,6-\i)--(2+\j,7-\i)--(1+\j,1)--(0,0);
                    \filldraw [black] (2+\j,7-\i) circle (3pt);
                    \draw[-,black,very thick] (2+\j,7-\i)--(1+\j,1);
                    }{
                    \filldraw[fill=white,opacity=0.2] (0,0)--(1,6-\i)--(2+\j,7-\i)--(1+\j,1)--(0,0);
                    }
                    \draw[-,black] (0,0)--(1,6-\i)--(2+\j,7-\i)--(1+\j,1)--(0,0);
                    %\filldraw [green!50!black] (2+\j,2+\i) circle (1.5pt);
                    \draw[-,blue,thick] (0,0)--(1+\j,1);
                }
                \draw[-,red,thick] (0,0)--(1,6-\i);
            }
            \draw[-,dashed,blue!50!black] (2.7,0.7)--(5.3,0.7)--node[left]{unlabeled}(5.3,1.3)--(2.7,1.3)--(2.7,0.7);
            %\draw[-,dashed,blue!50!black] (2.5,0.7) ellipse (0.15 and 0.4);
            \node[red] at (0,3) {$E_{\uparrow v}$};
            \node[blue] at (2,0) {$E_{\rightarrow v}$};
            \node[below right] at (0,0) {$v$};
            %\draw[-,dashed,black] (1.7,1.7)--(6.3,1.7)--(6.3,6.3)--node[above]{$F_v$}(1.7,6.3)--(1.7,1.7);
            \filldraw [black] (0,0) circle (3pt);
            \filldraw [black] (1,5) circle (3pt);
            \draw [decorate,decoration={brace,amplitude=8pt}]
                (2.7,0.3) -- (5.3,0.3) node [black,midway,below,yshift=-0.8em] {$A$%: blue edges from $u$ to not labeled vertices
                };
            \end{scope}
        \end{tikzpicture}
        \caption{Active elements in the star of $v$: black circles---labeled vertices, green faces---active faces from $(E_{\rightarrow v}\setminus A)\times E_{\uparrow v}$, blue faces---active faces from $A\times E_{\uparrow v}$,
        thick black edges---active vertical edges that are not incident to $v$. Each thick edge is incident to a~labeled vertex which is the opposite to $v$ in this face.}
        \label{fig:active}
    \end{figure}
    Thus in both cases $\wt(f|_{A\times B})\ge s$. 
    Each active face is incident to 2 active vertical edges. Since $d\ge 4m> 2m+\lambda$, the conditions of Lemma \ref{lemma:lbl-edge} satisfied, therefore each of these active edges is incident to a labeled vertex.
    For an~active face $x\in A\times B$ one of its vertical edges is incident to the
    vertex $v$; another vertical edge is incident to some labeled vertex $v_x$ which is not adjacent to $v$ in $\Lambda$ by a~horizontal edge, hence $v_x$ is the opposite vertex to $v$ in the face $x$, i.e. it is connected to $v$ by a~path of length~2 consisting of one horizontal and one vertical edge in the graph $\Lambda$ (see Fig.~\ref{fig:active}). It is not hard to see that all these length~2 paths are different (though some vertices $v_x$ may be equal), and for each $x\in A\times B$ the vertex~$v_x$  is adjacent\footnote{Note that $v_x$ can be equal to $v$, which gives a~loop in~$\Lambda^2$.} to~$v$ in~$\Lambda^2$. Thus $v$ is $s$-face-expanding.  
\end{proof}
From Lemma \ref{lemma:AB*-vert-cases} and the definition of the labeled vertices we obtain the following result.
\begin{corollary}\label{col:edge-face-exp}
    Suppose the pair of codes $(\ker h,\im h'^*)$ is $(s, 2m, \beta)$-product-expanding, $\beta w\ge 4m+3$, $d=\min(d(\ker h), d(\im h'^*))\ge 4m$, and $m\ge \max(4s/d,\lambda)$.
    If $c$ is a~locally minimal $1$-cycle, and $\wt_X(c)\le a/w$, then for each labeled vertex $v$ one of the following conditions holds:
    \begin{enumerate}
        \item $v$ is $m$-edge-expanding (i.e. it is adjacent to at least $m$ labeled vertices in $\Lambda$);
        \item $v$ is $s$-face-expanding (i.e. it is adjacent to at least $s$ labeled vertices in $\Lambda^2$).
    \end{enumerate}
\end{corollary}
\begin{proof}
    If the vertex $v$ is active, then the lemma assertion is true by Lemma \ref{lemma:AB*-vert-cases}. Otherwise, by definition, the vertex $v$ is adjacent to at least $m$ active vertices in $\Lambda$, and hence it is $m$-edge-expanding.
\end{proof}

\begin{lemma}\label{lemma:AB*vzero}
    Suppose the pair of codes $(\ker h,\im h'^*)$ is $(s, 2m, \beta)$-product-expanding, $\beta w\ge 4m+3$, $d=\min(d(\ker h), d(\im h'^*))\ge 4m$, $m\ge \max(4s/d, 2\lambda')$, and $s\ge 2\lambda''$ where\footnote{The parameters $\lambda'$ and $\lambda''$ correspond to the edge expansion of the graphs $\Lambda$ and $\Lambda^2$.} $\lambda'=2\lambda$ and $\lambda''=8\lambda^2(\ln w+2)$. If $c$ is a~locally minimal $1$-cycle, and $\wt_X(c)\le \frac{a}{2w}$, then $c=0$.
\end{lemma}
\begin{proof}
    Let $L$ be the set of labeled vertices.
    Then by Corollary \ref{col:edge-face-exp} each vertex $v\in L$ is either $m$-edge-expanding or $s$-face-expanding, i.e. $L=L_e\cup L_f$ where $L_e$ is the set of $m$-edge-expanding vertices, $L_f$ is the set of $s$-face-expanding vertices.
    By definition we have 
    \begin{equation}\label{eqn:Lef-lower}
      |E_{\Lambda}(L_e,L)|\ge m|L_e|,\qquad |E_{\Lambda^2}(L_f,L)|\ge s|L_f|.  
    \end{equation}
    Since each labeled vertex $v$ is either active ($v\in \supp c_V$) or incident to a~face-active edge, and hence adjacent to at least $d$ active faces, we get 
    \[
    |L|\le \wt_{X}(c_V)+4\wt_X(c_F)/d\le \wt_{X}(c_V)+\wt_X(c_F)=\wt_X(c)\le \frac{a}{2w}%\lambda n/w^2.
    \]
    Hence by $(a,\lambda')$-edge-expansion of $\Lambda$ we have
    $$|E(L_e,L)|\le\lambda'\sqrt{|L||L_e|}.$$
    Similarly, from $(a/2w,\lambda'')$-edge-expansion of $\Lambda^2$ we obtain
    $$|E(L_f,L)|\le\lambda''\sqrt{|L||L_f|}.$$
    Taking into account \eqref{eqn:Lef-lower}, we obtain
    $$m|L_e|\le \lambda'\sqrt{|L||L_e|},\qquad s|L_f|\le \lambda''\sqrt{|L||L_f|},$$
    and hence
    $$|L_e|\le \rbr{\frac{\lambda'}{m}}^2|L|\le \frac{|L|}{4},\qquad |L_f|\le \rbr{\frac{\lambda''}{s}}^2|L|\le\frac{|L|}{4}.$$
    Since $|L|=|L_e\cup L_f|\le |L|/2$, we obtain $|L|=0$. Since each active vertical edge by Lemma \ref{lemma:lbl-edge} contains labeled vertices, we have that the number of active vertical edges is 0, and hence $c=0$.
\end{proof}

\subsection{Proof of the theorems}

\begin{proposition}\label{prop:main}
For every finite field $\FF_q$, intervals $(\rho_0, \rho_1), (\rho'_0, \rho'_1) \subseteq (0,1)$, constant $\mu>0$, and infinite set $W\subseteq \NN$, there exist matrices  $h\in\FF_q^{r\times w}$, $h'\in\FF_q^{r'\times w}$ for sufficiently large $w\in W$ such that $r/w\in (\rho_0, \rho_1)$, $r'/w\in (\rho'_0, \rho'_1)$, and for every $G$-lifted $w$-regular $(a,\mu \sqrt{w})$-edge-expanding simple graph $\hat{\Gamma}$ and Tanner codes $\cA\in \fT_G(\hat{\Gamma};h)$, $\cB\in \fT_G(\hat{\Gamma};h')$ with a~free action of a~group~$G$ we have 
    $$\dLM^{(1)}(\cA\otimes_G\cB^*)\ge a/2w,$$ 
    $$\dLM^{(1)}(\cB\otimes_G\cA^*)\ge a/2w.$$ 
\end{proposition}
\begin{proof}
     Let $w$ be a parameter which we will fix later. Define $\eps:=1/6$, $m:=w^{1/2+\eps}$, $s:=w^{1+\eps}$, $r:=\floorbr{\frac12(\rho_0 +\rho_1) w}$, $r':=\floorbr{\frac12(\rho'_0 +\rho'_1) w}$. 
     By Lemma~\ref{lemma:rand-gh} with $\alpha:=1$, $\gamma:=2$, there exist $\beta_1,\beta_2>0$ and $\delta_1,\delta_2>0$ such that for random matrices $h\in\FF_q^{r\times w}$, $h'\in\FF_q^{r'\times w}$ as $w\to\infty$ the following three conditions hold with high probability\footnote{Note that Lemma~\ref{lemma:rand-gh} is used here twice. First time $h$ is interpreted as a~parity-check matrix, but $h'$ as a~generator matrix, and the second time vice versa.}:
     \begin{enumerate}
         \item the matrices $h$ and $h'$ have maximal rank, i.e. $\rk h=r$, $\rk h'=r'$;
         \item the pair $(\ker h, \im {h'}^*)$ is $(s,2m,\beta_1)$-product-expanding and $\min(d(\ker h),d(\im {h'}^*))\ge \delta_1 w$;
         \item the pair $(\im h^*, \ker h')$ is $(s,2m,\beta_2)$-product-expanding and $\min(d(\ker h'),d(\im h^*))\ge \delta_2 w$.
     \end{enumerate}
    Therefore by the union bound for a~sufficiently large $w_0\in\NN$ for every $w\ge w_0$ there exists a~pair $(h,h')$ that satisfies these three conditions.  Let $\beta:=\min(\beta_1,\beta_2)$, $d:=\min(\delta_1,\delta_2)w$, $\lambda:=\mu\sqrt{w}$, $\lambda':=2\lambda$, $\lambda'':=8\lambda^2(\ln w+2)$. We have
    $$d=\Theta(w),\quad\lambda'=\Theta(w^{\frac{1}{2}})=o(m),\quad \lambda''=\Theta(w\ln w)=o(s),\quad m=\Theta(w^{\frac{1}{2}+\eps})=o(w)$$
    as $w\to\infty$.
    Hence there exists $w_1$ such that for every $w\ge w_1$ the following inequalities hold:
    \begin{equation}\label{eq:d_beta_m_cond}
    d>4m,\quad \beta w\ge 4m+3,\quad m>\max\rbr{\frac{4s}{d},2\lambda'},\quad s>2\lambda''.    
    \end{equation}

    Since the set $W$ is infinite, we can take $w:=\min\{w\in W\mid w\ge \max(w_0,w_1)\}$ and fix some pair $(h,h')$ that satisfy the conditions 1--3.  Now consider a~$G$-lifted $(a,\lambda)$-edge-expanding graph $\hat{\Gamma}$ and some $G$-lifted Tanner codes $\cA\in \fT_G(\hat{\Gamma};h)$, $\cB\in \fT_G(\hat{\Gamma};h')$.
    
    Since $\min(d(\ker h),d(\ker h'),d(\im h^*),d(\im {h'}^*))\ge d$, and conditions (\ref{eq:d_beta_m_cond}) hold, we can apply Lemma \ref{lemma:AB*vzero} to the pair of codes $(h,h')$ and obtain that every non-zero locally minimal 1-cycle of the chain complex $\cA\otimes_G\cB^*$ has the weight at least $a/2w$. Hence we have
    $$\dLM^{(1)}(\cA\otimes_G\cB^*)\ge a/2w.$$
    Since lemma~\ref{lemma:AB*vzero} is also applicable to the pair $(h',h)$, we have
    $$\dLM^{(1)}(\cB\otimes_G\cA^*)\ge a/2w,$$
    which completes the proof of the proposition.
\end{proof}

\setcounter{theorem}{0}
\begin{theorem}
      For every number $R \in (0, 1/2)$ and finite field $\mathbb{F}_q$ it is
  possible to find universal constants $s$ and~$\omega$ such that there exists
  an~explicit family of $(\omega, s)$-locally testable classical LDPC codes with
  the~parameters $[n, k \geqslant Rn, d = \Theta (n)]_q$ as $n \rightarrow
  \infty$.
\end{theorem}
\begin{proof}
    Fix some $R\in (0,1/2)$ and put $\eps:=(1-2R)/(6-2R)$.
    Note that for any $w$ from the infinite set $W:=\{p+1\in\NN\mid p\equiv 1\mod 4\mbox{ and }p\mbox{ is prime}\}$ there exist infinite family of graphs $\bar{X}^{w-1,t}$ from Example~\ref{ex:Xpq}. By Lemma \ref{lemma:Xexp} every graph $\bar{X}^{w-1,t}$ is $(n_0(t)/\sqrt{w},8\sqrt{w})$-edge-expanding, where $n_0(t)=t(t^2-1)=|V(\bar X^{w-1,t})|$.
    Consider the~chain complex 
    \[
    \cC := \cT(\bar{X}^{w-1,t},h)\otimes_G\cT^*(\bar{X}^{w-1,t},h'),
    \]
    with the boundary operator $\partial$, where $G:=\mathrm{PSL}(\FF_t^2)$, and $h$, $h'$ are the parity-check matrices of the local codes, which we will fix later. Let $|\cdot|$ be the block weight norm $\wt_X(\cdot)$ defined on $\cC$, considered as a~chain complex with a~local system on the cell poset $X = \bar{X}^{w-1,t}\times_G (\bar{X}^{w-1,t})^*$.
    By Proposition~\ref{prop:main} for the intervals $(1-\eps,1)$, $(0,\eps)$ and the parameter $\mu=8$ there exist $w\in W$ and matrices $h\in \F_q^{r\times w}$, $h'\in \F_q^{r'\times w}$ such that for every $\bar{X}^{w-1,t}$ we have
    $$\dLM^{(1)}(\cC)\ge n_0(t)/2w\sqrt{w}$$
    where $r/w>1-\eps$, $r'/w<\eps$.
    Let $n:=\dim\cC_2$ and $m:=\dim\cC_1$, then $n=n_0(t)r w$, $m=\frac{1}{2}n_0(t)(w^2+4rr')$.
    Hence $\dLM^{(1)}(\cC)\ge \frac{n}{2w^2r\sqrt{w}}>\frac{n}{2w^{7/2}}$.
    By Lemma \ref{lemma:dist-from-lm} for all $c\in\cC_2$ we have
    $$|\partial c|\ge\min(\dLM^{(1)}(\cC), |c+Z_2(\cC)|).$$
    Since $|y|\le \wt(y)$ for $y\in\cC$ and $\wt(c)\le r|c|\le w|c|$ for $c\in\cC_2$, taking into account that $n\ge \wt(c+Z_2(\cC))$ finally we obtain
    $$\wt(\partial c)\ge \min\rbr{\frac{n}{2w^{7/2}}, \frac{\wt(c+Z_2(\cC))}{w}}\ge \frac{1}{2w^{7/2}}\wt(c+Z_2(\cC)).$$
    We have 
    $$\frac{m}{n}=\frac{w^2+4rr'}{2rw}=\frac{1+4\frac{r}{w}\cdot\frac{r'}{w}}{2r/w}\le \frac{1+4\eps}{2(1-\eps)}=1-R.$$
    In particular, we have $m<n$, and hence
    $$\frac{1}{m}\wt(\partial c)\ge\frac{w^{-7/2}}{2m}\wt(c+Z_2(\cC))\ge \frac{w^{-7/2}}{2n}\wt(c+Z_2(\cC)).$$
    Therefore the code $Z_2(\cC)$ is $(\omega,s)$-locally testable where $\omega := 2w$ and  $s:=\frac{1}{2}w^{-7/2}$.
    For the dimension $k=\dim Z_2(\cC)$ we have $k\ge n-m\ge Rn$.
    
To complete the proof we also need to show that the linear code $Z_2(\cC)$ has the minimal distance $\Theta(n)$ as $n\to\infty$.
    It is not hard to see that the minimal distance of $Z_2(\cC)$ is not less than the distance of the component Tanner code $\cT(\bar{X}^{w-1,t},h)$, which is a~classical expander code~\cite{Sipser:1996}. Thus, as it follows from the proof of Proposition~\ref{prop:main}, we can fix a~sufficiently large number $w$ such that $d(\ker h) > \lambda_2(\bar{X}^{w-1,t})$ and obtain that $d(\cT(\bar{X}^{w-1,t},h))=\Theta(n)$ as $n\to\infty$.
\end{proof}

\begin{theorem}
    For every number $R \in (0, 1)$ and finite field $\mathbb{F}_q$ there exists an~explicit
    family of quantum LDPC codes over~$\mathbb{F}_q$
    with the~parameters {$\llbracket n, k \geqslant Rn, d = \Theta (n)
    \rrbracket_q$} as $n \rightarrow \infty$.
\end{theorem}
\begin{proof}
    Fix some $R\in (0,1)$.
        Note that for every~$w$ from the infinite set $W:=\{p+1\in\NN\mid p\equiv 1\mod 4\mbox{ and }p\mbox{ is prime}\}$ there exist infinite family of graphs $\bar{X}^{w-1,t}$ from Example~\ref{ex:Xpq}.
        By Lemma~\ref{lemma:Xexp} the graph $\bar{X}^{w-1,t}$ is $(n_0(t)/\sqrt{w},8\sqrt{w})$-edge-expanding where $n_0(t)=t(t^2-1)=|V(\bar X^{w-1,t})|$.
    As~in the proof of Theorem~\ref{th:LTC}, we consider the complex $\cC=\cT(\bar{X}^{w-1,t},h)\otimes_G\cT^*(\bar{X}^{w-1,t},h')$ with the boundary operator $\partial$ where $G=\mathrm{PSL}(\FF_t^2)$. Let $|\cdot|$ be the block weight defined on $\cC$.
    By Proposition~\ref{prop:main} for $\rho_0=\rho_0'=0$,  $\rho_1 = \rho_1' = (1-R)/4$, and $\mu=8$ there exist $w\in W$ and matrices $h\in \F_q^{r\times w}$, $h'\in\F_q^{r'\times w}$ such that for all $\bar{X}^{w-1,t}$ we have
    $$\dLM^{(1)}(\cC)\ge n_0(t)/2w\sqrt{w},\qquad \dLM^{(1)}(\cC^*)\ge n_0(t)/2w\sqrt{w}$$
    where $r/w<(1-R)/4$, $r'/w<(1-R)/4$. 
    Let $n:=\dim\cC_1$, then $n=\frac{1}{2}n_0(t)(w^2+4rr')< w^2n_0(t)$.
    %Hence $\dLM^{(1)}(\cC)\ge \frac{n}{4w^3\sqrt{w}}>\frac{n}{4 w^{7/2}}$. 
    The chain complex $\cC$ defines the quantum CSS code $\cQ = \cQ(H_X,H_Z)$ with the parity-check matrices $H_X := \partial_1$ and $H_Z := \partial_2^*$.
    By Lemma \ref{lemma:dist-from-lm} for the complex~$\cC$ we have
    $$d_X(\cQ)=d(H_1(\cC))\ge \dLM^{(1)}(\cC)\ge \frac{n_0(t)}{2w\sqrt{w}}> \frac{n}{2w^{7/2}}.$$
    Similarly, since the dual chain complex $\cC^*$ is isomorphic\footnote{We say that two based chain complexes $\cC$ and $\cC'$ over $\F_q$ are \emph{isomorphic} if there exists a~one-to-one $\F_q$-linear map $f\colon \cC\to\cC'$ such that $f(\tilde{\cC}_i) = \tilde{\cC}'_i$ for every $i\in\ZZ$.} to the chain complex $\cB \otimes_G \cA^*$, then by Lemma~\ref{lemma:dist-from-lm} we have
    $$d_Z(\cQ)=d(H_1(\cC^*))\ge \dLM^{(1)}(\cC^*) > \frac{n}{2w^{7/2}},$$
    and hence $d(\cQ)=\min(d_X(\cQ),d_Z(\cQ))\ge \frac12 n/w^{7/2}$. 
    To complete the proof we also need to estimate the dimension $k = \dim(H_1(\cC))$ of the quantum code $\cQ$. We have 
    $$\dim\cC_0=n_0(t)rw=2n\frac{rw}{w^2+4rr'}<n(1-R)/2,$$ 
    $$\dim\cC_2=n_0(t)r'w=2n\frac{r'w}{w^2+4rr'}<n(1-R)/2,$$
    and therefore
    $$k=\dim(H_1(\cC))\ge n-\dim\cC_0-\dim\cC_2>n-n(1-R)/2-n(1-R)/2=nR.$$
    Thus $\cQ$ is a~$w$-limited quantum CSS code with the parameters $\llbracket n, k \geqslant Rn, d \ge \frac12 n/w^{7/2}\rrbracket_q$.
\end{proof}

% Consider a~complex $\mathcal{C}$ over $\mathbb{F}_{q^s}$, where $s \geqslant
% 1$. Since $\mathcal{C}$ is also a~vector space over $\mathbb{F}_q$, and its
% boundary map $\partial\colon\mathcal{C} \rightarrow \mathcal{C}$ is
% $\mathbb{F}_q$-linear, then we can consider $\mathcal{C}$ as a~complex over
% $\mathbb{F}_q$, which we denote by~$\mathcal{C}/ \mathbb{F}_q$. It is obvious
% that $\dim Z_i (\mathcal{C}/ \mathbb{F}_q) = s \cdot \dim Z_i (\mathcal{C})$,
% $\dim B_i (\mathcal{C}/ \mathbb{F}_q) = s \cdot \dim B_i (\mathcal{C})$, and
% $\dim H_i (\mathcal{C}/ \mathbb{F}_q) = s \cdot \dim H_i (\mathcal{C})$.
% Moreover, if we fix a~basis $\alpha_1, \ldots, \alpha_s \in \mathbb{F}_{q^s}$
% in the 

\section*{Conclusions}

In this work, we showed that there exist asymptotically good families of quantum LDPC codes, which proves the well-known qLDPC conjecture. We also conjecture that a~decoder, similar to the small-set-flip decoding algorithm from~\cite{Leverrier:2015} (see also~\cite{Evra:2020}), can be used to correct in linear time any adversarial errors up to the constant fraction of the code length. 

The constructed qLDPC codes were obtained from the $G$-lifted product of two $G$-lifted Tanner codes, and to obtain qLDPC codes of linear minimum distance a~non-abelian group~$G$ was used.
In~fact, it is not hard to see that Proposition~\ref{prop:main} implies that using the~$\mathbf{C}_\ell$-lifted product of two $\mathbf{C}_\ell$-lifted Tanner codes from~\cite{Panteleev&Kalachev:2021}, where $\mathbf{C}_\ell$ is the~cyclic group of size $\ell = \Theta(n/\log n)$, one can obtain qLDPC codes with the parameters~$\llbracket n, k = \Theta (n), d = \Theta (n / \log n) \rrbracket_q$ as $n\to\infty$. Note that very recent results on explicit  $\mathbf{C}_\ell$-lifted expander graphs from~\cite{Jeronimo:2021} implies that the construction of these qLDPC codes can also be made explicit. 

In addition, as~a~byproduct of our proof of the qLDPC conjecture, we show that the second homology groups of the constructed in this work chain complexes can be used to obtain asymptotically good families of classical LDPC codes, which are also locally testable with constant query and soundness parameters. This resolves an~important conjecture in the field of locally-testable codes\footnote{An~independent solution to this problem was also proposed in~\cite{Dinur:2021}.}.

Though all the constructions we propose here can be considered as explicit, the constant size local codes used in our expander codes are still obtained by probabilistic methods. We think that it is an~interesting open problem to find an~explicit construction of such codes. One possible option would be to use MDS codes such as Reed-Solomon codes. In fact, such non-binary local codes can be used even if we want to get codes over $\F_2$ since every classical and quantum code over $\F_{2^s}$ can be also considered as a~code over $\F_2$, and the rate and minimal distance of such a~code is at least as good as for the non-binary one. However, it is not clear whether one can find a~pair of MDS codes that satisfies the product-expansion property required for our proof to work.

We also hope that some of the methods developed in the current work can  be used to show the existence of locally-testable qLDPC codes required to prove the qLTC conjecture, which in turn implies~\cite{Eldar:2017} the NLTS conjecture. 
A~natural candidate for such a~code would be a~$5$-term chain complex, where the three middle terms corresponds to a~good qLDPC code, and the remaining two terms represents its $X$- and $Z$-meta-checks (i.e., checks on checks). In fact, similar $5$-term complexes were already used in the context of single-shot decoding of qLDPC codes~\cite[Figure~1]{Campbell:2019}.

\section*{Acknowledgment}
We would like to thank Nikolas Breuckmann and Jens Eberhardt for very helpful and insightful discussions of possible ways to get good qLDPC codes, and for an~opportunity to report our results in the QCDA seminar. We want to express our gratitude to many people, including Thomas Vidick, Sergey Sadov, Victor Albert, Shouzhen Gu, who read our manuscript and made a~number of valuable comments. We also want to thank anonymous reviewers for indicating several unclear places in our work and for pointing out the connection of our product-expansion property to the robust testability of tensor product codes.

This work was supported by the Ministry of Science and Higher Education of the Russian Federation (Grant 075-15-2020-801).

\bibliographystyle{IEEEtran}
\bibliography{codes.bib}

%\appendices
\appendix

\section{Chain complexes}\label{sc:chain}

Let $\F$ be a~field. We say
that an~$n$-dimentional vector space $V$ over $\mathbb{F}$ is {\tmem{based}} if it comes with some distinguished basis $\tilde{V} := \{ v_1, \ldots, v_n \} \subseteq V$. In this case we can
naturally identify $V$ with the~coordinate vector space~$\mathbb{F}^n$.
Moreover, we can consider the standard inner product $\langle v, v \rangle$
defined on the basis as $\langle v_i, v_j \rangle := \delta_{i \nospace
j}$ and extend it by linearity. This also allows us to identify
the~{\tmem{dual}} vector space $V^{\ast} := \tmop{Hom} (V, \mathbb{F})$
with $V$ and hence with $\mathbb{F}^n$ if for every $v \in V$ we let~$v (x)
:= \langle v, x \rangle$. Now consider an~$\F$-linear map $\varphi\colon  U \rightarrow V$ between   based vector spaces $U\cong \F^m$ and $V\cong\F^n$. We usually identify such maps with the corresponding $m\times n$ matrix over $\F$. For every such map $\phi\colon U\to V$, we can consider the corresponding \emph{transpose} map~$\varphi^{\ast} : V^{\ast} \rightarrow U^{\ast}$ that takes
each linear function $f \in V^{\ast}$ to the~function $f \circ \varphi \in
U^{\ast}$. It is easy to check that the $n\times m$ matrix of the transposed map $\phi^*$ is the transpose of the matrix for $\phi$.  

Consider a~field $\mathbb{F}$. A~{\tmem{chain complex}} ({\tmem{over}}
$\mathbb{F}$) \ is a~collection of vector spaces\footnote{In fact, the
definitions given below also can be generalized to the case when $\mathbb{F}$
is an~arbitrary commutative ring. In this case, instead of vector spaces over
$\mathbb{F}$ one should consider free $\mathbb{F}$-modules.}
$(\mathcal{C}_i)_{i \in \mathbb{Z}}$ over $\mathbb{F}$, which is convenient to
consider as one big vector space $\mathcal{C}= \bigoplus_{i \in \mathbb{Z}}
\mathcal{C}_i$, with some fixed linear operator $\partial\colon\mathcal{C}
\rightarrow \mathcal{C} $ called the~{\tmem{boundary map}} such that $\partial
\mathcal{C}_{i + 1} \subseteq \mathcal{C}_i$ and $\partial^2 = 0$ for all $i
\in \mathbb{Z}$. The~condition $\partial \mathcal{C}_{i + 1} \subseteq
\mathcal{C}_i$ says that one can define the~maps $\partial_i := \partial
|_{\cC_i}\colon\mathcal{C}_i \rightarrow \mathcal{C}_{i - 1}$, $i \in
\mathbb{Z}$; while the condition $\partial^2 = 0$ implies that $\partial_i
\circ \partial_{i + 1} = 0$ for all $i \in \mathbb{Z}$ or, equivalently, $B_i
(\mathcal{C}) \subseteq Z_i (\mathcal{C})$, where $B_i (\mathcal{C}) :=
\tmop{im} \partial_{i + 1}$, $Z_i (\mathcal{C}) := \ker \partial_i$.
Therefore for every~$i \in \mathbb{Z}$ we can define the quotient group $H_i
(\mathcal{C}) := Z_i (\mathcal{C}) / B_i (\mathcal{C})$ called
the~{\tmem{$i$-th homology group}} of the complex $\mathcal{C}$. The elements
from $\mathcal{C}_i$, $Z_i (\mathcal{C})$, and $B_i (\mathcal{C})$ are called
the {\tmem{$i$-chains}}, {\tmem{$i$-cycles}}, and {\tmem{$i$-boundaries}}
of~$\mathcal{C}$, respectively. We say that a~complex $\mathcal{C}$ is
{\tmem{based}} if every space $\mathcal{C}_i$ comes with a~distinguished basis
$\tilde{\mathcal{C}}_i \subseteq \mathcal{C}_i$, which elements are called
{{\tmem{$i$-cells}}}. In~this work we consider only {\tmem{bounded}} chain
complexes, i.e., when $\mathcal{C}_i = 0$ for all $i \nin [s, t]$. A~bounded
chain complex $\mathcal{C}$ is usually represented by the following diagram:
\[ \mathcal{C}_s \xrightarrow{\partial_s} \mathcal{C}_{s - 1}
   \xrightarrow{\partial_{s - 1}} \cdots \xrightarrow{\partial_{t + 1}}
   \mathcal{C}_t, \]
where $t - s + 1$ is called the {\tmem{length}} of $\mathcal{C}$. A complex  of length $n$ is also called an \emph{$n$-term} complex.

The definition of a~chain complex and the related terminology come from
algebraic topology, where {an~$i$-cell} $c \in \tilde{\mathcal{C}}_i$ usually
corresponds to some $i$-dimensional object, and $\partial c$ is an~algebraic
representation of its $(i - 1)$-dimensional boundary. For example, one can
consider for any simple graph $\Gamma = (V, E)$ its $2$-term chain complex $\mathcal{C}_{\bullet}
(\Gamma; \mathbb{F}_2)$ over $\mathbb{F}_2$:
\[ 
\underbrace{\mathbb{F}_2 E}_{\cC_1}  \xrightarrow{\partial_1}  \underbrace{\mathbb{F}_2 V}_{\cC_0}, 
\]
where $\tilde{\mathcal{C}}_0 := V$, $\tilde{\mathcal{C}}_1 := E$,
and the boundary map $\partial$ is defined as $\partial e := v + v'$, for every $e = \{ v, v' \} \in E$.

Sometimes it is also convenient to consider the~dual notion of a~chain complex
called {\tmem{cochain complex}}. If we have a~chain complex $\mathcal{C}$ we
can obtain the corresponding cochain complex for~$\mathcal{C}$ if we
replace $\mathcal{C}$ by its dual vector space $\mathcal{C}^{\ast} :=
\tmop{Hom} (\mathcal{C}, \mathbb{F}_q)$, and the boundary map {$\partial\colon 
\mathcal{C} \rightarrow \mathcal{C}$} by the~corresponding {\tmem{coboundary
map}} $\delta\colon\mathcal{C}^{\ast} \rightarrow \mathcal{C}^{\ast}$ that takes
each linear function $x \mapsto f (x) \in \mathcal{C}^{\ast}$ to $x \mapsto f
(\partial x) \in \mathcal{C}^{\ast}$.  Since $\partial^2 = 0$, it follows that
$\delta^2\colon x \mapsto f (\partial^2 x)$ is the~zero map, and we also get
$\delta^2 = 0$. Moreover, since $\mathcal{C}= \bigoplus_{i \in \mathbb{Z}}
\mathcal{C}_i$, we see that $\mathcal{C}^{\ast} = \bigoplus_{i \in \mathbb{Z}}
\mathcal{C}^i$ and $\delta (\mathcal{C}^i) \subseteq \mathcal{C}^{i + 1}$,
where $\mathcal{C}^i := \tmop{Hom} (\mathcal{C}_i, \mathbb{F}_q)$, $i \in
\mathbb{Z}$. Similar to the case of chain complexes, we can define
the~maps $\delta_i := \delta |_{\cC_i}\colon\mathcal{C}_i \rightarrow
\mathcal{C}_{i + 1}$, and the condition $\delta^2 = 0$ implies that $\delta_{i
+ 1} \circ \delta_i = 0$ for all $i \in \mathbb{Z}$, or, equivalently, $B^i
(\mathcal{C}) \subseteq Z^i (\mathcal{C})$, where $B^i (\mathcal{C}) :=
\tmop{im} \delta_{i - 1}$, $Z^i (\mathcal{C}) := \ker \delta_i$. 
Hence we have the spaces $\mathcal{C}^i$, $Z^i (\mathcal{C})$, and $B^i
(\mathcal{C})$ of {\tmem{$i$-cochains}}, {\tmem{$i$-cocycles}}, and
{\tmem{$i$-coboundaries}}, respectively. Since for every $i \in \mathbb{Z}$ we
have $B^i (\mathcal{C}) \subseteq Z^i (\mathcal{C})$, we can also define the
quotient group $H^i (\mathcal{C}) := Z^i (\mathcal{C}) / B^i
(\mathcal{C})$ called the~{\tmem{$i$-th cohomology group}} of $\mathcal{C}$.

Since in the current work we always assume that each $\mathcal{C}_i$ comes with some
distinguished basis $\tilde{\mathcal{C}}_i$, we can identify both
$\mathcal{C}_i$ and $\mathcal{C}^i$ with the corresponding coordinate vector
space $\mathbb{F}_q^{n_i}$, where $n_i := | \tilde{\mathcal{C}}_i |$. In
this case, the maps $\partial_i \colon \mathbb{F}_q^{n_i} \rightarrow
\mathbb{F}_q^{n_{i - 1}}$ and $\delta_{i - 1} \colon \mathbb{F}_q^{n_{i - 1}}
\rightarrow \mathbb{F}_q^{n_i}$ can be also identified with the corresponding
matrices over $\mathbb{F}_q$, and it is easy to verify that $\delta_{i - 1}$
is the transpose of~$\partial_i$.

Every chain (resp. cochain) complex can be also considered as a~cochain (resp.
chain) complex if we use the following convention $\mathcal{C}^i
=\mathcal{C}_{- i}$. Thus in what follows we are going to consider the 
cochain complex $\mathcal{C}^{\ast}$ also as the chain complex, in which case we call it the \emph{dual chain complex} of $\cC$. For example, if we
have a~chain complex, corresponding to a~quantum CCS code $\cQ$ with matrices $H_X$ and $H_Z$:
\[ 
\mathcal{C}_{\bullet} (H_X, H_Z) := \left( 
     \underbrace{\mathbb{F}_q^{m_Z}}_{\mathcal{C}_1} \xrightarrow{H^{\T}_Z}  \underbrace{\mathbb{F}_q^n}_{\mathcal{C}_0} 
     \xrightarrow{H_X}  \underbrace{\mathbb{F}_q^{m_X}}_{\mathcal{C}_{-1}} 
   \right),
\]
then its cochain complex is
\[ \mathcal{C}^{\bullet} (H_X, H_Z) := \left( \mathbb{F}_q^{m_Z}
   \xleftarrow{H_Z} \mathbb{F}_q^n \xleftarrow{H^{\T}_X} \mathbb{F}_q^{m_X}
   \right) \]
and the dual chain complex for~$\cC$ is
\[ \mathcal{C}_{\bullet}^{\ast} (H_X, H_Z) := \left( \mathbb{F}_q^{m_X}
   \xrightarrow{H_X^*} \mathbb{F}_q^n \xrightarrow{H_Z} \mathbb{F}_q^{m_Z}
   \right), \]
and we see that $\mathcal{C}_{\bullet}^{\ast} (H_X, H_Z)
=\mathcal{C}_{\bullet} (H_Z, H_X)$, i.e., the dual chain complex corresponds to the \emph{dual} CSS code $\cQ^*$, where the roles of $H_X$ and $H_Z$ are reversed.

\section{Lifted product of two classical codes}\label{sc:lp-mat}

The lifted product was introduced in~{\cite{Panteleev&Kalachev:2021}} as a~way
to generalize many known
constructions~{\cite{qldpc,Hagiwara:2007,Tillich&Zemor:2014,Haah:2011,Kovalev:2013}}
of qLDPC codes. The~general idea was to {\tmem{lift}} the~hypergraph product
construction~{\cite{Tillich&Zemor:2014}}, which, for any two classical codes
with parity-check matrices $A \in \mathbb{F}_q^{m_a \times n_a}$ and $B \in
\mathbb{F}_q^{m_b \times n_b}$, gives the~quantum CSS code $\HP (A, B)$ with
the~following parity-check matrices\footnote{If the characteristics of $\F_q$
is $2$, we can omit the sign in the definition of $\HX$.}:
\begin{eqnarray*}
  H_X & := & [A \otimes I_{m_b}, - I_{m_a} \otimes B],\\
  H_Z & := & [I_{n_a} \otimes B^{\ast}, A^{\ast} \otimes I_{n_b}] .
\end{eqnarray*}
If we replace the elements of the~matrices $A := (a_{i \nocomma j})_{m_a
\times n_a}$ and $B := (b_{i \nocomma j})_{m_b \times n_b}$ by some $\ell
\times \ell$ matrices over $\mathbb{F}_q$, we obtain matrices $\hat{A} := (\hat{a}_{i \nocomma j})_{m_a \times n_a} \in
R^{m_a \times n_a}$ and $\hat{B} := (\hat{b}_{i
\nocomma j})_{m_b \times n_b} \in R^{m_b \times n_b}$ over the matrix ring $R := \F_q^{\ell\times\ell}$. We can also consider the matrices $\hat{A}$ and $\hat{B}$ as
the $\ell$ times larger block matrices over $\F_q$, which, in turn, are used to define the $\ell$ times larger analogs of $H_X$
and $H_{Z}$ in the~following way:
\begin{equation}\label{eq:LP}
  \begin{array}{ccc}
    \hat{H}_X & := & [\hat{A} \otimes I_{m_b}, - I_{m_a}
    \otimes \hat{B}],\\
    \hat{H}_Z & := & [I_{n_a} \otimes \hat{B}^{\ast},
    \hat{A}^{\ast} \otimes I_{n_b}],
  \end{array} 
\end{equation}
where in the~transposed block matrices
$\hat{A}^{\ast}$ and $\hat{B}^{\ast}$ we also transpose each $\ell \times
\ell$ block. As~it was shown in~{\cite{Panteleev&Kalachev:2021}}, if every element (i.e., a~matrix from $R$) of $\hat{A}$ commutes with every element of
$\hat{B}$, then this~construction always gives a~quantum CSS code with
the~parity-check matrices $\hat{H}_X$ and $\hat{H}_Z$, called
the~{\tmem{lifted product}} of $\hat{A}$, $\hat{B}$ and denoted by $\LP
(\hat{A}, \hat{B})$. Actually, it is easy to see that this commutativity
condition is a~necessary and sufficient condition to produce a~well defined
CSS code. Indeed, we have:
\[ \hat{H}_X  \hat{H}_Z^{\ast} = 0 \quad \Longleftrightarrow \quad (\hat{A}
   \otimes I_{m_b})  \left( I_{n_a} \hspace{-0.17em} \otimes
   \hat{B} \right) = (I_{m_a} \hspace{-0.17em} \otimes \hat{B}) 
   (\hat{A} \otimes I_{n_b}), \]
where the~last equation is equivalent to $\hat{a}_{i \nocomma j}  \hat{b}_{s
\nocomma t} = \hat{b}_{s \nocomma t}  \hat{a}_{i \nocomma j}$ for all $i, j,
s, t$.

The most straightforward way to make this general definition always work is to
use $\ell \times \ell$ matrices from some commutative matrix ring $R \subseteq
\mathbb{F}_q^{\ell \times \ell}$. However, it also works well with
{\tmem{any}} $\ell$-dimensional associative algebra $R$ over $\mathbb{F}_q$,
not necessary a~commutative\footnote{Let us note that for all the~examples of
lifted products in {\cite{Panteleev&Kalachev:2021}} the~algebra $R$ is
commutative, and the first examples of non-abelian lifted products first
appeared in~{\cite{Breuckmann:balanced:2021}} in the context of a~very similar
construction called {\tmem{balanced product}}.} one, if we use the~right
(resp. left) regular matrix representation of its elements as the~entries of
$\hat{A}$ (resp. $\hat{B}$). Indeed, if we fix a~basis in the~algebra $R$,
then the~{\tmem{right}} (resp. {\tmem{left}}) {\tmem{regular matrix
representation}} of an~element~$r \in R$ is defined as the $\ell \times \ell$
matrix of the linear operator $\rho_r := x \mapsto xr$ (resp. $\lambda_r
:= x \mapsto rx$). Since the multiplication in $R$ is associative, then
for any $a, b \in R$ the~operators $\rho_a$ and $\lambda_b$ always commute:
\[ (\rho_a \lambda_b) (x) = (bx) a = b (xa) = (\lambda_b \rho_a) (x) . \]
\ \ \ Hence, for any two matrices $A \in R^{m_a \times n_a}$ and $B \in
R^{m_b \times n_b}$ we can replace their elements by the~corresponding right
and left matrix representations to obtain the block matrices $\hat{A}$,
$\hat{B}$ and get the well-defined CSS code using \Cref{eq:LP}, which we
denote by $\tmop{LP} (A, B)$.

Let us note that when the~algebra~$R$ is commutative, then $\rho_r =
\lambda_r$ for each $r \in R$, and we do not need to distinguish the left and
the right representations of $R$. A very simple example of a~lifted product
code in this case is Kitaev's toric code~{\cite{Kitaev:2002}}, which can be
obtained as $\tmop{LP} (1 + x, 1 + y)$ with the~ring~$R =\mathbb{F}_2 [x, y] /
(x^L - 1, y^L - 1)$. Another important example is Haah's cubic
code~{\cite{Haah:2011}}, which is equal to $\tmop{LP} (1 + x + y + z, 1 + xy +
xz + yz)$, and $R =\mathbb{F}_2 [x, y, z] / (x^L - 1, y^L - 1, z^L - 1)$.
In~these two examples the parameter $L$ is the lattice size.
We see that in both these cases the ring~$R$ is a group algebra $\F_q G$ for some
finite group~$G$. Indeed, $G = \mathbf{C}_L^2$ for Kitaev's code, and $G =
\mathbf{C}_L^3$ for Haah's code, where $\mathbf{C}_L$~is the cyclic group of
order~$L$.

\begin{remark}
  Let us note that lifted products can also be used not only for group rings
  $R =\mathbb{F}_q G$. For example, if $R =\mathbb{F}_q [x] / (x^{\ell} -
  \alpha)$, where $\alpha \in \mathbb{F}_q^{\times}$, then any matrix $H \in
  R^{m \times n}$ defines the code $\mathcal{C} (H)$, which is called
  {\tmem{quasi-twisted}} code, or {\tmem{constacyclic}} if $m = n = 1$. \ Such
  codes~{\cite{Berlekamp:1968,Aydin:2001,Jia:2012}} sometimes have better
  parameters than quasi-cyclic and cyclic codes, which are their special cases
  when $\alpha = 1$. Thus it is an~interesting open problem whether lifted
  products of these classical codes can give quantum CSS codes with good
  parameters (cf. {\cite{Lv:2020}}).
\end{remark}

\section{Normed abelian groups}\label{sc:normed-group}

Let $\mathcal{M}$ be a~finite metric space with a~distance function $d (x,
y)$. For any non-empty subset $\mathcal{C} \subseteq \mathcal{M}$ we can
define its {\tmem{minimal distance}} $d (\mathcal{C})$ as
\begin{equation}
  d (\cC) := \min \{d (x, y) \mid x \ne y ; x, y \in \cC \},
  \label{eq:min-dist}
\end{equation}
where we assume that $d (\mathcal{C}) := \infty$ if $| \mathcal{C} | =
1$.

We can also define $d (x, \mathcal{Y})$ and $d (\mathcal{X}, \mathcal{Y})$ for
$x \in \mathcal{M}$ and $\mathcal{X}, \mathcal{Y} \subseteq \mathcal{M}$ in
a~straightforward way:
\begin{eqnarray}
  d (x, \mathcal{Y}) & := & \min_{y \in \mathcal{Y}} d (x, y), \\
  d (\mathcal{X}, \mathcal{Y}) & := & \min_{x \in \mathcal{X}, y \in
  \mathcal{Y}} d (x, y) .  \label{eq:dist-sets}
\end{eqnarray}
In what follows, we always assume that the~metric space~$\mathcal{M}$ is
an~{\tmem{abelian normed group}}, which means that it has an~abelian group
structure $(\mathcal{M}, +, \tmmathbf{0})$, and the~distance $d (\cdot,
\cdot)$ is {\tmem{invariant}}, i.e., $d (x + h, y + h) = d (x, y)$ for any $x,
y, h \in \mathcal{M}$. For example, if we have a~based vector space $\cM\cong \F_q^n$, then the standard Hamming distance $d(x,y):= \wt(x - y)$ is invariant.  It~is a~well-known and easily verified fact that
the~invariant distances $d (\cdot, \cdot)$ are in a~one-to-one correspondence
with the~functions $| \cdot | \colon \mathcal{M} \rightarrow \mathbb{R}_{\geqslant
0}$ called {\tmem{norms}} such that for all $x, y \in \mathcal{M}$ we have:
\begin{eqnarray}
  |x| = 0 & \Longleftrightarrow & x =\tmmathbf{0},  \label{eq:weight1}\\
  | - x| & = & |x|,  \label{eq:weight2}\\
  |x + y| & \leqslant & |x| + |y| ;  \label{eq:weight3}
\end{eqnarray}
where the~correspondence is given by $d (x, y) := | x - y |$ and $| x |
:= d (x, \tmmathbf{0})$. Such invariant distances on normed groups are
sometimes also called {\tmem{group norm metrics}}~{\cite{Deza:2013o}}. One can
easily check that if $\mathcal{C}$ is a~subgroup of~$\mathcal{M}$, then
the~minimal distance $d (\mathcal{C})$ can be also found by the~formula:
\begin{equation}
  d (\mathcal{C}) = \min_{x \in \mathcal{C}\backslash \{ \tmmathbf{0} \}} | x
  | . \label{eq:min-dist2}
\end{equation}

In fact, a~group norm metric on $\mathcal{M}$ also induces the~corresponding
metric on the~quotient group $\mathfrak{M}=\mathcal{M}/\mathcal{N}$ called
the~{\tmem{quotient norm metric}}~{\cite{Deza:2013o}}, where $\mathcal{N}$ is
some subgroup of $\mathcal{M}$. In this case, the~norm $| \mathcal{X} |$ for
$\mathcal{X} \in \mathfrak{M}$ is defined as
\begin{equation}
  |\mathcal{X}| := \min_{x \in \mathcal{X}}  | x | . \label{eq:weight}
\end{equation}
It is trivial to check that this norm satisfies
(\ref{eq:weight1})-(\ref{eq:weight3}), and the~corresponding distance
\[ d (\mathcal{X}, \mathcal{Y}) := | \mathcal{X}-\mathcal{Y} | \]
for $\mathcal{X}, \mathcal{Y} \in \mathfrak{M}$ is equivalent to the~distance
defined by (\ref{eq:dist-sets}). Thus, the~quotient group $\mathfrak{M}$ is
a~metric space, and for any group~$\mathcal{C}$ such that $\mathcal{N}
\subseteq \mathcal{C} \subseteq \mathcal{M}$ we can define the~minimal
distance of the~subgroup $\mathfrak{C}=\mathcal{C}/\mathcal{N} \subseteq
\mathfrak{M}$ as in (\ref{eq:min-dist}):
\[ d (\mathfrak{C}) := \min \{d (\mathcal{X}, \mathcal{Y}) \mid
   \mathcal{X} \ne \mathcal{Y}; \mathcal{X}, \mathcal{Y} \in \mathfrak{C}\} .
\]
In fact, using~(\ref{eq:min-dist2}) and~(\ref{eq:weight}) we can get a~much
simpler formula:
\begin{equation}
  d (\mathfrak{C}) = \min_{\tmscript{\begin{array}{c}
    \mathcal{X} \in \mathfrak{C}\backslash \{ \mathcal{N} \}
  \end{array}}}  | \mathcal{X} | = \min_{x \in
  \mathcal{C}\backslash\mathcal{N}}  | x | . \label{eq:quo-min-dist}
\end{equation}
Moreover, if $[\cdot] \colon \mathcal{M} \rightarrow \mathfrak{M}$ is a~canonical
projection, giving by $x \in \mathcal{M}$ its coset $[x] = x +\mathcal{N} \in
\mathfrak{M}$, then we get: $d ([x], \mathcal{Y}) = d (x, \mathcal{Y})$ and $|[x] | = d (x, \mathcal{N})$ for $x \in \mathcal{M}$ and $\mathcal{Y} \in
\mathfrak{M}$. This allows us to define for any subgroup $\cN\subseteq \cM$ a~new norm on $\cM$ that we call a~\emph{systolic norm}  as 
\[
|x|_\cN := |[x] | = d (x, \mathcal{N}).
\]

\newpage
\section{List of symbols and standard notations}\label{sc:symbols}

\begin{center}
\begin{tabular}{cp{0.45\textwidth}}
  $[n]$ & set $\{1,2,\dots,n\}$\\
  $\F_q$ & finite field with $q$ elements\\ 
  $R^{m\times n}$ & set of $m\times n$ matrices over $R$ \\
  $I_n$ & identity $n\times n$ matrix \\
  $\ker A$ & kernel of the linear map $v\mapsto A v$\\ 
  $\im A$ & image of the linear map $v\mapsto A v$\\ 
  $A^*$ & transpose map or transposed matrix for $A$\\
  $\cC^*$ & dual chain complex\\
  $\cF X$ & abelian group of formal sums $\sum_{x\in X} a_x x$ with \mbox{coefficients} $a_x\in\cF_x$ in a~local system $\cF$ \\
  $\wt(a)$ &  Hamming weight of $a\in\F_q^n$\\
  $\wt_S(a)$ &  block Hamming weight of $a\in\cF X$ relative to the subset~$S\subseteq X$\\
  $\abs{a}$ &  norm of $a\in A$ in a~normed abelian group $A$\\
  $\supp a$ & support $\{x\in X\mid a_x\ne 0\}$ for $a\in \cF X$\\
  $a|_S$ & restriction $\sum_{x\in S} a_x x$ to the subset $S\subseteq X$ of the formal sum $a=\sum_{x\in X} a_x x\in \cF X$ or a~vector $a\in \F_q^X$\\
  $\mathbb{K} G$ & group algebra over~$\mathbb{K}$ for the~group~$G$ \\
  $v \adj_e v'$ & $e$ connects $v$ and $v'$\\
  $G$-lift & $\abs{G}$-fold regular cover\\
  $A(\Gamma)$ &  adjacency matrix of $\Gamma$\\
  $\Gamma^2$ & square of the graph $\Gamma$, i.e., $A(\Gamma^2) = (A(\Gamma))^2$\\
  $E_\Gamma(S,T)$ &  set of oriented edges from $S$ to $T$ in~$\Gamma$\\
  $x\succ_P y$ & $x$ covers $y$ in a~poset~$P$ \\
  $\bar X^{p,q}$ &  double-cover of the~Ramanujan graph~$X^{p,q}$  \\
  $\cA \otimes_G \cB$ & $G$-lifted product of complexes $\cA$ and $\cB$\\
  $X \times_G Y$ & $G$-lifted product of posets $X$ and $Y$\\
  $[x : x']$ & incidence number for $x\in X(i)$, $x'\in X(i-1)$\\
  $\fT(\Gamma;h)$ & Tanner codes on $\Gamma$ with local code $\ker h$ \\
  $\fT_G(\hat{\Gamma};h)$ & $G$-lifted Tanner codes from $\fT(\hG;h)$\\
  $A\sim B$ & permutation equivalent codes or matrices\\
  $Z_i(\cC)$, $B_i(\cC)$ & spaces of $i$-cycles and $i$-boundaries for $\cC$\\
  $H_i(\cC)$ & $i$-th homology group of $\cC$\\
  $\dLM^{(i)}(\cC)$ & $i$-th locally minimal distance of $\cC$  \\
  $\partial_{S\to T}$ & restriction $\partial_{S\to T}\colon \cF S \to \cF T$ of a~boundary map $\partial\colon\cF X\to\cF X$ from $\cC_\bullet(X;\cF)$\\
\end{tabular}
\end{center}

\end{document}